\documentclass[format=acmsmall, review=false]{acmart}
\usepackage{acm-ec-26}

\setcitestyle{acmnumeric}

\renewcommand{\P}{\mathsf{P}}
\newcommand{\NP}{\mathsf{NP}}
\newcommand{\RP}{\mathsf{RP}}

\usepackage{dsfont}
\newcommand{\bone}{\mathds{1}}
\renewcommand{\phi}{\varphi}
\renewcommand{\vartheta}{\theta}
\newcommand{\dotcup}{\mathbin{\dot{\cup}}}

\usepackage{packages}
\usepackage{notation}
\usepackage{theorems}

\title{Maximally Random Sortition}

\author{Gabriel de Azevedo}
\affiliation{%
  \institution{Cornell University}
  \city{Ithaca}
  \state{NY}
  \country{USA}
}

\author{Paul G\"olz}
\affiliation{%
  \institution{Cornell University}
  \city{Ithaca}
  \state{NY}
  \country{USA}
}

\begin{document}


\begin{abstract}
Citizens' assemblies are a form of democratic innovation in which a randomly selected panel of constituents deliberates on questions of public interest. We study a novel goal for the selection of panel members: maximizing the entropy of the distribution over possible panels.

We design algorithms that sample from maximum-entropy distributions, potentially subject to constraints on the individual selection probabilities.
We investigate the properties of these algorithms theoretically, including in terms of their resistance to manipulation and transparency.
We benchmark our algorithms on a large set of real assembly lotteries in terms of their intersectional diversity and the probability of satisfying unseen representation constraints, and we obtain favorable results on both measures. We deploy one of our algorithms on a website for citizens’ assembly practitioners.
\end{abstract}
\maketitle

\setcounter{tocdepth}{1} %
\tableofcontents

\clearpage

\section{Introduction}
In a time of democratic backsliding~\cite{VDem25}, many countries are experimenting with \emph{citizens' assemblies} to complement and reinvigorate their democracies~\cite{OECD20,CoE2023}.
These assemblies gather a random sample of constituents to propose policy on a given topic.
Recent examples include assemblies in Ireland on same-sex marriage and abortion~\cite{Courant21}, which led to constitutional change, the French citizens' convention on climate change~\cite{GAA+22}, and permanently instituted assemblies in Brussels and Paris~\cite{OECD21}.

How exactly to do the random selection of panel members, called \emph{sortition}, is not obvious.
While practitioners and political theorists praise the virtues of random selection~\cite{Engelstad89,CM99,Dowlen08,MASSLBP17}, they consider an idealized sortition process in which the panel members are directly drawn (uniformly and without replacement) from the population, which is unrealistic since many people will not agree to participate. %
In practice, sortition proceeds in two stages: a large random sample of constituents is invited, and those willing to participate opt into a \emph{pool} of volunteers; then, a \emph{selection algorithm} randomly selects the panel from among the pool.
To be representative, this panel must satisfy \emph{quotas} imposed by the assembly organizers\,---\,say that, between 20 and 23 out of the 100 participants should have a college degree, exactly 50 should be female, and so forth.

Hence, the question is: \emph{Which probability distribution over the possible panels should the selection algorithm implement, to approximate idealized sortition?}
\citet{FGG+21} gave a first answer in \emph{Nature}:
since one desirable property of idealized sortition is that all constituents are equally likely to be selected~\cite{Engelstad89,CM99,Fishkin09}, their algorithmic framework implements a lottery such that the pool members' selection probabilities are as close to equal as possible. They refer to this objective as \emph{fairness}.
Algorithms within this framework\footnote{These algorithms are \textsc{MaxiMin}, its refinement \textsc{LexiMin}~\cite{FGG+21}, and \textsc{Goldilocks}~\cite{BF24}, which differ in their optimized measure of closeness to equality.} are now widely used in practice, through an open source tool\footnote{\url{https://github.com/sortitionfoundation/stratification-app/}} and \href{https://panelot.org}{\emph{panelot.org}}.

Though equalizing individual selection probabilities has clear appeal, the \emph{column generation} technique used by their algorithms comes with several drawbacks.
In particular, the resulting distributions' support consists of only few panels, which can cause extreme correlations in the selection of several pool members.
A wider problem is that many distributions over panels induce the same individual selection probabilities, leaving the choice of distributions to implementation details.
These limitations present an opening for new algorithms with complementary strengths.

In this paper, we design selection algorithms attaining a different objective: maximizing the randomness, specifically the \emph{entropy}~\cite{Shannon48}, of the probability distribution.
If the only constraints are the practitioners' quotas, this means choosing the panel uniformly from the set of all possible panels.
If organizers prescribe the pool members' selection probabilities to ensure fairness as above, our algorithms sample from the maximum-entropy distribution subject to this constraint.

Sampling from maximum-entropy distributions is attractive for several reasons:
\begin{itemize}
\item In a straight-forward way, the maximum entropy distribution makes the sortition \textbf{as random as possible} given the constraints.
It is well known that the maximum entropy distribution minimizes the Kullback--Leibler divergence to the uniform distribution over all subsets of the desired panel size~\cite[][Eq. (2.94)]{CT06}, so our distribution is \textbf{as close to a uniform selection} without replacement from the pool as possible given the constraints.
\item Our approach follows the \textbf{principle of maximum entropy} of Jaynes~\citep{Jaynes57,Jaynes57a}, who argues that, among all distributions satisfying known constraints, the maximum-entropy distribution is the only one that avoids introducing “arbitrary assumption of information”~\citep{Jaynes57}.
He aims to estimate the expected value of a function on the random variable, which captures the central question of how much representation in the panel a group not protected by quotas should have (on average).
Jaynes'~\citep{Jaynes57} reassures us that, by drawing from the maximum-entropy distribution, we represent all such groups according to the best guess given our information.
\item We found our choice of distribution \textbf{simple to explain} to practitioners and politicians since it can be defined by the following rejection procedure: uniformly draw the desired number of pool members,\footnote{If there are constraints on selection probabilities, each pool member has an adjusted probability of being selected in this stage, as we describe in \cref{sec:fme}. The rest of the process is unchanged.} check if it happens to satisfy all quotas, and repeat this until a valid panel is found.
Though this process is hopelessly slow, and our selection algorithm quite non-trivial, our algorithm only implements a speed-up of this common-sense, intuitively fair procedure.
\item \citet{Dowlen08} and \citet{Stone11}, influential philosophers of sortition, value randomness as a tool to decide between equally acceptable options \textbf{without any influence by illicit reasons}, which is what we do when choosing uniformly between panels satisfying the quotas.
In contrast to the goal of equalizing selection probabilities, maximizing entropy defines a unique probability distribution, which means that the distribution entirely emerges from the input (the quotas, which practitioners already must justify).
This leaves \textbf{no room for discretion by the organizers or algorithm} that could raise a ``suspicion of bias''~\cite{CM99}.
A maximum-entropy distribution also makes the result of the sortition as unpredictable as possible, which is seen as an obstacle to undue outside influence such as bribery~\cite{Stone11}.
\end{itemize}

\subsection{Our Techniques and Results}
Sampling from a maximum-entropy distribution over valid panels is clearly a hard algorithmic problem because even deciding if a valid panel exists is $\NP$-complete~\cite{FGG+21}.
What allows us to nevertheless solve the problem is a heavily optimized code implementation and a sequence of algorithmic optimizations exploiting the structure of practical sortition problems.

At the heart of our algorithm is a dynamic program (DP) for counting valid panels, similar to Papadimitriou's~\citep{Papadimitriou81} pseudo-polynomial algorithm for integer linear programming with a constant number of constraints.
For a typical medium-sized assembly, there are 20 to 30 quotas, which is limited but hopeless to include in this DP given its exponential running-time dependency.
By exploiting symmetries between pool members, going through pool members and quotas in a smart order, and aggressively pruning states, we scale the DP to as many quotas as possible.
Though this rarely scales to all quotas, it allows us to uniformly sample from the candidate panels that satisfy the incorporated quotas and to enforce the remaining quotas with rejection sampling.
The resulting algorithm, \textsc{MaxEntropy}, can uniformly sample among all feasible panels for 78 out of 86 real-world sortition problems, leaving out only instances with unusually many quotas.

To ensure fair selection probabilities, our algorithm \textsc{FairMaxEntropy} can take in target selection probabilities, which might be computed by column generation.
We formulate finding the maximum-entropy probability distribution as a convex problem with a constraint on the selection probabilities %
and consider its dual, which has fewer variables.
Each solution of the dual corresponds to an exponential-family distribution over panels, from which we can sample by using a variant of the DP above as an oracle.
By estimating the pool members' selection probabilities through sampling, we obtain an estimator for the gradient of the dual objective and optimize it using stochastic gradient descent.
We prove that \textsc{FairMaxEntropy} always produces a maximum-entropy distribution, whose selection probabilities converge to the targets at a rate of $O(1/\sqrt{T})$ in the number of iterations $T$. %

Next, we study the properties of our selection algorithms, theoretically (\cref{sec:properties}) and empirically (\cref{sec:empirics}).
From a theoretical perspective, \textsc{MaxEntropy} and \textsc{FairMaxEntropy} coincide with stratified sampling if the quotas are over disjoint strata and satisfy a notion of equal treatment of equals not satisfied by column generation.
We prove that \textsc{MaxEntropy} satisfies asymptotically optimal bounds on resistance to misrepresentation of features as defined by \citet{FlaniganLPW24} and how our algorithms can be implemented as a physical lottery to transparently demonstrate its selection probabilities~\cite{FlaniganKP21}.

Empirically, we evaluate our algorithms on a dataset of real-world sortition problems that is several times larger than the datasets of earlier papers, and compare them to existing selection algorithms.
In terms of fairness, pure \textsc{MaxEntropy} lies between the fairness of \textsc{Legacy} (an algorithm developed by the Sortition Foundation) and the column generation approaches.
Already few iterations of gradient descent seem to clearly increase fairness.
We find that the panels produced by our algorithms tend to have higher intersectional diversity than those produced by column generation.
If we hold out one category of quotas from the quotas (say, age), our algorithms are more likely to satisfy the held-out quotas by chance than column generation algorithms, primarily because column generation satisfies a fair fraction of features with zero probability.

Finally, we describe in \cref{sec:panelot} how we are deploying \textsc{MaxEntropy} on \href{https://panelot.org}{\emph{panelot.org}}, to facilitate its adoption in real-world assemblies.

\subsection{Related Work}

We already touched on several works on sortition algorithms~\cite{FGG+21,BF24,FlaniganKP21,FlaniganLPW24}. %
By providing a practical algorithm for maximum-entropy sortition, our paper resolves a question left open by \citet{FGG+21}, who also discuss potential strengths of maximum-entropy sortition in resisting outside influence.
Similarly, \citet{flanigan2023minipublic} call for ``selection algorithms for goals beyond individual fairness''.
Other works on sortition algorithms aim to equalize the probabilities of \emph{constituents} to be selected rather than those of pool members~\cite{FGG+20} and study how to replace exiting panel members~\cite{ABF+25}. %
\citet{DAL+21} study the recruitment of assembly members by random dialing (as done in French assemblies), which makes selection an online problem, and \citet{GMS+25} study how to compose the pool if population registers are located in municipalities (as in German assemblies).

Two recently proposed selection algorithms are \emph{Diversimax} by \citet{matar2025diversimax} and \emph{simulated annealing} by \citet{GMK+25}.
Diversimax only selects a deterministic panel, maximizing some notion of diversity by mixed integer linear programming, which is why it cannot be compared in terms of its random properties.
The simulated annealing algorithm, on the other hand, does not respect quotas, which is why we can also not compare it with other algorithms on equal terms.
\citet{GMK+25} stress as a key advantage over \textsc{LexiMin} that, among 10,000 sampled panels, almost all are distinct; \textsc{MaxEntropy} provides an even wider range of panels while additionally satisfying quotas.

Finally, several papers study sortition when assembly members can be directly picked from the population, and how well the assembly represents the underlying population.
This setting is studied in terms of representation of features without quotas~\cite{BGP19}, the welfare of majority votes performed by the assembly~\cite{MST21}, and representativeness measures in a metric space~\cite{EKM+22,EM25,CMP24}.

On a technical level, our work is related to papers on sampling from combinatorial structures based on uniform and general maximum-entropy distributions~\cite{AGM+17,AS10,JerrumSV04}.
Our paper uses a well-known reduction from sampling and counting, made formal by \citet{JerrumVV86}.
Specifically for maximum-entropy distributions over combinatorial sets, \citet{SinghV14} establish the equivalence between counting and optimizing via the ellipsoid method.
Other approaches for counting combinatorial structures include Markov Chain Monte Carlo~\cite{JerrumSV04,KV97,MS04} and generating functions~\cite{Barvinok94,DeLoeraHTY04}.

\section{Preliminaries}
\noindent
\textbf{Sortition.}
An instance of the sortition problem consists of a set of \emph{pool members} $N = [n] = \{1,\dots, n\}$, a set of features $F$, the values of these features $V$, the features' lower and upper quotas $\ell, u$, and a desired panel size $k$.
Each pool member is characterized by their value for all the features $f \in F$. For instance, one feature $f \in F$ might be \emph{gender}, and its possible values $V_f$ be $\{\text{female}, \text{male}, \text{other}\}$.
We write each feature as a function $f : N \to V_f$, so that $f(i)$ gives pool member $i$'s gender.
For each \emph{feature-value pair} $f,v \in \FV \coloneqq \{(f, v) : f \in F, v \in V_f\}$, we are given integer quotas $0 \leq  \ell_{f,v} \leq u_{f,v} \leq k$ specifying the minimum and maximum allowed number of selected panel members with feature $f$ taking value $v$.

At times, it will be convenient to consider pool members' attributes in the shape of a characteristic vector.
Let $w : N \to \{0,1\}^{\FV}$ be this mapping from pool members to \emph{feature-value vectors}, so that $w(i)_{f,v} = \Iverson{f(i) = v}$ for all $f, v \in \FV$.
We denote the set of distinct feature-value vectors in the pool by $\calW$.
For any set of feature-value vectors $W \subseteq \curly{0, 1}^\FV$, we denote the subset of pool members with features in $W$ by $\restr{N}{W} \coloneqq \setst*{i \in N}{w(i) \in W}$.

A \emph{panel} is a set of $k$ pool members such that, for every feature-value pair $f, v \in \FV$, the number of pool members with feature $f$ taking value $v$ is between $\ell_{f, v}$ and $u_{f, v}$.
We denote the set of all panels by $\calP$, and the set of panels containing some pool member $i \in N$ by $\calP(i)$.
A \emph{selection algorithm} takes in the instance, and randomly selects a panel. %

Given a \emph{panel distribution} $\lambda \in \Delta(\calP)$, pool member $i$'s \emph{selection probability} is $\pi_\lambda(i) \coloneqq \mathbb{P}_{P \sim \lambda}[i \in P] = \sum_{P \in \calP(i)} \lambda(P)$. 
When $\lambda$ is clear from the context, we drop it from the subscript.
Note that sum of selection probabilities $\sum_{i \in N} \pi_\lambda(i) = k$ is constant for all panel distributions.

A \emph{fairness measure} is an objective intended to quantify how evenly the selection probability is distributed across pool members.
Formally, a fairness measure is a concave function $g: [0, 1]^N \to \R \cup \{-\infty\}$ that maps a vector of selection probabilities to a \emph{fairness score}.
A panel distribution maximizes a fairness measure if its induced selection probabilities maximize the fairness measure.\medskip

\noindent
\textbf{Convex analysis.}
We review some standard notions from convex analysis, which we will use for the development of \textsc{FairMaxEntropy} in \cref{sec:fme} see \citet{Rockafellar70} for a comprehensive treatment.
Given a set $U \subseteq \R^d$, its \emph{affine hull} $\aff(U)$ is the set of all affine combinations of elements of $U$. 
For a $x \in \R^d $, we define the $\ell_\infty$ \emph{ball} as $B_\infty(x, \eta) \coloneqq \setst{y \in \R^d}{\norm{x - y}_\infty \leq \eta}$.
The $\eta$-\emph{relative interior} of a set $U$ is defined as $\ri_\eta(U) \coloneqq \setst{x \in U}{ B_\infty(x, \eta) \cap \aff(U) \subseteq U}$.

Given a distribution $\gamma \in \Delta(S)$ over a finite set $S$,
we define its \emph{(Shannon) entropy} as $H(\gamma) \coloneqq -\sum_{s \in S} \gamma(s)\log \gamma(s)$.
Since entropy is strictly concave, it has a unique maximizer over any convex compact set of distributions, so the maximum-entropy distribution is uniquely defined, with or without additional constraints on selection probabilities.

\section{\textsc{MaxEntropy}: Sampling Uniformly among Possible Panels}
\label{sec:me}

When the only restrictions on the sortition are the quotas, the maximum-entropy distribution coincides with the uniform distribution.
In this section, we describe how our algorithm \textsc{MaxEntropy} samples from this distribution.
We defer all proofs to \cref{apx:me}, where we directly describe a generalization of this algorithm.
This generalization samples panels with different probabilities based on some pool member-specific weights and will be useful for our development of fair maximum-entropy algorithms in \cref{sec:fme}.

To sample uniformly from the feasible panels, \textsc{MaxEntropy} needs to solve a problem on practical instances that is in general intractable.
Indeed, sampling panels is clearly at least as hard as finding one, and deciding whether a panel exists is $\NP$-hard, even for only three features by reduction from exact 3-cover~\cite{FGG+21}.\footnote{Even with the promise that a panel exists, no algorithm can find one in polynomial time unless $\P=\NP$, and no polynomial-time algorithm can find one with constant probability unless $\RP=\NP$~\cite{FGG+21}.}
Our quota constraints are so expressive that \textsc{MaxEntropy} sortition is equivalent to uniformly sampling from the solutions of a quite general class of 0--1 integer linear programs, which underscores the difficulty of the problem.
General tools for this problem, like \emph{LattE}~\cite{DeLoeraHTY04} scale badly in the number of variables (i.e., pool members).
The one piece of good news from complexity theory is that the number of quotas tends to be somewhat bounded in practice, and that 0--1 integer linear programs with a constant number of constraints can be solved in polynomial-time~\cite{Papadimitriou81}.
Indeed, a similar dynamic program will be the starting point of our algorithm, its running time on even a smaller instance would be at least on the order of $k^{|\FV|} \approx 30^{20} \approx 3 \cdot10^{29}$, so hopelessly inefficient.

\subsection{Counting Panels by Dynamic Programming}
\label{sec:cp}
It is well known that uniform sampling over combinatorial structures can be reduced to counting the number of solutions~\cite{NW78,JerrumVV86}.
We will begin by formulating a na\"ive dynamic program for counting the number of panels, and then describe three optimizations that exploit the structure of real-world sortition problems to speed up the computation.
In addition to an optimized C++ implementation with careful memory management (see \cref{sec:implementation}), we need all these optimizations, in addition to the later rejection sampling stage, to be able to sample from highly constrained real-world sortition instances.

For a set of pool members $U\subseteq N$, define its \emph{(feature-value) profile} as $w(U) \coloneqq \sum_{i\in U} w(i)$.
We denote the set of \emph{quota-compliant profiles} by
\begin{align*}
    \overline{\calZ}\coloneqq \setst*{z \in \N^{\FV}}{\textstyle \ell_{f,v}\leq z_{f,v}\leq u_{f,v}\ \text{for all }f, v\in\FV \text{~and~}
    \sum_{v\in V_{f_0}} z_{f_0,v}=k \text{ for some fixed }f_0\in F },
\end{align*}
where the latter constraint uses the fact that $\sum_{v \in V_f} w(i)_{f,v} = 1$ for any pool member $i$ and feature $f$ to count the number of pool members.
The set $U$ constitutes a panel if and only if $w(U)\in\overline{\calZ}$, though some profiles in $\overline{\calZ}$ may be unrealizable.
Our dynamic program calculates the following \emph{counting function}:
\begin{align*}
    \phi : [n] \times \N^{\FV} \to \N, \quad (i, z) \mapsto \card{ \setst{U \subseteq \curly{i, \dots, n}}{ z + w(U) \in \overline{\calZ}}},
\end{align*}
which counts the number of ways to complete a partial profile $z$ to a quota-compliant profile using only a suffix of pool members $\{i, \dots, n\}$.
In particular, $\phi(1, \vec{0}) = \card{\calP}$ computes all possible panels. 
We may compute the counting function according to the following recursion, based on the simple idea that the ways for completing $z$ with the pool members $\{i, \dots, n\}$ can be partitioned into those that choose $i$ and those that do not.

\begin{proposition}
\label{prp:unif_dp}
For every $i \in N$ and $z \in \N^{\FV}$, we have that
\begin{align*}
    \phi(i, z) = \phi\big(i + 1, z + w(i)\big) + \phi(i + 1, z),
\end{align*}
with the base case $\phi(n + 1, z) = \Iverson{z \in \overline{\calZ}}$.
\end{proposition}

Since the number of profiles that can possibly be completed to be quota-compliant is at most $(k + 1)^{\card{\FV}}$ and since we do a constant amount of operations per state, we may bound the running time by $O\big(n \, (k+1)^{\card{\FV}}\big)$, which as discussed above is hopelessly large. \medskip

\noindent
\textbf{Processing identical pool members simultaneously.}
When several pool members have the same feature-value vector and we place them next to each other in the ordering of $N$, we can perform the recurrence for these pool members in a single step, in which we decide how many of them to choose.
Specifically, if pool members $i, i+1, \dots, i+n_{\omega}-1$ have the feature-value vector $\omega = w(i) = \cdots = w(i + n_\omega-1)$, where $n_\omega$ is the number of pool members with feature-value vector $\omega$, we can write
\[ \phi(i,z) = \sum_{d=0}^{n_\omega} \binom{n_\omega}{d} \cdot \phi(i+ n_\omega, z + d \cdot \omega) \quad \text{for each $z \in \mathbb{N}^\FV$.} \]
While this does not improve the theoretical complexity, and instances with many quotas have few entirely identical pool members, it will pay off in combination with a later optimization. \medskip

\noindent
\textbf{Decreasing necessary state space for one feature.}
To further reduce the running time, we sort the pool members by their value for some chosen feature $f^*$.
The benefit of this approach is that, in the recurrence step for the last pool member with some value $v \in V_{f^*}$, we can immediately decide whether the quotas for $f^*, v$ are violated (so the count is zero) or satisfied (so this quota will be satisfied for all completions).
As a result, the dynamic program only needs to keep track of the profile for the current value of $f^*$ rather than all of them, which reduces the running time to $O\big(n \, (k+1)^{|\FV| - |V_{f^*}|+1}\big)$.

This optimization is highly effective because many sortition problems include a feature for geographical location, which can have a large number of values.
If there are, say, 10 regions in the country or 10 neighborhoods in a city, processing the pool members by geography reduces the running times by a considerable factor of $(k+1)^9$.\medskip

\noindent
\textbf{Pruning states based on subsets of the features.}
In our experience, a large majority of states cannot be completed to a quota-compliant profile, even when only a portion of the quotas is considered.
One reason for this is that practitioners often set the lower and upper quotas with a narrow gap between them (often around the census percentages~\cite{OECD20,GillL22}).
Moreover, some feature-values, especially those of populations that opt into the pool at lower rates, cannot vary entirely independently on each other.

To avoid spending large amounts of memory and computation on states that cannot be completed, we iteratively compute versions of the dynamic program that include more and more features.
Fix a subset of features $F' \subseteq F$, %
set $\restr{z}{F'}$ for the restriction of a profile $z$ to the features in $F'$, and define the counting function $\phi'$ for only the features $F'$ as above.
The key insight is that the counting function for $F'$ gives an upper bound on the full counting function:
\begin{proposition}
    \label{prp:phi_heuristic}
    For all $i \in N$, $z \in \N^{\FV}$, and $F' \subseteq F$, we have that
    $\phi(i, z) \leq \phi'(i, \restr{z}{F'})$.
\end{proposition}
Hence, we can prune all states $(i, z)$ with $\phi'(i, \restr{z}{F'}) = 0$.

We found this pruning procedure to be a large improvement over direct pruning steps that do not take into account the interdependency of features.
Our approach of iteratively increasing the number of features also makes the first optimization pay off to a much greater degree because much more pool members appear identical when only considering a subset of the features.

This approach also raises the question of in which order features should be added to the dynamic program.
We use an intelligent heuristic which, in each iteration, prioritizes features that are rarely satisfied among a large number of panels sampled from the previous iteration's dynamic program.
The goal of this heuristic is to rule out a large number of dynamic programming states, and to prioritize the features that cannot be addressed well through rejection sampling as described below.

\subsection{Uniformly Sampling a Panel}
After having computed the dynamic program for any subset of the features, we can use it to uniformly sample from all sets of $k$ pool members that satisfy the incorporated features.
For exposition, we describe this process for the na\"ive dynamic program without the optimizations.

To sample using the dynamic program, we follow the recurrence of the dynamic program, building our panel person by person.
We start at the state $(1, \vec{0})$ and with an empty panel.
At each state $(i, z)$, we add pool member $i$ to the panel with probability
\[ \frac{\phi(i+1, z + w(i))}{\phi(i, z)} = \frac{\phi(i+1, z + w(i))}{\phi(i+1, z + w(i)) + \phi(i+1, z)},\]
which gives the ratio of completions that include pool member $i$.
If $i$ is included in the panel, the algorithm recurses to state $(i + 1,  z + w(i))$, otherwise, to state $(i + 1,  z)$.
The algorithm stops when it reaches a base state, at which point the panel is guaranteed to be feasible by \cref{prp:unif_dp}.
\medskip

\noindent
\textbf{Enforcing remaining features with rejection sampling.}
If we can compute the dynamic program for all features, the above is sufficient for our goal of uniform sampling.
Most of the time, however, this is only feasible for a subset of the features.

As we iteratively build the dynamic program for larger subsets $F' \subseteq F$ of features, we keep sampling from each dynamic program, uniformly among all pool member sets of size $k$ that satisfy the quotas for $F'$.
This sampling helps us choose the next feature to add, but it can also reach the point where a sample satisfies all remaining quotas.
In this case, this panel is a uniform sample from the set of all panels, and can be returned:
\begin{theorem}
    \label{thm:me_alg}
    Given an instance of a sortition problem, \textsc{MaxEntropy} samples a panel uniformly from the set of all panels. 
\end{theorem}

\section{\textsc{FairMaxEntropy}: Sampling with Prescribed Selection Probabilities}
\label{sec:fme}
We see the \textsc{MaxEntropy} algorithm as a serious contender for use in practical citizens’ assemblies, in particular due to its ease of explanation.
But many practitioners use algorithms that optimize the fairness of selection probabilities and may be reluctant to give up their optimal fairness properties.

With these practitioners in mind, we develop a selection algorithm called \textsc{FairMaxEntropy}, which targets selection probabilities prescribed by the practitioners and draws from the maximum-entropy distribution subject to this constraint.
If practitioners derive these marginals from the output of \textsc{LexiMin} or \textsc{Goldilocks}, they retain the optimal fairness guarantees provided by these algorithms, without suffering from the limitations of column generation (small support, lack of unique specification) we describe in the introduction. \medskip

\noindent
\textbf{Optimization formulation and structure.}
To find these maximum-entropy distributions, we rely on convex optimization.
Fix a sortition instance, and let $\pi$ denote the set of target selection probabilities.
The goal of \textsc{FairMaxEntropy} is to sample panels from the distribution $\lambda^*_\pi$ with maximum entropy satisfying this constraint.
To achieve this, we compute a solution to the following pair of primal and dual entropy-regularized convex programs (see \cite{BoydV04}):

\begin{center}
\vspace{0.1cm}
\begin{minipage}[b]{0.4\textwidth}
\begin{equation}
  \begin{aligned}
    \text{max~} & H(\lambda) \\
    \text{s.t.~} & \pi_\lambda = \pi \\
     & \lambda \in \Delta(\calP).
  \end{aligned}
  \tag{P\ensuremath{(\pi)}}\label{frm:p_max_entropy}
\end{equation}
\end{minipage}
\hfill
\begin{minipage}[b]{0.56\textwidth}
\begin{equation}
  \begin{aligned}
    \text{min~} & \LSE(A^\transp\theta) -\iprod{\theta}{\pi} \\
    \text{s.t.~} & \vartheta \in \R^N. \\
    ~
  \end{aligned}
  \tag{D\ensuremath{(\pi)}}\label{frm:d_max_entropy}
\end{equation}
\end{minipage}
\vspace{0.1cm}
\end{center}
In the dual,
$A$ is the \emph{incidence matrix} of $\calP$, that is, 
$A\in \{0, 1\}^{N\times \calP}$ where 
$A_{i,P} = \Iverson{i \in P}$ for all $i \in N$ and $P \in \calP$,
and $\LSE$ denotes the \emph{log-sum-exp} function
\begin{align*}
    \LSE(A^\transp \vartheta) = 
    \log\paren*{\sum_{P \in \calP} \exp\paren*{\sum_{i \in P} \vartheta(i)}}.
\end{align*}
Whereas the primal program has an exponentially large number of variables, the dual has only $n$ and might therefore be more amenable to optimization.
We begin by relating the primal and dual optimal solutions of the convex program.
\begin{proposition}
    \label{prp:gibbs}
    A panel distribution $\lambda$ such that $\pi_{\lambda} = \pi$ and a vector $\vartheta \in \R^N$ are an optimal primal--dual pair iff
    \begin{equation}
        \lambda(P) 
        = \frac{ \exp\paren*{ \sum_{i \in P} \vartheta(i)}}
        {\sum_{Q \in \calP} \exp\paren*{\sum_{i \in Q} \vartheta(i)}} 
        \quad \text{ for all } P \in \calP. \label{eq:expfamily}
    \end{equation}
    Whenever $\pi$ is contained in the relative interior of the vector of selection probabilities.
\end{proposition}

This proposition shows that the maximum-entropy distribution $\lambda_\pi^*$ has a compact representation in exponential-family form when parameterized by the optimal dual variables (see \cite{BoydV04,SinghV14} for a proof of the statement).
It also shows that, for any $\theta \in \R^N$, the exponential-family distribution $\lambda$ defined by \cref{eq:expfamily} is an optimal primal solution for the primal (\hyperref[frm:p_max_entropy]{P$(\pi_{\lambda})$}), i.e., maximizes entropy among all distributions with the same selection probability vector.
It will be natural to write these exponential-family distributions in a multiplicative form, as
\begin{align*}
    \lambda_\mu(P) \coloneqq \frac{\prod_{i \in P}\mu(i)}
    {\sum_{Q \in \calP} \prod_{i \in Q} \mu(i)} 
    \quad \text{ for all } P \in \calP
\end{align*}
for given weights $\mu \in \R_{> 0}^N$, where $\mu(i)$ plays the role of $e^{\theta(i)}$.

This characterization shows that any exponential-family distribution can still be described by a variant of the rejection sampling procedure for uniform sampling we described in the introduction.
Define a probability distribution over the pool members, in which pool member $i$ has mass proportional to $\mu(i)$.
We draw $k$ pool members with replacement from this distribution; check if they are all distinct and happen to satisfy the quotas; and repeat this process until we find a valid panel.
Clearly, each panel $P$'s probability of being sampled in this process is proportional to $\prod_{i \in P} \mu(i)$, which shows that the procedure implements $\lambda_{\mu}$.
This means that, if assembly organizers can explain the different weights $\mu(i)$ for the pool members (say, giving higher weight to those appearing on few panels), the resulting probability distributions are still aligned with common sense. \medskip

\noindent
\textbf{Selection algorithm.} If we know $\mu \in \R_{> 0}^N$, a generalization of our algorithm in \cref{sec:me} can sample panels from its associated distribution $\lambda_\mu$.
To do so, our dynamic program should no longer simply count the number of panels, but compute a weighted sum $\sum_{P \in \calP} \prod_{i \in P} \mu(i)$, which is easy to change (see \cref{apx:gcp}).\footnote{Here, we ignore questions of numerical precision in the dynamic program, and treat it as an oracle. In our practical implementation, we round the $\mu(i)$ to be able to represent the values of the dynamic program with bounded integers, which entails some error.}
So all that remains is to find the optimal dual weights $\mu^*$ that give us the correct selection probabilities.

The dual \eqref{frm:d_max_entropy} has few variables, but the log-sum-exp term in its objective contains a sum ranging over all exponentially many panels.
What helps us optimize it is that the gradient of the dual objective at a point $\theta$ is simply
\[ \left(\frac{\sum_{P \in \calP(i)} \exp(\sum_{j \in P} \theta(j))}{\sum_{P \in \calP} \exp(\sum_{j \in P} \theta(j))} -\pi(i)\right)_{i \in N} = \left(\sum_{P \in \calP(i)} \lambda_\mu(P) - \pi(i)\right)_{i \in N} = \pi_{\lambda_\mu} - \pi \]
for the weights $\mu(i) = e^{\theta(i)}$.

It is a remarkable coincidence that computing the gradients reduces to the one thing we know how to do\,---\,sampling many panels according to the exponential-family distribution $\lambda_\mu$ and using them to estimate the pool members' selection probabilities.
Specifically, the empirical frequency $\overline{\pi}(i)$ of drawn panels that include pool member $i$ is an unbiased estimator for $\pi_{\lambda_\mu}(i)$.
By subtracting the known vector $\pi$ from $\overline{\pi}$, we obtain an unbiased estimator for the gradient and can optimize the dual using stochastic gradient descent.\medskip

\noindent
\textbf{Convergence bounds.} Let $\calM \coloneqq \setst{\pi_{\lambda}}{\lambda \in \Delta(\calP)}$  be the \emph{polytope of selection probabilities}.
Assume for now that the vector of target selection probabilities $\pi$ lies in $\ri_\eta(\calM)$ for some $\eta > 0$ so that strong duality holds. 
(Further below, we discuss the necessity of this assumption and what can be done when it does not hold.) We can then give following convergence bound.

\begin{theorem}
    Given a sortition instance and target selection probabilities $\pi \in \ri(\calM)$, run \textsc{FairMaxEntropy} for $T$ iterations, and let $\lambda$ be the distribution obtained from the last iterate.
    Then, $\lambda$ is an optimal solution for \emph{(\hyperref[frm:p_max_entropy]{P$(\pi_{\lambda})$})} (i.e., the maximum-entropy distribution subject to its own selection probabilities) and $%
    \mathbb{E}\big[\norm{\pi_{\lambda} - \pi}_2^2\big] \leq O(1/\sqrt{T})$.
\end{theorem}
\begin{proof}
    We define \textsc{FairMaxEntropy} by minimizing (\hyperref[frm:d_max_entropy]{D$(\pi_{\lambda})$}), using stochastic gradient descent with the step-size schedule of \citet{JNN21}, using the dynamic program as an oracle for the gradient.
    (This step-size rule ensures that we can give guarantees for the last iterate, rather than standard bounds which hold on average over the iterates.)

    To apply their convergence result, we must show that the dual objective $d(\theta) \coloneqq \LSE{}(A^\transp \theta) - \langle \theta, \pi \rangle$ is convex, Lipschitz continuous, that there is a dual optimal solution with bounded norm, and that the estimated gradients have bounded norm.
    The dual objective $d$ is a convex because the log-sum-exp function is convex.
    Its gradient has bounded norm $\norm{\pi_{\lambda_\mu} - \pi}_2 \leq \norm{\pi_{\lambda_\mu}}_2 + \norm{\pi}_2 \leq 2\sqrt{k}$, and so does any estimated gradient $\norm{\overline{\pi} - \pi}_2 \leq 2 \sqrt{k}$.
    Finally, the $\eta$-interior condition implies that there is a dual optimum solution $\vartheta^*$ with $\norm{\vartheta^*}_2 \leq k\log(n)/\eta$~\cite[][Thm.\ 4.2]{SinghV14}. 
    Taking $\vartheta$ as the solution obtained at iteration $T$, Theorem 1 of~\citet{JNN21} implies that
    \begin{align*}
        \mathbb{E}\big[d(\theta) - d(\theta^*)\big] \leq O(1/\sqrt{T}).
    \end{align*}
    The dual function is $k$-smooth \cite[][Lem.\ 3.10]{BuergisserLNW20}, which
    together with convexity implies that the norm of the gradient decays proportionally to the suboptimality gap.
    Formally, for any $\theta' \in \R^N$, it holds that $\norm{\nabla d(\theta')}_2^2 \leq 2 \, k(d(\vartheta') - d(\theta^*))$. 
    Therefore, taking $\lambda \coloneqq \lambda_{e^\theta}$ as the distribution obtained from the last iterate, we conclude that
    \begin{equation*}
        \mathbb{E}\big[\norm{\pi_{\lambda} - \pi}_2^2\big] = \mathbb{E}\big[ \norm{\nabla d(\theta)}_2^2\big] \leq 2 \, k \, \mathbb{E}\big[d(\theta) - d(\theta^*)\big] \leq O(1/\sqrt{T}).\qedhere
    \end{equation*}
\end{proof}

This $O(1/\sqrt{T})$ convergence rate is the best we can hope for our optimization problem since the dual objective is not strongly convex~\cite{AgarwalWBR09}.
If $\pi$ indeed lies in the relative interior, the theorem shows that \textsc{FairMaxEntropy} can sample from a maximum-entropy distribution with selection probabilities that are arbitrarily close to the desired vector $\pi$.\medskip

\noindent
\textbf{Discussion of relative-interior assumption.}
Recall that our theorem assumes that the target marginals $\pi$ lie in the relative interior of $\calM$.
If, instead, $\pi$ lies on its boundary, then any distribution $\lambda^*$ must lie on a facet of $\Delta(\calP)$, because $\pi = \pi_{\lambda^*} = A \, \lambda^*$, and the linear function $A$ (i.e., the incidence matrix defined below the convex programs) maps the interior of $\Delta(\calP)$ to the relative interior of $\calM$.
This means that, if $\pi$ is not in the relative interior, some panels have zero mass in $\lambda^*$, which is not possible for any exponential-family distribution and can only be approximated by having some $\theta(i)$ become very negative.
As a result, $\theta$ diverges in the optimization.

Given target selection probabilities $\pi$ on the boundary of $\calM$, we can fix this issue by conducting an interiorization step.
Let $\unifD$ be the uniform distribution over panels, which clearly lies in the relative interior of $\Delta(\calP)$.
As we argued above, its marginals must lie in the relative interior of $\calM$.
Hence, one can introduce a small perturbation to $\pi$ in the direction of $\pi_{\unifD}$ to obtain a $\pi'$ that is in $\ri(\calM)$.
Since $\pi_{\unifD}$ is not known, we can only estimate it by sampling from \textsc{MaxEntropy}, which might find a point in the relative interior, but is not guaranteed to. \medskip

\noindent
\textbf{Approaching arbitrary target probabilities.} The good news is that \textsc{FairMaxEntropy} can actually approach any target vector $\pi$\,---\,even ones that are not realizable (i.e., not in $\calM$), in which case the selection probabilities of \textsc{FairMaxEntropy} approach the projection of $\pi$ on the feasible set $\calM$.
This follows from very recent results by \citet{HiraiS25}, who show that gradient descent converges to the minimizer of the gradient norm (in our case, $\norm{\pi_\lambda - \pi}_2$, where $\lambda$ is the primal distribution for $\theta$) in the limit, even when the function is unbounded.
Though no meaningful bounds are known for the speed of convergence of the gradient, this means that our algorithm is robust to marginals on the boundary of $\calM$, or to errors in the target probabilities $\pi$.

Since gradient descent optimizes $\norm{\pi_\lambda - \pi}_2$, one might be tempted to set $\pi = (\frac{k}{n}, \dots, \frac{k}{n})$, to let \textsc{FairMaxEntropy} find the fairest selection probabilities without prior computation.
This works, in a way, and maximizes the fairness measure defined as the negative $\ell_2$ norm of selection probabilities from the point $\pi$ of perfectly equal probabilities.
Unfortunately, this happens to be a rather poor fairness measure, which tends to select many pool members with zero probability.
This has been previously observed by~\cite[][App.\ D.7]{FlaniganLPW24} for a relaxation of the sortition problem. %

\section{Theoretical Properties}
\label{sec:properties}
In the introduction, we gave arguments that make maximizing entropy an attractive objective on a conceptual level.
In this section, we add desirable properties satisfied by \textsc{MaxEntropy} and \textsc{FairMaxEntropy} requiring more mathematical elaboration and discuss our algorithms according to properties of selection algorithms studied in earlier papers.

\subsection{Simplicity of Distributions}
\label{sec:simplicity}
If a citizens' assembly is to legitimately speak for a population, it is important that the sortition process choosing its members be publicly trusted~\cite{Gasiorowska23}.
This poses a challenge because all selection algorithms, by virtue of solving an $\NP$-hard problem, are somewhat complex and hard for a layperson to understand.
This difficulty, however, can be ameliorated by showing that selection algorithms work in a simple way in combinatorially straight-forward instances, or by stressing natural properties of the distributions.

One sampling algorithm that is broadly understood (and referenced in practitioners' explanations of sortition~\cite{MASSLBP17}) is
\emph{stratified sampling}.
Stratified sampling only applies to a subsetting of the sortition problem, in which there is a single feature and there are no gaps between lower and upper quotas.
For instance, one might stratify by age group, with the values “18--29”, “30--44”, “45--64”, and “65+” and respective targets 6, 8, 7, and 4, for a total of 25 members.
The set of pool members with a given value is called a \emph{stratum}.
To produce the panel, stratified sampling draws, from each stratum, a uniform subset of the targeted size and returns their union.

In this restricted subsetting, we show that \textsc{MaxEntropy} coincides with stratified sampling.
The fact that \textsc{MaxEntropy} generalizes a common-sense selection procedure to more general settings supports the notion that this algorithm itself implements a natural panel distribution.
By contrast, column generation algorithms do not coincide with stratified sampling, which can be seen from the fact that their support captures only a vanishing fraction of possible panels.

\begin{proposition}
    \textsc{MaxEntropy} is equivalent to stratified sampling in the case of a single feature and tight quotas.
\end{proposition}
\begin{proof}
    Consider a sortition instance with a single feature $f$, values $V_f = \{v_1, \dots, v_m\}$, and targets $k_{v_j} = \ell_{f,v_j} = u_{f,v_j}$ for $j \in [m]$.
    Set $N_j \coloneqq \setst{i \in N}{f(i)=v_j}$ for the stratum belonging to $v_j$.
    For this instance to be feasible, it must hold that $k = \sum_{j=1}^m k_j$.
    It suffices to observe that
    \[ \calP = \setst{P \subseteq N}{\card{P} = k\text{ and }\card{N_i \cap P} = k_i\text{ for all } i \in [m]} \]
    also describes the set of all panels that can be selected by stratified sampling, and that each of these is selected with an equal probability, namely
    \[ \prod_{j = 1}^m 1/{\textstyle \binom{|N_i|}{k_i}}. \qedhere \]
\end{proof}

One attractive property of stratified sampling is that it treats pool members with the same value perfectly symmetrically.
Not only do two such pool members have the same probability of being selected but they, for example, also have the same probability of being selected jointly with a third pool member.
We generalize this observation in the following axiom.
\begin{definition}
A panel distribution $\lambda$ satisfies \emph{higher-order equal treatment of equals} if, for any two pool members $i, j$ with identical feature-value vectors and, if applicable, equal target selection probabilities, the distribution $\lambda$ is invariant under swapping $i$ and $j$. That is, if $\tau_{i,j} P$ denotes the swapping of $i$ for $j$ and $j$ for $i$ in $P$, and if $\tau_{i,j}\lambda$ is the distribution defined by $(\tau_{i,j} \lambda)(P) \coloneqq \lambda(\tau_{i,j} P)$ for all $P \in \calP$, it holds that $\lambda = \tau_{i,j} \lambda$.
\end{definition}
Note that this notion is substantially stronger than the \emph{equal treatment of equals} defined by \citet[][App.\ 15.3]{FGG+21}, who only require $i$ and $j$ to have equal selection probabilities.

\begin{theorem}
    \label{thm:anonimity}
    Given a sortition instance, the probability distributions of \textsc{MaxEntropy} and \textsc{FairMaxEntropy} (when optimized to optimality) satisfy higher-order equal treatment of equals.
\end{theorem}
\begin{proof}
    For \textsc{MaxEntropy}, this is immediate since all panels have an equal probability (and swapping pool members with identical feature-value vectors preserves quota satisfaction).

    Now, let $\lambda^*$ denote the optimal panel distribution for \textsc{FairMaxEntropy}, i.e., the optimal solution to \eqref{frm:p_max_entropy}.
    Fix two pool members $i,j$ with equal feature-value vectors and target probabilities $\pi(i)=\pi(j)$, and consider $\tau_{i,j} \lambda^*$ as defined above.
    Clearly, $\pi_{\tau_{i, j} \lambda^*}(i) = \pi_{\lambda^*}(j) = \pi(j) = \pi(i)$, $\pi_{\tau_{i, j} \lambda^*}(j) = \pi_{\lambda^*}(i) = \pi(i) = \pi(j)$, and all other pool members' selection probabilities are also unchanged.
    Hence, $\tau_{i, j} \lambda^*$ is also a feasible solution to \eqref{frm:p_max_entropy}.
    Since we just permuted the probability masses of panels, $\tau_{i,j} \lambda^*$ must furthermore have the same, maximal entropy as $\lambda^*$.
    Since the entropy objective is strictly concave, the maximizer is unique, and hence $\lambda^* = \tau_{i,j} \lambda^*$.
\end{proof}

\subsection{Resistance to Manipulation}
\label{sec:manip}
When recruiting the members for an assembly, the pool members report their own feature-value vectors when opting into the pool.
This process might be vulnerable to misreporting by a group of malicious agents who seek to strategically change their own and others' selection probabilities.

Under the assumptions made by \citet{FlaniganLPW24}, who first studied this issue theoretically, \textsc{MaxEntropy} achieves optimal manipulability across three measures. \emph{Internal manipulability} measures the largest increase in selection probability that can be achieved for one of the members of the manipulating coalition; \emph{external manipulability} measures the largest decrease in selection probability that can be imposed on a non-coalition member; and \emph{composition manipulability} measures the largest amount of seats that a malicious group can, in expectation, misappropriate in favor of some group. 

\citet{FlaniganLPW24} showed that both \textsc{LexiMin} and \textsc{Nash} are arbitrarily manipulable; that is, as the pool grows, malicious agents can raise their selection probabilities by almost 1. They also show that minimizing the $\ell_\infty$ distance to the uniform vector of selection probabilities achieves optimal manipulability. There are two drawbacks here. First, minimizing the $\ell_\infty$ distance typically results in many pool members having zero selection probability, making it a poor choice of algorithm.
Second, their analysis is in a continuous relaxation, meaning that their algorithms only satisfy ex-ante targets rather than ex-post quotas.

We show that, on the same set of assumptions, adapted to the non-relaxed case with ex-post quotas, \textsc{MaxEntropy} also achieves optimum manipulability for all three measures, and it does so without introducing the degeneracies of the $\ell_p$ norm.
\begin{theorem}[Resistance to Manipulation, Informal statement]
    Given a sortition instance satisfying an equivalent set of assumptions as the ones described in \cite{FlaniganLPW24}, \textsc{MaxEntropy} achieves optimum manipulability.
\end{theorem}

We defer a formal treatment of this result to \cref{apx:manip}.
\textsc{Goldilocks} also achieves a notion of optimum manipulability~\cite{BF24}, but under stronger assumptions on this instance, which makes these results hard to compare.

\subsection{Transparency}
\label{sec:transparency}
Now we turn our attention to the \emph{transparency} of the random selection.
To show that pool members are indeed selected with their fair selection probabilities (rather than included based on the whims of the organizers), \citet{FGG+21,FlaniganKP21} perform part of the panel selection in a public event, in what they call a \emph{transparent lottery}.
Specifically, having generated a panel distribution with column generation, they produce from this lottery a list of $m$ (say, 1000) panels, possibly with duplicates.
Then, they draw the final panel uniformly from these $m$ options, using a publicly observable process such as a lottery machine.

\citet{FlaniganKP21} propose a rounding algorithm that, given an explicit representation of the distribution with small support, obtains a list of $m$ panels, such that the selection probabilities when drawing from this list does not deviate much from the selection probabilities of the original distribution, say by at most $O(k/m)$ for each pool member. %
Their approach does not apply to the distributions of \textsc{MaxEntropy} and \textsc{FairMaxEntropy}, since it relies on a distribution that is explicitly given and has small support.
The best we can do is to sample $m$ panels independently from our distribution, in which case a straight-forward concentration bound yields:
\begin{proposition}
Let $\lambda$ be any panel distribution, and create a transparent lottery consisting of $m$ independently drawn panels.
Then, the selection probabilities $\hat{\pi}$ for the transparent lottery satisfy
\begin{align*}
    \norm{\hat{\pi} - \pi_{\lambda}}_\infty \leq \sqrt{\frac{\ln(2n) + \ln(\delta^{-1})}{2m}}.
\end{align*}
with probability at least $1-\delta$, for any $0 < \delta < 1$.
\end{proposition}
This bound is substantially worse than those of \citet{FlaniganKP21}, since it scales like $O(1/\sqrt{m})$ rather than in $O(1/m)$ in the lottery size $m$, and it only holds probabilistically.
Nevertheless, a larger number of panels $m$ can ensure that selection probabilities are faithfully represented with high probability; for example, in a pool of size 1,000, a transparent lottery of size 10,000 yields a maximum deviation of $2.5\%$ with a failure probability of at most $1\%$.

Note that, since the entropy of any choice over a list of size $m$ is at most $\log(m)$, there is little chance to preserve the high entropy (which requires large support) in a transparent lottery (which should draw from a manageable number of panels).
One upside is that drawing $m$ independent samples and choosing one uniformly among them is of course equivalent to drawing a single sample.
As a result, creating a transparent lottery does not reduce the entropy (or any other properties) of the overall sortition.

\section{Empirical Evaluation}
\label{sec:empirics}
Now we evaluate the empirical properties of our algorithms on $86$ real-world sortition instances provided to us by several organizations, mainly the \emph{Sortition Foundation} (United Kingdom), \emph{MASS LBP} (Canada), and the \emph{newDemocracy Foundation} (Australia).
This dataset is several times larger than those used in previous papers on sortition algorithms~\cite{FGG+21,GMK+25}, which allows us to gain more robust insights into the breadth of problems encountered in practice.
For a list of all instances and their sizes, see \cref{tab:instances} in \cref{apx:instances}.
Instance names (for example, sf\_b\_20) refer to the organization (``sf'' for Sortition Foundation) and the panel size (20).

We benchmark our algorithms against three prior ones: \textsc{LexiMin} \cite{FGG+21}, \textsc{Goldilocks} \cite{BF24}, and \textsc{Legacy}\footnote{\url{https://github.com/sortitionfoundation/stratification-app/}} (a heuristic used by the Sortition Foundation before their adoption of column generation).  %
For \textsc{Goldilocks}, we use a smooth version of the algorithm that replaces the $\ell_\infty$ norm~\cite{BF24} with the $\ell_{100}$ norm, which makes the optimal selection probabilities uniquely determined and is the version that has been used in practical assemblies. %
For both \textsc{LexiMin} and \textsc{Goldilocks}, we use column generation to obtain an explicit distribution, and then target the resulting selection probabilities with \textsc{FairMaxEntropy}.
We refer to the respective runs of \textsc{FairMaxEntropy} as \textsc{Maximum-Entropy Leximin} and \textsc{Maximum-Entropy Goldilocks}.

Due to the large number of experiments, we impose timeouts for the runs.
For column generation, the timeouts were $15$ minutes, $20$ minutes for \textsc{MaxEntropy}, and $30$ minutes for \textsc{FairMaxEntropy} (see \cref{apx:timeouts} for more details). \medskip

\noindent
\textbf{Statistics.}
Since the column generation versions of \textsc{LexiMin} and \textsc{Goldilocks} yield explicit distributions, selection probabilities and other empirical quantities for them are exact rather than estimates.
For the other algorithms, we estimate selection probabilities from the empirical frequencies  of inclusion over a number of sampled panels ($10^5$ for \textsc{MaxEntropy} and \textsc{Legacy} and $10^4$ for \textsc{FairMaxEntropy}).
For each pool member, the number of panels in which they appear is distributed binomially with an unknown success probability.
We compute a Jeffreys $95\%$ confidence interval for each selection probability, which yields the error bars in the plots.
We follow the same approach for estimating the probability that a random panel satisfies a feature's quotas.

\subsection{Implementation and Running Times}
\label{sec:implementation}
Implementing our algorithms was a challenging engineering task. The number of states in the dynamic programming algorithm grows rapidly, which makes memory consumption a bottleneck.
Our implementation involves specialized data structures for store the dynamic programming states and a careful memory layout to even make medium-sized instances viable.

The algorithm trades off memory and execution time by deciding which features to satisfy explicitly using dynamic programming and which to satisfy via rejection sampling, which we can do in several threads in parallel.
After each feature is added, the algorithm samples $10^5$ panels satisfying the current set of features to estimate the satisfaction probabilities of the unenforced features. 
Then, the algorithm either attempts to add the feature with the lowest satisfaction probability to the dynamic program or defers it to rejection sampling.

In \cref{tab:running_times}, we give the running times of \textsc{MaxEntropy} on the set of instances studied by~\citet{FGG+21}. 
These running times were computed on a 2023 MacBook Pro M3 with 18 GB of memory, with 5 threads allocated for sampling. 
The exception is sf\_e\_110, for which we used a cloud machine with 64 GB of memory and 5 threads allocated for sampling; the peak memory usage was $50$ GB. 
We were unable to successfully sample a panel from obf\_30, even on the cloud machine.

On average, \textsc{MaxEntropy} was able to sample a panel from $78$ out of $86$ instances and sample more than $10^5$ panels on $68$ of them.
The eight instances for which \textsc{MaxEntropy} failed to sample a panel within the $20$ minute timeout %
all have at least $8$ features, the largest number among the 10 sortition instances studied by \citet{FGG+21}, and most have a large panel size. %
\medskip
\begin{table}[tbh]
\footnotesize
\caption{Running times of \textsc{MaxEntropy} on the instances studied by \citet{FGG+21}.
For all instances except cca\_75 and sf\_e\_110, the sampling times are averaged over $10^4$ samples; for cca\_75, over 100 samples; and for sf\_e\_110, a single sample.}%
\label{tab:running_times}%
\begin{tabular}{rccccc}
\toprule
Instance    & Pool size & Features & Values & Counting (s) & Sampling (ms) \\
\midrule
cca\_75     & 825       & 4        & 66     & 226.5        & 136,398.3             \\
hd\_30      & 239       & 7        & 33     & 448.6        & 30.53              \\
mass\_a\_24 & 70        & 5        & 11     & 1.8          & 0.0045             \\
nexus\_170  & 342       & 5        & 33     & 229.1        & 18.9              \\
obf\_30*    & 321       & 8        & 38     & Timeout            & Timeout                  \\
sf\_a\_35   & 312       & 6        & 22     & 69.9         & 13.8               \\
sf\_b\_20   & 250       & 6        & 20     & 6.1           & 0.074             \\
sf\_c\_44   & 161       & 7        & 20     & 478.4        & 8.9                \\
sf\_d\_40   & 404       & 6        & 19     & 6.2          & 1.7                \\
sf\_e\_110* & 1727      & 7        & 31     & 623.2        & 369,374.2           \\
\bottomrule

\end{tabular}
\end{table}
\begin{table}[tbh]
\footnotesize
\caption{Running times of \textsc{FairMaxEntropy} on \textsc{LexiMin} selection probabilities, on instances studied by \citet{FGG+21} without timeout.}%
\label{tab:running_times_fme}%
\begin{tabular}{rccccc}
\toprule
Instance    & Pool size & Features & Values & Time (s)  \\
\midrule
mass\_a\_24 & 70        & 5        & 11     & 18.8                    \\
sf\_a\_35   & 312       & 6        & 22     & 940.7 \\
sf\_b\_20   & 250       & 6        & 20     & 115                   \\
sf\_c\_44   & 161       & 7        & 20     & 1416.6     \\
sf\_d\_40   & 404       & 6        & 19     & 607.2                     \\
\bottomrule
\end{tabular}
\end{table}

Extending \textsc{MaxEntropy} to \textsc{FairMaxEntropy} imposes additional challenges. 
At each iteration of the algorithm, it (1)~computes the weighted counts with the dynamic program, and (2)~samples enough panels to estimate selection probabilities.
For (1), a simple but powerful observation is that the non-zero states of the dynamic programming are the same for every (positive) weight vector $\mu$.
Thus, steps past the first immediately have access to a perfectly pruned dynamic program, which speeds up the reweighting considerably.
For (2), as the optimization routine progresses, weights yielded by the dual solution increase in magnitude, which increases the computational cost of sampling and would risk numeric instability with floating-point numbers.
We solve this by rounding the weights to integers in the range $[1, 10^{12}]$ and performing the computations with exact integer precision.

Finally, we use ADAM~\cite{KingmaB15} as the underlying solver for practical performance.
In \cref{tab:running_times_fme}, we present selected running times for \textsc{FairMaxEntropy} on \textsc{LexiMin} selection probabilities for those instances from \cref{tab:running_times} for which it does not time out.
In total, we were able to solve $48$ instances with \textsc{LexiMin} selection probabilities and $43$ with those of \textsc{Goldilocks}.
Since \textsc{Goldilocks} is, in general, slower than \textsc{LexiMin}, 
the column generation timed out more often, which caused some distributions to be missing for \textsc{FairMaxEntropy} to optimize on. (See \cref{tab:instances} for more information).

\subsection{Fairness in Selection Probabilities}
\label{sec:fairness}
We start our discussion of empirical properties by comparing the fairness of selection probabilities across the selection algorithms. 
In general, there are no fairness guarantees for \textsc{MaxEntropy}, as the selection probabilities rely on the structure of the pool,
and even if there is a distribution that equalizes all selection probabilities, \textsc{MaxEntropy} might not sample from it.

We give two representative examples in \cref{fig:selec_prob}.
Regarding minimum selection probabilities, no algorithm can perform better than \textsc{Leximin}, since it optimizes that measure. \textsc{Goldilocks} is often quite close, but might loosen the lower probabilities in favor of reducing the upper probabilities.
Both \textsc{MaxEntropy} and \textsc{Legacy} tend to have some pool members with very low selection probabilities, with \textsc{Legacy} substantially closer to zero~\cite{FGG+21}.
Regarding the maximum selection probability, \textsc{MaxEntropy}
sits halfway between the best performer \textsc{Goldilocks} and the worst, \textsc{Leximin}, which has the tendency to assign very high selection probabilities to some pool members.
As illustrated in the plot for sf\_c\_44, whenever there is a pool member that is contained in all panels, \textsc{Goldilocks} tends to assign high selection probabilities to multiple members, whereas \textsc{MaxEntropy} does not.

We use the \emph{Gini inequality coefficient} and \emph{geometric mean} as global measures of fairness that go beyond minimum and maximum selection probabilities, the results are contained in \cref{tab:fairness_avg_med}.
On those metrics, \textsc{MaxEntropy} sits again between the unfairer \textsc{Legacy} and fairer column-generation optimizing a notion of fairness. %
\begin{figure}[tbh]
  \begin{subfigure}[t]{0.49\textwidth}
    \centering
    \includegraphics[width=\linewidth]{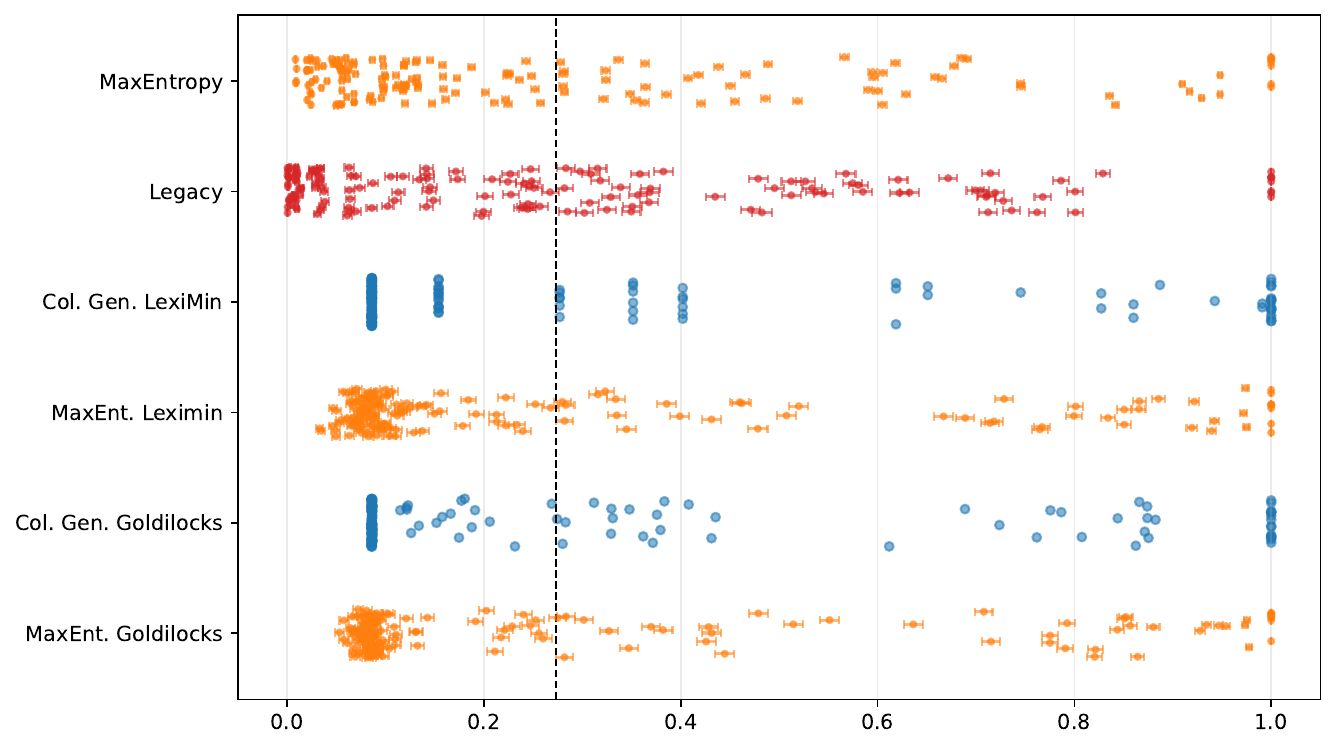}
    \caption{Selection probabilities for sf\_c\_44.}
    \label{fig:selec_prob_c}
  \end{subfigure}
  \hfill
  \begin{subfigure}[t]{0.49\textwidth}
    \centering
    \includegraphics[width=\linewidth]{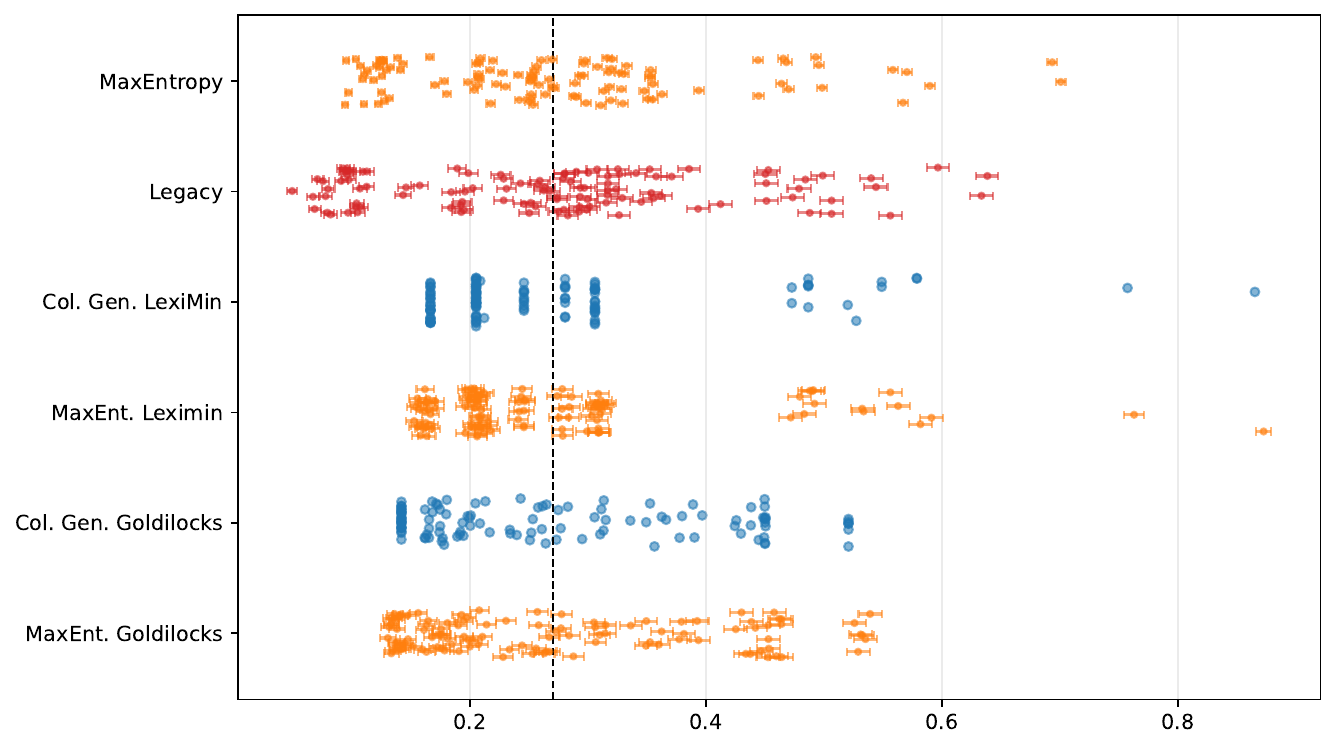}
    \caption{Selection probabilities for sf\_k\_30.}
   \end{subfigure}
  \caption{Strip plots for the selection probabilities, each point represents a pool member and its position represents the selection probability. 
  Vertical lines mark the equalized selection probabilities $(\frac{k}{n}, \dots, \frac{k}{n})$.}
  \label{fig:selec_prob}
\end{figure}
\begin{table}[tbh]
\setlength{\tabcolsep}{4pt}
\centering
\footnotesize
\caption{Aggregate fairness metrics (average and median) across all 47 common instances. Metrics are based on sample averages for algorithms other than column generation. “$\blacktriangle$” means “higher is better”.}
\begin{tabular}{r*{4}{cc}}
\toprule
& \multicolumn{2}{c}{\smash{\textsc{MaxEntropy}}} & \multicolumn{2}{c}{\smash{\textsc{Legacy}}} & \multicolumn{2}{c}{\smash{\textsc{Leximin}}} & \multicolumn{2}{c}{\smash{\textsc{Goldilocks}}} \\
\cmidrule(lr){2-3} \cmidrule(lr){4-5} \cmidrule(lr){6-7} \cmidrule(lr){8-9}
& Avg. & Med. & Avg. & Med. & Avg. & Med. & Avg. & Med. \\
\midrule
Gini ($\blacktriangledown$) & 0.478 & 0.482 & 0.505 & 0.512 & 0.441 & 0.448 & 0.443 & 0.448 \\
Geometric mean ($\blacktriangle$) & 0.122 & 0.077 & 0.096 & 0.062 & 0.129 & 0.083 & 0.129 & 0.083 \\
Minimum probability ($\blacktriangle$) & 0.034 & 0.008 & 0.021 & 0.001 & 0.080 & 0.052 & 0.077 & 0.050 \\
Maximum probability ($\blacktriangledown$) & 0.909 & 1.000 & 0.876 & 1.000 & 0.975 & 1.000 & 0.958 & 1.000 \\
\bottomrule
\end{tabular}
\label{tab:fairness_avg_med}
\end{table}

The \textsc{FairMaxEntropy} distributions in \cref{fig:selec_prob_c} only completed two iterations of gradient descent in the allocated time.
Even though the distributions have clearly not converged to their column-generation targets in this time, they have made significant progress in that direction from their starting point, i.e., the \textsc{MaxEntropy} distribution, in particular in terms of the minimum probability.
This is encouraging, since it shows that even a few iterations of gradient descent may be worthwhile to avoid low selection probabilities.

\subsection{Intersectional Diversity}
\label{sec:diversity}
One of the central ideals of sortition is \emph{representativeness}: if the panel was a truly uniform sample from the population, the selected panel approximates the population composition across all dimensions of interest \cite{Engelstad89,CM99}.
The quotas imposed by practitioners enforce representation across the quota-protected features, but
they do not directly
constrain the representation of \emph{joint categories} such as \emph{Education $\times$ Income $\times$ Region}.
This limitation is a central concern in the literature on intersectionality, which emphasizes that single-axis categories can obscure systematic disadvantages and patterns of underrepresentation that arise at their intersections~\cite{Wojciechowska19,GillL22,OECD20}.
Imposing quotas on these intersectional groups is generally not possible, since there are so many of them, and since the population share of complex intersectional groups is often not known.

Due to these obstacles, we focus on intersectional diversity rather than representation~\cite{SBJ+20}, by which we mean that it is desirable to see a variety of distinct intersectional groups and a lack of unnecessary correlations between features within the same panel.
We summarize intersectional outcomes using three metrics: two focusing on all features at once, and one focusing on pairs of features in isolation.
We will call \emph{vector count} the expected number of distinct feature-value vectors present in a sampled panel, and apply the information-theoretic measure of \emph{total correlation}~\cite{Watanabe60} to quantify how much a panel's feature-value vector distribution differs from a product distribution over all features.
Since much of the focus is on intersections of two features, we measure, for each pair of features, the \emph{normalized mutual information}.
That is, knowing the panel and drawing a random person from it, how much does revealing this person's value for $f_1$ reduce the information needed to communicate their value for $f_2$? (Larger values mean more correlation.)
The global metrics are contained in \cref{tab:diversity} and the plots for the pairwise mutual information comparisons are in \cref{apx:diveristy}.
All values are normalized to facilitate the comparison between algorithms. We restricted the analysis to the set of 39 instances  where all selection algorithms where able to produce a distribution.

\begin{table}[tbh]
\footnotesize
\caption{Average of expected vector count, average of total correlation, and median normalized mutual information for each selection algorithm. “$\blacktriangle$” means “higher is better”.}
\label{tab:diversity}
\begin{tabular}{rcccccc}
\toprule
                  &  &  & \multicolumn{2}{c}{\smash{\textsc{LexiMin} fairness}} & \multicolumn{2}{c}{\smash{\textsc{Goldilocks} fairness}} \\ 
                  \cmidrule(lr){4-5} \cmidrule(lr){6-7}
                  & \smash{\textsc{MaxEntropy}}  &  \smash{\textsc{Legacy}} & col.\ generation & max.\ entropy & col.\ generation & max.\ entropy \\ 
\midrule
Vector Count ($\blacktriangle$)         & 0.914  & 0.908    & 0.823   & 0.910 ($10.4\%\uparrow$)      & 0.830      & 0.908 \hphantom{1}($9.1\%\uparrow$)\\
Total Corr.\ ($\blacktriangledown$) & 0.599  & 0.602    & 0.655   & 0.601 \hphantom{1}($8.2\%\downarrow$)     & 0.651      & 0.602 \hphantom{1}($7.5\%\downarrow$)     \\ 
Median NMI ($\blacktriangledown$)       & 0.143  & 0.146    & 0.168   & 0.145 ($13.6\%\downarrow$)     & 0.166      & 0.144 ($13.2\%\downarrow$)     \\ 
\bottomrule
\end{tabular}
\end{table}

\cref{tab:diversity} shows that, out of the box, maximum-entropy produces the most diverse panels according to all three measures, whereas column generation performs worst, and similarly bad for both fairness measures.
Given that \textsc{FairMaxEntropy} achieves almost the same diversity as pure \textsc{MaxEntropy}, this is not a consequence of the fairness measure.
We conjecture that the pricing problem, used by column generation to generate new panels, tend to have optimal solutions that have below-average diversity.
\textsc{MaxEntropy} and \textsc{Legacy} show similar results across all metrics, which also supports the assumption that low diversity is an artifact of column generation. %
If one wants to generate panels with even higher vector count, this could be easily incorporated in the dynamic program to sample with maximum entropy from only maximal vector-count panels.

\subsection{Generalization to Unprotected Features}
\label{sec:generalization}
An often-cited ideal~\cite{Fishkin09,VanReybrouck16} is that the citizens' assembly should be ``in miniature an exact portrait of the people at large'' (in the words of John Adams).
No matter which features assembly organizers choose to set quotas on, there will invariably be many more, relevant to the deliberation, not protected by quotas.
In this section, we compare selection algorithms based on their ability to satisfy an unknown feature, that is, one for which there are no quotas; we denote this by the \emph{generalization probability}.
\citet{flanigan2023minipublic} have argued that selection algorithms with fairer selection probabilities may tend be more representative in this situation.

To test this scenario with the available data, we leave out quotas.
That is, for each instance and feature, we pretend that this feature were not part of the instance.
Then, we run each selection algorithm and compute its probability of satisfying the held-out feature.
Since the total number of instances in this experiment was 554, we imposed a stricter timeout for the selection algorithms (10 minutes for column-generation algorithms, 15 minutes for \textsc{MaxEntropy} and \textsc{Legacy} with $10^5$ samples, and 25 minutes for \textsc{FairMaxEntropy} with $10^4$ samples).

\begin{figure}[tbh]
  \centering
  \begin{subfigure}[t]{0.45\textwidth}
    \centering
    \includegraphics[width=\linewidth]{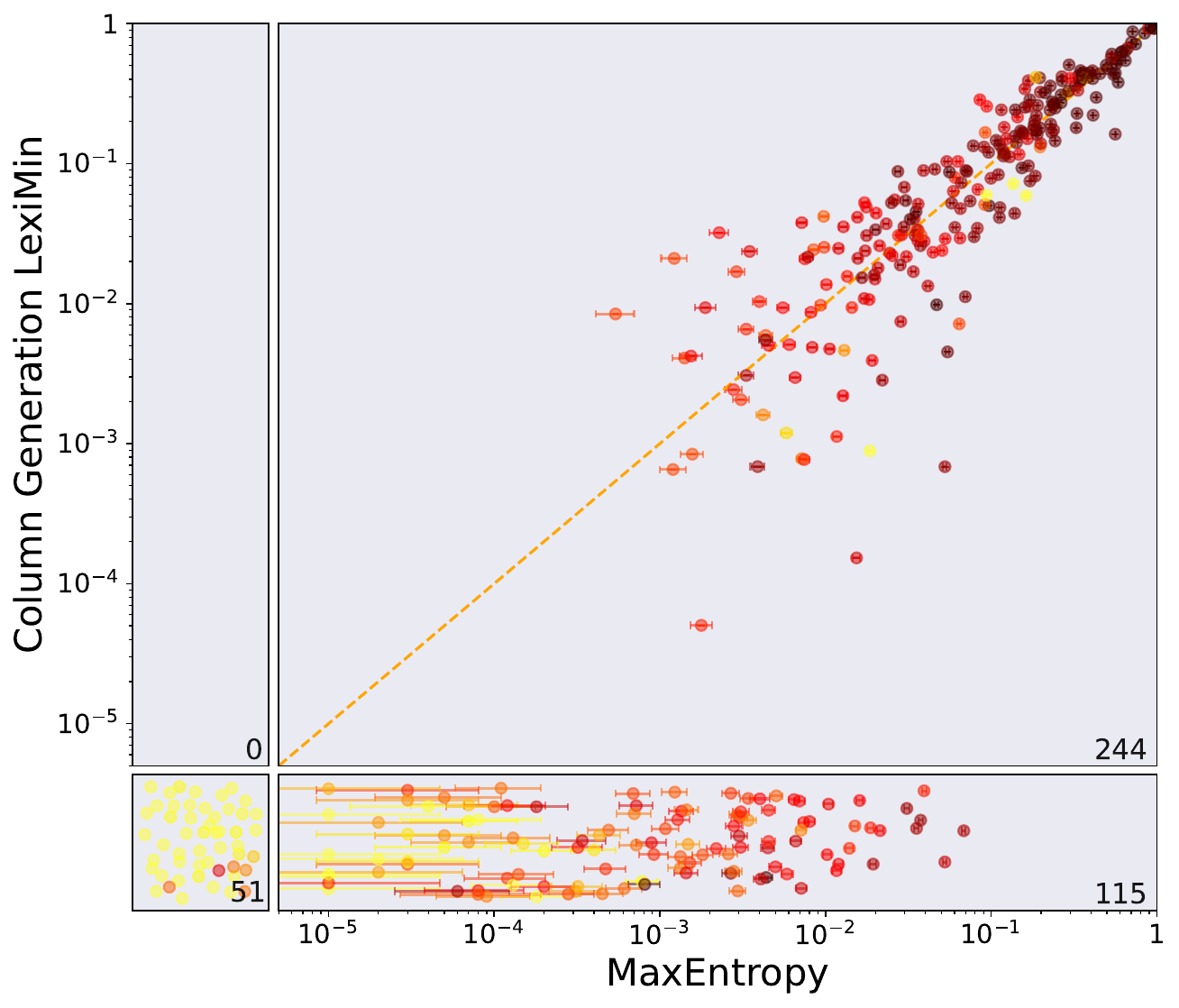}
    \label{fig:left}
  \end{subfigure}\hfill
  \begin{subfigure}[t]{0.535\textwidth}
    \centering
    \includegraphics[width=\linewidth]{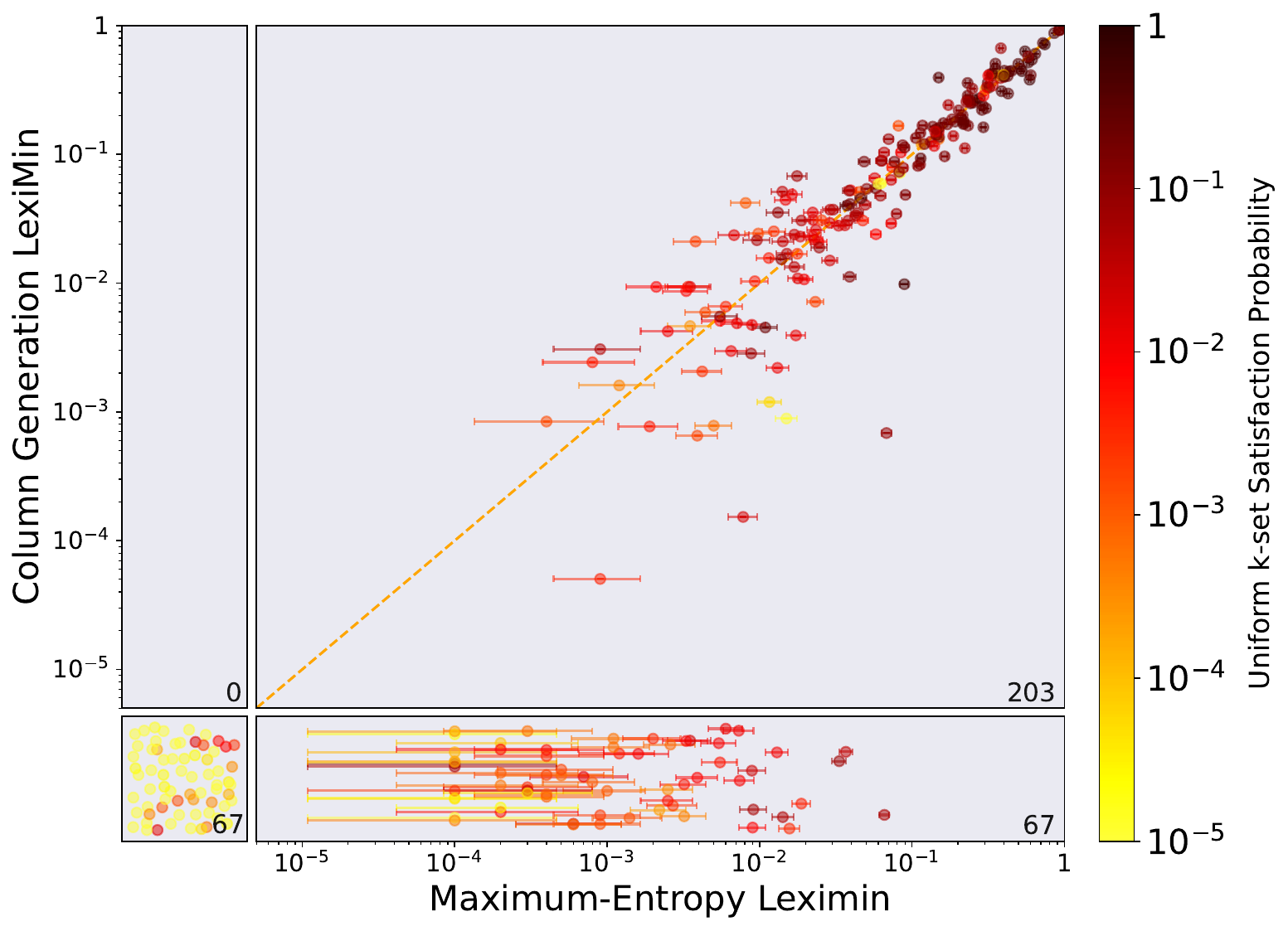}
    \label{fig:right}
  \end{subfigure}
  \caption{Generalization probabilities. Points are colored according to the probability that a uniform set of $k$ pool members satisfies the held-out feature. Points with probabilities/empirical means less than $5 \cdot 10^{-6}$ are moved into side boxes, see footnote.}
  \label{fig:generalization}
\end{figure}

We display two representative comparisons in \cref{fig:generalization} and defer the others to \cref{apx:generalization}.
Both cases illustrate the general finding that column generation has lower satisfaction probabilities for the held-out feature than other approaches.
While many feature-instance pairs have similar satisfaction probabilities in, say, \textsc{MaxEntropy} and column generation \textsc{LexiMin}, there are many that are satisfied with zero probability by the column generation algorithm.\footnote{Since column generation produces an explicit representation, we verify that these probabilities are indeed equal to zero, not just smaller than the cutoff.
For \textsc{MaxEntropy}, the probability of satisfying the held-out feature is computed as a sample mean. Since there are panels satisfying the left-out feature and our algorithms put positive mass on all panels, these probabilities can never be zero, and the probabilities in the lower left corner are simply too small to show up in our samples.}

This is not a simple artifact of the fairness measure, as maximum-entropy \textsc{LexiMin} obtains higher generalization probabilities.
In both plots, when column-generation matches the maximum-entropy, it typically happens on the instance-feature combinations that are already easily satisfiable by sampling $k$ pool members uniformly from the pool, which suggests that those quotas are easy to satisfy.
The pattern of zero generalization probabilities for column generation is not limited to any specific organization's instances, as we show in \cref{fig:generalization_org}.

In our view, the better generalization probabilities of our algorithms validate Jaynes'~\cite{Jaynes57} justification of the maximum-entropy distribution.
By restricting their support to a small subset of panels, column generation algorithms introduce a “bias” that arbitrarily commits to never representing some features in their correct ranges.

\section{Practical Deployment}
\label{sec:panelot}
Since, as we have argued, maximum-entropy algorithms has a range of desirable properties, we are making it available to practitioners.
We release our implementation as open source on Github\footnote{https://github.com/gabrielmorete/maximally-random-sortition}.
We believe that self-hosting our algorithm will be attractive to privacy-concerned organizations, which will be easier since our algorithm does not rely on proprietary optimization libraries.

To facilitate the algorithm's uptake by practitioners, we will also make it available as one of the algorithmic choices on \href{https://panelot.org}{\emph{panelot.org}}.
The left panel \cref{fig:panelot} shows the planned screen in which practitioners can choose between algorithms and learn about their distinct properties.
The figure's right panel shows the execution page, which will update statistics from the run and pool in real time, providing a progress estimate and an illustration of the current state of the dynamic program.
During execution, the interface presents how each successive feature shrinks the count of panels, and which remaining features are hard to satisfy.
This may help practitioners become familiar with the algorithm and help them adapt their quotas if they cannot be solved by \textsc{MaxEntropy}. %
To complete, the algorithm will sample $10^4$ panels to plot the selection probabilities and provide a physical lottery.
Based on practitioner interests and improvements in algorithm performance, we might also make versions of \textsc{FairMaxEntropy} available in the future.

\begin{figure}[tbh]

  \begin{subfigure}[t]{0.64\textwidth}
    \centering
    \fbox{\includegraphics[width=\linewidth]{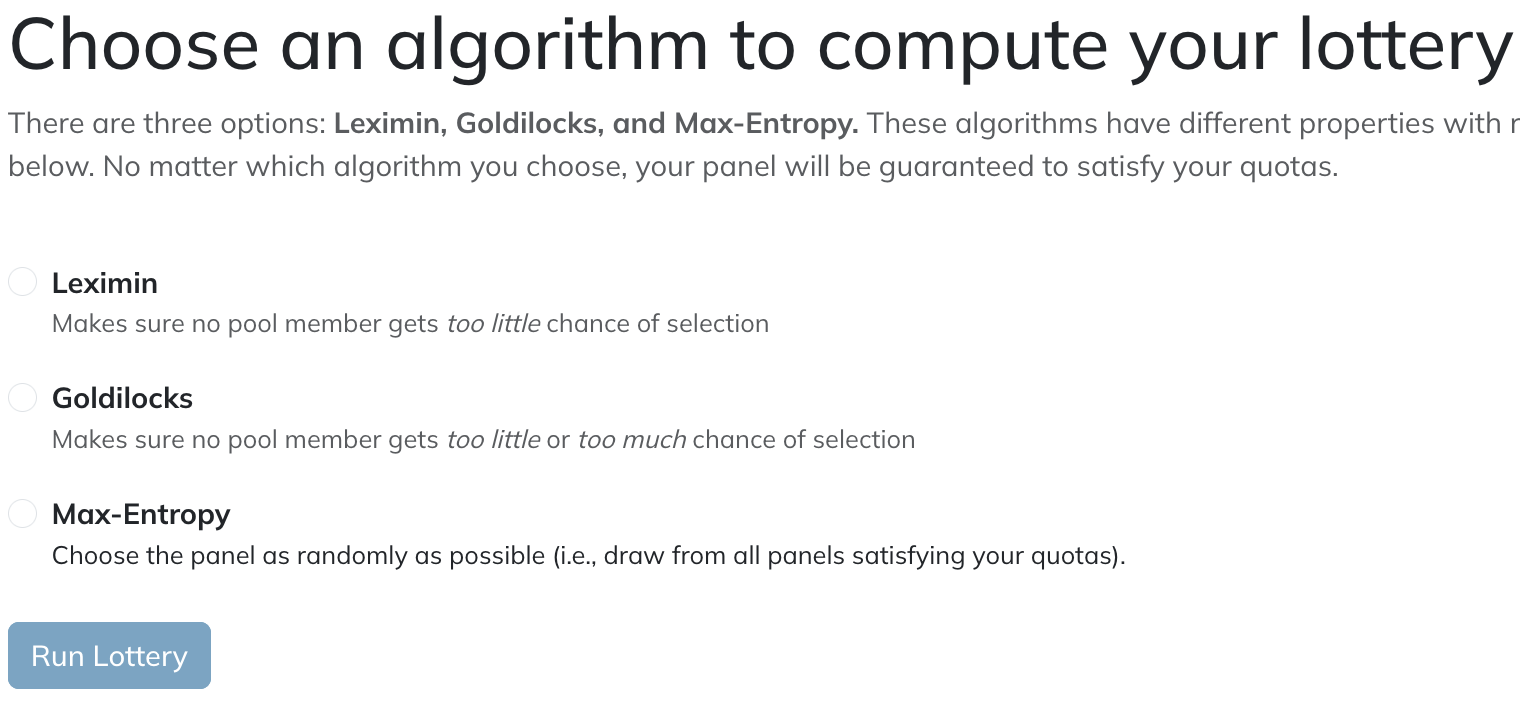}}
  \end{subfigure}
  \hfill
  \begin{subfigure}[t]{0.322\textwidth}
    \centering
    \fbox{\includegraphics[width=\linewidth]{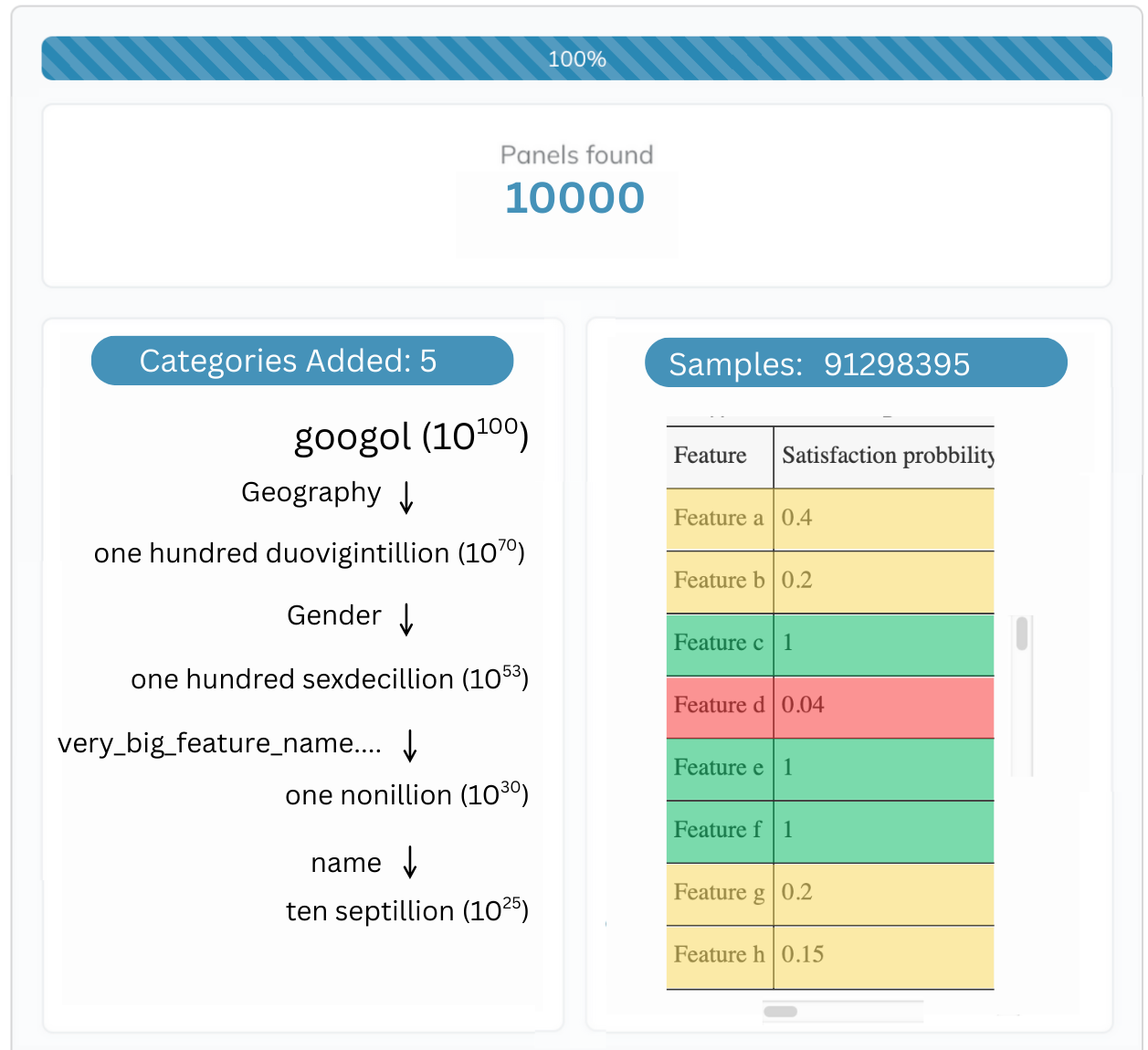}}
   \end{subfigure}
  \caption{Planned interface of Panelot for choosing the selection algorithm (left) and visualization of the state of \textsc{MaxEntropy} (right).}
  \label{fig:panelot}
\end{figure}

\section{Conclusion}
\label{sec:conclusion}
By providing a new selection algorithm with distinct advantages, we hope to contribute to practitioners' efforts to increase public trust in citizens' assemblies.
We believe that this is timely, since higher-stakes uses for assemblies will likely lead to higher scrutiny as well.

Our empirical evaluation points to drawbacks of column generation algorithms, in terms of the diversity of their panels and their generalization to unknown features.
At the same time, it is very encouraging to see that the \textsc{FairMaxEntropy} seemed not subject to these limitations despite converging to the same, fairness-optimal selection probabilities.
These findings suggest that, when possible to compute, maximizing entropy is a worthwhile complementary objective even when the primary objective is fairness.

\section*{Acknowledgments}
We thank Soroosh Shafiee and Noah Stephens-Davidowitz for helpful technical conversations. We thank the Center for Climate Assemblies, Healthy Democracy, MASS LBP, New Democracy, Nexus Institut, Of by For, and the Sortition Foundation for providing the sortition instances for our empirical analysis, and Carmel Baharav and Bailey Flanigan for sharing cleaned versions of the newDemocracy data. We thank Carmel Baharav, Anthony Bucci, and Bailey Flanigan for their contributions in implementing \textsc{MaxEntropy} in the upcoming release of Panelot.

\bibliographystyle{ACM-Reference-Format}
\bibliography{bibliography,paulzotero}

\clearpage

\appendix
\crefalias{section}{appendix}
\crefalias{subsection}{subappendix}      %
\crefalias{subsubsection}{subsubappendix} %

\section*{\Large APPENDIX}
\section{Deferred Details for \nameref{sec:me}}
\label{apx:me}

\subsection{The Generalized Counting Problem}
\label{apx:gcp}
In this appendix, we provide proofs and generalize the dynamic programming approach from \cref{sec:me}.
Throughout this section, fix a sortition instance $\calI$ and weights $\mu\in\R^N_{>0}$.
For $\calU\subseteq\calP$, define the \emph{weighted count} of panels as
\begin{align*}
\Count{\mu}{\calU}\coloneqq
  \sum_{P\in\calU}\prod_{i\in P}\mu_i.
\end{align*}
In particular, $\card{\calP}=\Count{\ones}{\calP}$.
The \emph{generalized counting problem} consists of computing $\Count{\mu}{\calP}$.

Fix an ordering $1,\dots,n$ of pool members. 
We define the \emph{weighted counting function} as follows:
\begin{align*}
\phi_\mu(i, z)\coloneqq \sum_{\substack{U\subseteq \{i,\dots,n\}:\\ z+w(U)\in\overline{\calZ}}} \prod_{t\in U}\mu_t.
\end{align*}
Note that
\begin{equation}
    \label{eqn:phi_count}
    \phi_\mu(1,0)
    = \sum_{\substack{U\subseteq N:\\ w(U)\in\overline{\calZ}}}
    \prod_{t\in U}\mu_t
    = \sum_{P\in\calP}\prod_{t\in P}\mu_t
    = \Count{\mu}{\calP},
\end{equation}
where the second equality uses that $U\subseteq N$ is a panel if and only if $w(U)\in\overline{\calZ}$.

\paragraph{Proof of \cref{prp:phi_heuristic}}
We restate \cref{prp:phi_heuristic} in weighted form.
\begin{proposition}
    \label{prp:phi_heuristic_wgt}
    For all $i \in N$, $z \in \N^{\FV}$, and $F' \subseteq F$, we have that
    $\phi_\mu(i, z) \leq \phi_\mu'(i, \restr{z}{F'})$.
\end{proposition}
\begin{proof}
Let $z'\coloneqq \restr{z}{F'}$. Define the \emph{set of valid extensions} as follows
\begin{align*}
    \calU(i, z) \coloneqq \setst{U \subseteq \curly{i, \dots, n}}{z + w(U) \in \overline{\calZ}}
    \quad \text{and} \quad
    \calU'(i, z') \coloneqq \setst{U \subseteq \curly{i, \dots, n}}{z' + w'(U) \in\overline{\calZ}'}.
\end{align*}
If $U \in \calU(i, z)$, then $z + w(U) \in \overline{\calZ}$. 
Thus, $\restr{\paren*{z + w(U)}}{F'} = z' + w'(U) \in \overline{\calZ}'$.
Since $\overline{\calZ}'$ is obtained by restricting the constraints of $\overline{\calZ}$. 
Hence, $\calU(i, z) \subseteq \calU'(i,z')$.
Finally, since $\mu > 0$, we have that
\begin{equation*}
    \phi_\mu(i, z) = \sum_{U \in \calU(i,z)}\prod_{t\in U} \mu_t
    \leq \sum_{U \in \calU'(i, z')}\prod_{t \in U} \mu_t
    = \phi_\mu'(i, z'). \qedhere
\end{equation*}
\end{proof}

\paragraph{Proof of  \cref{prp:unif_dp}}
Here, we consider more natural to state the DP in terms of the feature-vector vectors instead of pool members.
Fix an ordering $\omega_1,\dots,\omega_{n_\calW}$ of the distinct feature-value vectors in the pool.
For each $j \in [n_\calW]$ and $d = 0, \dots, n_{\omega_j}$, denote the \emph{weighted combinations} of $d$ members with feature-value vector $\omega_j$ by
\begin{equation}
    \label{eqn:weighted_convolution}
    C_\mu(j, d)\coloneqq \sum_{\substack{S\subseteq \restr{N}{\omega_j}:\\ \card{S}=d}} \prod_{i\in S}\mu_i.    
\end{equation}
Equivalently, $C_\mu(j,d)$ is the coefficient of $x^d$ in the polynomial $\prod_{i\in \restr{N}{\omega_j}}(1+\mu_i x)$.
We restate \cref{prp:unif_dp} in weighted form.

\begin{theorem}
\label{thm:wgt_dp}
    For every $j \in [n_\calW]$ and $z \in \N^{\FV}$, we have that
    \begin{align*}
        \phi_\mu(j,z) = \sum_{d=0}^{n_{\omega_j}}\,  C_\mu(j,d)\; \phi_\mu(j+1, z + d\ \omega_j),
    \end{align*}
    with base case $\phi_\mu(n_\calW + 1, z) = \Iverson{z \in \overline{\calZ}}$.
\end{theorem}

\begin{proof}
We prove the statement by induction.
The base case follows from the definition of panel.

For each $j = 1, \dots, n_\calW$, let $N_j$ be the set of pool members with feature-value vectors $\omega_j, \dots, \omega_{n_\calW}$.
Fix $j\in[n_\calW]$ and $z\in \N^{\FV}$. 
Every $U \subseteq N_j$ decomposes uniquely as a disjoint union $U = S \dotcup{} T$, where $S \subseteq \restr{N}{\omega_j}$ and $T \subseteq N_{j + 1}$.
For all $i \in S$, since $w(i)=\omega_j$, we have $w(S) = \card{S} \omega_j$. Thus,
\begin{align*}
\phi_\mu(j,z) & = \sum_{S\subseteq \restr{N}{\omega_j}}
    \prod_{i\in S} \mu_i
    \sum_{\substack{T \subseteq N_{j + 1}:\\
    z + w(S) + w(T) \in \overline{\calZ}}} \prod_{i\in T} \mu_i \\
    & = \sum_{d=0}^{n_{\omega_j}}
    \sum_{\substack{S \subseteq \restr{N}{\omega_j}:\\ \card{S} = d}}
    \prod_{i \in S}\mu_i \phi_\mu(j + 1, z + d\omega_j) \\
    & = \sum_{d=0}^{n_{\omega_j}} C_\mu(j,d)
    \phi_\mu(j + 1, z + d\omega_j). \qedhere
\end{align*}
\end{proof}

As in the uniform case, the coefficients $C_\mu(j,d)$ can be precomputed for all $j$ and all $d\le k$ in total time $O(nk)$, yielding the same complexity.

\subsection{Sampling from Product Distributions}
\label{apx:sampling}
In this section, we give a formal description of the generalized sampling algorithm. 

Throughout this section, fix a sortition instance $\calI$ and weights $\mu\in\R^N_{>0}$. We prove our statements for the profile aggregated version of the counting function $\phi_\mu$.
For each $j \in [n_\calW]$, $z \in \Z_{\geq 0}^{\FV}$, and $d \in \curly*{0, \dots,n_{\omega_j}}$ let $C_\mu(j, d)$ be as defined in \eqref{eqn:weighted_convolution}.
Define the following probability distributions:
\begin{align*}
    p_{j,z}(d) \coloneqq \frac{C_\mu(j,d) 
    \phi_\mu(j + 1, z + d\omega_j)}{\phi_\mu(j,z)},
    \quad \text{ and } \quad
    q_{j,d}(S)\coloneqq \frac{\prod_{i\in S}\mu_i}{C_\mu(j,d)},
    \quad S\subseteq \restr{N}{\omega_j}, \card{S} = d.
\end{align*}
The distribution $p_{j, z}$ describes how to move to the next state and $q_{j, d}$ defines a distribution of subsets of $\restr{N}{\omega_j}$ of size $d$. 

\noindent
\begin{minipage}[t]{.49\textwidth}
\vspace{0pt}
\begin{algorithm}[H]
\SetAlgoLined
\KwIn{A sortition instance $\calI'$
and weights $\mu \in \R^{N}_{>0}$.}
\KwOut{A panel $P \sim \lambda_\mu'$.}
Compute $\phi_\mu'$\;
$(j,z)\gets(1,0)$\;
$P\gets\varnothing$\;

\While{$j\le n_\calW$}{
  Sample $d$ according to $p_{j,z}$\;
  Sample $S\subseteq \restr{N}{\omega_j}$ according to $q_{j,d}$\;
  $P \gets P \cup S$\;
  $(j, z)\gets (j + 1, z + d\omega_j)$\;
}
\Return $P$\;
\caption{Exact Sampling}
\label{alg:sampling}
\end{algorithm}
\end{minipage}\hfill
\begin{minipage}[t]{.49\textwidth}
\vspace{0pt}
\begin{algorithm}[H]
\SetAlgoLined
\KwIn{A sortition instance $\calI$ and weights $\mu\in\R^N_{>0}$.}
\KwOut{A panel $P\sim\lambda_\mu$.}
Let $F' \subseteq F$ be a set of features\;
\Repeat{$P\in\calP$}{
  Sample $P\sim \lambda'_{\mu}$ using \cref{alg:sampling}\;
}
\Return $P$\;
\caption{\textsc{MaxEntropy}}
\label{alg:max_entropy}
\end{algorithm}
\end{minipage}

Fix $F'\subseteq F$ and let $w' \coloneqq \restr{w}{F'}$, and $\overline{\calZ}'$ be the quota-compliant profiles restricted to $F'$. Let
\begin{align*}
\calP'\coloneqq\setst*{P\subseteq N}{w'(P)\in\overline{\calZ}'}
\end{align*}
be the set of panels feasible with respect to $F'$. 
Define the distribution $\lambda'_\mu \in \Delta(\calP')$ by
\begin{align*}
    \lambda'_\mu(P) \coloneqq
    \frac{\prod_{i\in P'} \mu_i }{\Count{\mu}{\calP'}},
    \quad \text{for all } P' \in \calP'.
\end{align*}

\begin{proposition}
    \label{prop:sampling}
    For every panel $P' \in \calP'$, it holds that
    \begin{align*}
        \Pr\sqbrac{\text{\cref{alg:sampling} outputs } P'} = \frac{\prod_{i\in P'}\mu_i}{\Count{\mu}{\calP'}}.
    \end{align*}
    
\end{proposition}
\begin{proof}
    Fix $P'\in\calP'$. For each $j\in[n_\omega]$, let
    $S_j \coloneqq P'\cap \restr{N}{\omega_j}$ and 
    $d_j \coloneqq \card{S_j}$.
    Define $z_1 \coloneqq 0$ and $z_{j+1}\coloneqq z_j + 
    d_j\omega_j$. 
    Then $z_{{n_\calW} + 1} = w(P') \in \overline{\calZ'}$. 
    The algorithm outputs $P$ if and only if it selects $d_j$ at state $(j,z_j)$ and then selects $S_j$ at layer $j$.
    Thus, we have the following
    \begin{align*}
    \Pr\sqbrac{\text{Alg.\ outputs }P'}
    &= \prod_{j=1}^{n_\calW} \Pr\sqbrac{d = d_j \mid (j, z_j)}
       \prod_{j=1}^{n_\calW}
       \Pr\sqbrac{S = S_j \mid d_j, (j, z_j)}\\
    &= \prod_{j=1}^{n_\calW} p_{j, z_j}(d_j)
       \prod_{j=1}^{n_\calW} q_{j, d_j}(S_j).
    \intertext{Substituting the definitions of $p_{j,z}(d)$ and $q_{j,d}(S)$ gives}
    \Pr\sqbrac{\text{Alg.\ outputs }P'}
    &= \prod_{j=1}^{n_\calW}
      \frac{C_\mu(j, d_j) \phi_\mu(j + 1, z_{j+1})}{\phi_\mu(j, z_j)}
      \prod_{j=1}^{n_\calW} \frac{\prod_{i\in S_j}\mu_i}{C_\mu(j, d_j)}\\
    & = \frac{\phi_\mu(n_\calW + 1,z_{{n_\calW} + 1})}{\phi_\mu(1, 0)}
      \prod_{j=1}^{n_\calW} \prod_{i\in S_j}\mu_i.\\
    &= \frac{\prod_{i \in P'} \mu_i}{\Count{\mu}{\calP'}}
    \end{align*}
    Where the last equality follows from \cref{eqn:phi_count}.
\end{proof}

\begin{theorem}
    \label{thm:rejaction_sampling}
    Given an instance $\calI$ of a sortition problem and weights $\mu_{>0} \in \R^N$.
    \textsc{MaxEntropy} samples a panel according to $\lambda_\mu$, that is,
    \begin{align*}
        \Pr\sqbrac{\text{\cref{alg:max_entropy} outputs } P} =
        \frac{\prod_{i\in P}\mu_i}{\Count{\mu}{\calP}},
        \quad \text{ for all } P \in \calP.
    \end{align*}
\end{theorem}
\begin{proof}
    Let $P'_1, P'_2 \dots$ be the sequence of i.i.d samples obtained by sampling using \cref{alg:sampling} on the restricted feature
    set $F'$.
    We define the acceptance probability
    \begin{equation}
        \label{eqn:rejaction_prob}
        \alpha \coloneqq \Pr_{P' \sim \lambda'_\mu}\sqbrac{P' \in \calP}
        = \sum_{P \in \calP} \lambda'_\mu(P) 
        = \frac{\Count{\mu}{\calP}}{\Count{\mu}{\calP'}}
    \end{equation}
    We assume that $\calP \neq \varnothing$ so $\alpha > 0$ and
    the algorithm terminates with high probability.
    Let $T \coloneqq \min\setst*{t \ge 1}{P'_t \in \calP}$ be the
    stopping time, then, \cref{alg:max_entropy} outputs $P'_T$.
    
    Fix $P \in \calP$. Then, since the samples are i.i.d, we have 
    that
    \begin{align*}
        \Pr\sqbrac{P'_T = P} &= \sum_{t \geq 1}
        \Pr\sqbrac{P'_T = P \wedge T = t}\\
        &= \sum_{t \geq 1}\Pr\sqbrac{
        P'_t =P,  P'_1 \notin \calP, \dots, P'_{t-1} \notin \calP},
        \intertext{since the samples are i.i.d, we have that}
        &= \sum_{t \geq 1} \Pr\sqbrac{P'_t = P}
        \prod_{s = 1}^{t - 1}\Pr\sqbrac{P'_s \notin \calP}\\
        &= \lambda'_\mu(P)\sum_{t\geq 1}(1 - \alpha)^{t-1}.
    \end{align*}
    By the convergence of a geometric series,
    $\sum_{t\geq 1}(1 - \alpha)^{t-1} = 1/\alpha$. 
    Thus, bu \eqref{eqn:rejaction_prob}, we obtain that $\Pr\sqbrac{P'_T = P} = \lambda_\mu(P)$.
\end{proof}

\section{Deferred Details for \nameref{sec:properties}}

\subsection{\nameref{sec:manip}}
\label{apx:manip}
On this section, we present the formal statements and deferred proofs from \cref{sec:manip}.

Fix a sortition instance $\calI$. A \emph{coalition} is a set $C \subseteq N$ of pool members who jointly and strategically misreport their feature-value vectors.
At all times, we assume that the coalition has full knowledge of the selection algorithm and pool composition. 
We denote the \emph{reported feature-value vectors} by $\tilde{w}$, where $\tilde{w}(i) = w(i)$ for all $i \in N \setminus C$ and $\tilde{w}(i) \in \curly{0, 1}^\FV$ for all $i \in C$.
The resulting \emph{manipulated instance} $\tilde{\calI}$  is characterized by $(\calI, C, \restr{\tilde{w}}{C})$.
Let $\overline{\calW}_C \in (\{0, 1\}^{\FV})^{C}$ be the set of all possible joint misreports for the coalition $C$. Given the size of the coalition $c \in \mathbb{N}$ and a sortition algorithm $\calA$,  we formally define the 3 measures of manipulability:
\begin{align*}
    \Manipi(\calI, c, \mathcal{})
    &\coloneqq
    \max_{\substack{C \subseteq N\\ \card{C} = c}}
    \max_{\restr{\tilde{w}}{C} \in \FVspace_C}
    \max_{i \in C}
    \tilde{\pi}(i) - \pi(i),\\
    \Manipe(\calI, c, \calA)
    &\coloneqq
    \max_{\substack{C \subseteq N\\ \card{C} = c}}
    \max_{\restr{\tilde{w}}{C} \in \FVspace_C}
    \max_{i \notin C}
    \pi(i) - \tilde{\pi}(i), \text{ and }\\
    \Manipc(\calI, c, \calA)
    &\coloneqq
    \max_{\substack{C \subseteq N\\ \card{C} = c}}
    \max_{\restr{\tilde{w}}{C} \in \FVspace_C}
    \max_{f, v \in \FV}
    \sum_{\substack{i \in N\\ f(i) = v}}
    {\tilde{\pi}(i) - \pi(i)},
\end{align*}
where $\pi$ denotes the selection probabilities for feature vectors $w$ and $\tilde{\pi}$ those for $\tilde{w}$ under $\calA$, in our case \textsc{MaxEntropy}.

Without additional assumptions, there is no meaningful bound on manipulation. Consider $n$ pool members and a single feature \emph{Region = \{Urban, Rural\}}, where the quotas require $k - 1$ urban panel members and $1$ rural one.
If the pool contains $k$ urban pool members (and $n-k$ rural ones), any algorithm that satisfies equal treatment of equals will assign selection probability $1 - 1/k$ to urban pool members and $1/(n - k)$ to rural ones.
As $n$ grows, the probability of a rural pool member being selected goes to zero. Thus, if any rural pool members misreport their value, they increase their selection probability by a constant arbitrarily close to 1.

To rule out such trivial manipulability, \citet{FlaniganLPW24} say that an instance is $\kappa$-\emph{rich} if for some constant
$\kappa > 0$ if there exists a subset $W^*$ of the feature-value vectors $\calW$ present in the pool such that
\begin{enumerate}[i)]
    \item $\big|\restr{N}{w}\big| \geq \kappa n + k$ for all $w \in W^*$; and 
    \item there exists a panel  $P_0 \in \restr{\calP}{W^*}$,
\end{enumerate}

where $\restr{\calP}{W^*}$ is the set of panels only containing members in $\restr{N}{W^*}$. Under this assumption, we have the following bound on the selection probabilities

\begin{lemma}
    \label{lem:manip_rich_prb_ub}
    If $\calI$ is $\kappa$-rich, then for all $i \in N$, we have that $\pi_{\unifD}(i) \leq \frac{k}{\kappa^k n}$.
\end{lemma}
\begin{proof}
    We first lower bound the number of feasible panels exclusively containing members with feature-value vectors in $W^*$.
    Since $\calI$ is $\kappa$-rich, there exists $P_0 \in \restr{\calP}{W^*}$.
    For each $w \in W^*$, let $p_0(w)$ be the number of members of $P_0$ with feature-value vector $w$, so $\sum_{w \in W^*} p_0(w)=k$.
    Since $\card{\restr{N}{w}} \geq \kappa n + k$, we obtain
    \begin{align}
        \label{eqn:kappa_rich_sets}
        \card{\restr{\calP}{W^*}}
        \geq \prod_{w \in W^*}
        \binom{\kappa n + k}{p_0(w)}
        = \prod_{w \in W^*}
        \frac{\prod_{i = 0}^{p_0(w) - 1}(\kappa n + k - i)}{p_0(w)!}
        \geq \frac{(\kappa n)^k}{\prod_{w \in W^*} p_0(w)!}
        \geq \frac{(\kappa n)^k}{k!},
    \end{align}
    where the last inequality follows from $(a + b)! \geq a!b!$ for $a, b \in \Z_{\geq 0}$.
    For each $i \in N$, the number of panels containing $i$ is at most $\binom{n-1}{k-1}$. 
    By \cref{eqn:kappa_rich_sets}, we have that
    \begin{equation*}
        \pi_{\unifD}(i)= \frac{\card{\calP(i)}}{\card{\calP}}
        \leq \frac{\binom{n-1}{k-1}}
        {\card{\restr{\calP}{W^*}}}
        \leq \frac{n^{k-1}}{(k-1)!}\frac{k!}{(\kappa n)^k}
        = \frac{k}{\kappa^k n}. \qedhere
    \end{equation*}
\end{proof}    
To bound $\Manipc$, we use the following observation due to \citet{BF24}.

\begin{lemma}
    \label{lem:manipc_bound}
    Fix a coalition $C \subseteq N$ of size $c$.
    Then, for any $f, v \in \FV$ and $\lambda \in \Delta(\calP)$, we have that
    \begin{align*}
        \sum_{\substack{i \in N:\\ f(i) = v}} \tilde{\pi}_\lambda(i) 
        - \pi_\lambda(i) \leq \paren*{u_{f,v} - \ell_{f,v}}
        + c \max\paren*{\tilde{\pi}_\lambda}.
    \end{align*}
\end{lemma}
\begin{proof}
    Fix a feature-value pair $f, v \in \FV$.
    Define $N_{f,v} \coloneqq \setst{i \in N}{f(i) = v}$ and $D_{f, v} \coloneqq \setst{i \in N_{f, v}}{\tilde{f}(i) \neq v}$.
    By feasibility, for any panel $P \in \calP$, we have that, $\ell_{f, v} \leq \card*{N_{f, v} \cap P} \leq u_{f, v}$. 
    Thus, 
    \begin{align*}
        \ell_{f, v} \leq \sum_{i \in N_{f, v}} \pi(i) = 
        \sum_{P \in \calP} \lambda(P) \card{N_{f, v} \cap P} \leq u_{f,v}.
    \end{align*}
    The analogous result holds for $\tilde{\pi}$ over the set of reported feature-value vectors $N_{f,v} \setminus D_{f, v}$ , as the manipulated instance is defined over the same quotas.
    We split the total post manipulation expected difference in assigned seats to truthful members and misreported members.
    By the expression above, we have that
    \begin{align*}
        \sum_{i \in N_{f, v}}\tilde{\pi}(i)-\pi(i)
        = \sum_{i \in N_{f, v} \setminus D_{f, v}}
        \tilde{\pi}(i) +  
        \sum_{i \in D_{f, v}} \tilde{\pi}(i) - \sum_{i \in N_{f, v}} \pi(i)\\
        \leq u_{f, v} - \ell_{f, v} + \sum_{i \in D_{f, v}} 
        \tilde{\pi}(i) 
        \leq u_{f, v} - \ell_{f, v} + c \max{\tilde{\pi}}.
    \end{align*}
\end{proof}

This leads to the following theorem
\begin{theorem}[Resistance to manipulation, formal statement]
    \label{thm:manip_uniform}
    Let $\calI$ be a $\kappa$-rich instance and let
    $c \leq \kappa' n$ for some constant
    $\kappa' \in [0,\kappa)$. Then, \textsc{MaxEntropy} satisfies the
    following:
    \begin{align*}
        \Manipi(\calI, c) &= O\paren*{k/n},\\
        \Manipe(\calI, c) &= O\paren*{k/n},\text{ and }\\ 
        \Manipc(\calI, c) &= \max_{f, v \in \FV}\paren*{u_{f, v} -
        \ell_{f, v}} + O(ck/n).
    \end{align*}
\end{theorem}
\begin{proof}
    Fix any coalition $C$ with $\card{C} = c$ and reported feature-value vectors $\tilde{w}$.
    Let $\beta \coloneqq \kappa - \kappa'$. 
    Then, for any $w \in W^*$, we have that, $\card{\restr{\tilde{N}}{w}} \geq \card{\restr{N}{w}} - c \geq \beta n + k$.
    Let $P_0$ be the panel from the $\kappa$-rich definition and $p_0(w)$ be the number of members of $P_0$ with feature-value vector $w$. 
    Since $p_0(w) \leq k$ for all $w \in W^*$ and $\card{\restr{\tilde{N}}{w}} \geq k$, we can construct a panel $\tilde{P}_0$ in the manipulated instance by selecting $p_0(w)$ distinct members from each $\tilde{N}(w)$. 
    Thus, $\tilde{\calI}$ is $\beta$-rich.

    Now, by \cref{lem:manip_rich_prb_ub}, $\tilde{\pi}(i) \leq {k}/{\beta^k n}$ for all $i \in N$. Thus, for any $i \in C$, we have that 
    $\tilde{\pi}(i) - \pi(i) \leq \tilde{\pi}(i) 
    \leq {k}/{\beta^k n}$, which proves the stated bound for $\Manipi$.
    The bound for $\Manipe$ follows analogously from \cref{lem:manip_rich_prb_ub} by imposing that $\pi(i) \leq {k}/{\kappa^kn}.$

    Finally, the bound for $\Manipc$ follows from taking the maximum over all feature-value pairs $f, v \in \FV$ and applying \cref{lem:manip_rich_prb_ub,lem:manipc_bound}. 
\end{proof}
\cref{thm:manip_uniform}  is asymptotically tight when taking $n \to \infty$, as it matches the lower bound obtained in \citet[][Thm. 4.3]{FlaniganLPW24} for $\Manipi$ and $\Manipe$.
Their bound also implies the optimality of $\Manipc$ whenever the quotas are tight.

\section{Deferred Details for \nameref{sec:empirics}}

\subsection{Instances}
\label{apx:instances}

\begingroup
\scriptsize
\setlength{\tabcolsep}{3pt}

\begin{longtable}{lrrrccccccc}
\caption{List of instances in our dataset. The target panel size is the last subscript in the instance name.
The runs were done in parallel on a cloud machine with $64$ GB of memory and $32$ cores.
Runs marked with an * had a 90-minute timeout. “NO” means that the experiment timed out or failed.}\label{tab:instances}\\

\toprule
 &  &  &  &  & \multicolumn{2}{c}{\smash{\textsc{MaxEntropy}}} &
 \multicolumn{2}{c}{\smash{\textsc{LexiMin} fairness}} &
 \multicolumn{2}{c}{\smash{\textsc{Goldilocks} fairness}} \\
\cmidrule(lr){6-7} \cmidrule(lr){8-9} \cmidrule(lr){10-11}
\textbf{Instance} & pool & features & values & \smash{\textsc{Legacy}} &
single sample & distribution &
col.\ generation & max.\ entropy &
col.\ generation & max.\ entropy \\
\midrule
\endfirsthead

\toprule
 &  &  &  &  & \multicolumn{2}{c}{\smash{\textsc{MaxEntropy}}} &
 \multicolumn{2}{c}{\smash{\textsc{LexiMin} fairness}} &
 \multicolumn{2}{c}{\smash{\textsc{Goldilocks} fairness}} \\
\cmidrule(lr){6-7} \cmidrule(lr){8-9} \cmidrule(lr){10-11}
\textbf{Instance} & pool & features & values & \smash{\textsc{Legacy}} &
single sample & distribution &
col.\ generation & max.\ entropy &
col.\ generation & max.\ entropy \\
\midrule
\endhead

\bottomrule
\endfoot

\bottomrule
\endlastfoot

cca\_75 & 825 & 4 & 66 & \textsc{no} & \textsc{yes} & \textsc{no} & \textsc{yes} & \textsc{no} & \textsc{no} & \textsc{no} \\
hd\_30 & 239 & 7 & 33 & \textsc{yes} & \textsc{yes} & \textsc{yes} & \textsc{yes} & \textsc{no} & \textsc{yes} & \textsc{no} \\
mass\_a\_24 & 70 & 5 & 11 & \textsc{yes} & \textsc{yes} & \textsc{yes} & \textsc{yes} & \textsc{yes} & \textsc{yes} & \textsc{yes} \\
mass\_b\_37 & 349 & 8 & 32 & \textsc{yes} & \textsc{yes} & \textsc{yes} & \textsc{yes} & \textsc{yes} & \textsc{yes} & \textsc{yes} \\
mass\_c\_35 & 118 & 5 & 13 & \textsc{yes} & \textsc{yes} & \textsc{yes} & \textsc{yes} & \textsc{yes} & \textsc{yes} & \textsc{yes} \\
mass\_d\_31 & 168 & 9 & 36 & \textsc{no} & \textsc{yes} & \textsc{yes} & \textsc{yes} & \textsc{no} & \textsc{yes} & \textsc{no} \\
mass\_e\_36 & 322 & 4 & 13 & \textsc{yes} & \textsc{yes} & \textsc{yes} & \textsc{yes} & \textsc{yes} & \textsc{yes} & \textsc{yes} \\
mass\_f\_33 & 158 & 3 & 9 & \textsc{yes} & \textsc{yes} & \textsc{yes} & \textsc{yes} & \textsc{yes} & \textsc{yes} & \textsc{yes} \\
mass\_g\_35 & 379 & 8 & 33 & \textsc{yes} & \textsc{yes} & \textsc{yes} & \textsc{yes} & \textsc{yes} & \textsc{yes} & \textsc{yes} \\
mass\_h\_36 & 204 & 5 & 14 & \textsc{yes} & \textsc{yes} & \textsc{yes} & \textsc{yes} & \textsc{yes} & \textsc{yes} & \textsc{yes} \\
mass\_i\_36 & 106 & 8 & 30 & \textsc{yes} & \textsc{yes} & \textsc{yes} & \textsc{yes} & \textsc{yes} & \textsc{yes} & \textsc{yes} \\
mass\_j\_28 & 421 & 7 & 22 & \textsc{yes} & \textsc{yes} & \textsc{yes} & \textsc{yes} & \textsc{yes} & \textsc{yes} & \textsc{yes} \\
mass\_k\_36 & 57 & 7 & 17 & \textsc{yes} & \textsc{yes} & \textsc{yes} & \textsc{yes} & \textsc{yes} & \textsc{yes} & \textsc{yes} \\
mass\_l\_34 & 133 & 5 & 29 & \textsc{yes} & \textsc{yes} & \textsc{yes} & \textsc{yes} & \textsc{yes} & \textsc{yes} & \textsc{yes} \\
mass\_m\_26 & 72 & 13 & 33 & \textsc{no} & \textsc{yes} & \textsc{yes} & \textsc{yes} & \textsc{yes} & \textsc{yes} & \textsc{yes} \\
mass\_n\_32 & 1098 & 8 & 23 & \textsc{yes} & \textsc{yes} & \textsc{yes} & \textsc{yes} & \textsc{yes} & \textsc{no} & \textsc{no} \\
mass\_o\_36 & 186 & 6 & 27 & \textsc{yes} & \textsc{yes} & \textsc{yes} & \textsc{yes} & \textsc{yes} & \textsc{yes} & \textsc{yes} \\
mass\_p\_29 & 108 & 7 & 39 & \textsc{yes} & \textsc{yes} & \textsc{yes} & \textsc{yes} & \textsc{yes} & \textsc{yes} & \textsc{yes} \\
mass\_q\_37 & 543 & 10 & 37 & \textsc{yes} & \textsc{yes} & \textsc{no} & \textsc{yes} & \textsc{no} & \textsc{yes} & \textsc{no} \\
mass\_r\_34 & 117 & 10 & 22 & \textsc{yes} & \textsc{yes} & \textsc{yes} & \textsc{yes} & \textsc{yes} & \textsc{yes} & \textsc{yes} \\
mass\_s\_29 & 56 & 10 & 50 & \textsc{yes} & \textsc{yes} & \textsc{yes} & \textsc{yes} & \textsc{yes} & \textsc{yes} & \textsc{yes} \\
mass\_t\_31 & 121 & 6 & 19 & \textsc{yes} & \textsc{yes} & \textsc{yes} & \textsc{yes} & \textsc{yes} & \textsc{yes} & \textsc{yes} \\
mass\_u\_31 & 131 & 6 & 23 & \textsc{yes} & \textsc{yes} & \textsc{yes} & \textsc{yes} & \textsc{yes} & \textsc{yes} & \textsc{yes} \\
mass\_v\_49 & 369 & 7 & 28 & \textsc{yes} & \textsc{yes} & \textsc{yes} & \textsc{yes} & \textsc{yes} & \textsc{no} & \textsc{no} \\
mass\_w\_34 & 148 & 6 & 35 & \textsc{yes} & \textsc{yes} & \textsc{yes} & \textsc{yes} & \textsc{yes} & \textsc{yes} & \textsc{yes} \\
mass\_x\_27 & 233 & 6 & 21 & \textsc{yes} & \textsc{yes} & \textsc{yes} & \textsc{yes} & \textsc{yes} & \textsc{yes} & \textsc{yes} \\
mass\_y\_28 & 161 & 7 & 28 & \textsc{yes} & \textsc{yes} & \textsc{yes} & \textsc{yes} & \textsc{yes} & \textsc{yes} & \textsc{yes} \\
mass\_z\_48 & 498 & 9 & 21 & \textsc{yes} & \textsc{yes} & \textsc{yes} & \textsc{yes} & \textsc{yes} & \textsc{yes} & \textsc{yes} \\
nd\_a\_49 & 248 & 4 & 17 & \textsc{no} & \textsc{yes} & \textsc{yes} & \textsc{yes} & \textsc{yes} & \textsc{yes} & \textsc{yes} \\
nd\_b\_52 & 125 & 3 & 12 & \textsc{yes} & \textsc{yes} & \textsc{yes} & \textsc{yes} & \textsc{yes} & \textsc{yes} & \textsc{yes} \\
nd\_c\_40 & 133 & 5 & 22 & \textsc{yes} & \textsc{yes} & \textsc{yes} & \textsc{yes} & \textsc{yes} & \textsc{yes} & \textsc{yes} \\
nd\_d\_50 & 290 & 3 & 12 & \textsc{yes} & \textsc{yes} & \textsc{yes} & \textsc{yes} & \textsc{yes} & \textsc{no} & \textsc{no} \\
nd\_e\_50 & 417 & 5 & 28 & \textsc{yes} & \textsc{yes} & \textsc{yes} & \textsc{yes} & \textsc{no} & \textsc{yes} & \textsc{no} \\
nexus\_170 & 342 & 5 & 33 & \textsc{yes} & \textsc{yes} & \textsc{yes} & \textsc{yes} & \textsc{no} & \textsc{yes} & \textsc{no} \\
obf\_30 & 321 & 8 & 38 & \textsc{no} & \textsc{no} & \textsc{no} & \textsc{yes} & \textsc{no} & \textsc{no} & \textsc{no} \\
sf\_a\_35 & 312 & 6 & 22 & \textsc{yes} & \textsc{yes} & \textsc{yes} & \textsc{yes} & \textsc{yes} & \textsc{yes} & \textsc{yes} \\
sf\_aa\_13 & 65 & 6 & 30 & \textsc{yes} & \textsc{yes} & \textsc{yes} & \textsc{yes} & \textsc{yes} & \textsc{yes} & \textsc{yes} \\
sf\_ab\_13 & 37 & 6 & 30 & \textsc{yes} & \textsc{yes} & \textsc{yes} & \textsc{yes} & \textsc{yes} & \textsc{yes} & \textsc{yes} \\
sf\_ac\_28 & 268 & 6 & 22 & \textsc{yes} & \textsc{yes} & \textsc{yes} & \textsc{yes} & \textsc{yes} & \textsc{yes} & \textsc{yes} \\
sf\_ad\_28 & 303 & 6 & 22 & \textsc{yes} & \textsc{yes} & \textsc{yes} & \textsc{yes} & \textsc{yes} & \textsc{no} & \textsc{no} \\
sf\_ae\_32 & 163 & 8 & 40 & \textsc{yes} & \textsc{yes}* & \textsc{no} & \textsc{yes} & \textsc{no} & \textsc{yes} & \textsc{no} \\
sf\_af\_40 & 125 & 7 & 19 & \textsc{yes} & \textsc{yes} & \textsc{yes} & \textsc{yes} & \textsc{yes} & \textsc{yes} & \textsc{yes} \\
sf\_ag\_23 & 122 & 7 & 27 & \textsc{yes} & \textsc{yes} & \textsc{yes} & \textsc{yes} & \textsc{yes} & \textsc{yes} & \textsc{yes} \\
sf\_ah\_42 & 179 & 8 & 34 & \textsc{no} & \textsc{yes} & \textsc{no} & \textsc{yes} & \textsc{no} & \textsc{yes} & \textsc{no} \\
sf\_ai\_20 & 182 & 8 & 41 & \textsc{no} & \textsc{yes} & \textsc{yes} & \textsc{yes} & \textsc{yes} & \textsc{yes} & \textsc{yes} \\
sf\_aj\_25 & 205 & 7 & 33 & \textsc{yes} & \textsc{yes} & \textsc{yes} & \textsc{yes} & \textsc{no} & \textsc{yes} & \textsc{no} \\
sf\_ak\_15 & 58 & 6 & 22 & \textsc{yes} & \textsc{yes} & \textsc{yes} & \textsc{yes} & \textsc{yes} & \textsc{yes} & \textsc{yes} \\
sf\_al\_82 & 515 & 8 & 32 & \textsc{no} & \textsc{no} & \textsc{no} & \textsc{yes} & \textsc{no} & \textsc{no} & \textsc{no} \\
sf\_am\_82 & 533 & 8 & 32 & \textsc{yes} & \textsc{no} & \textsc{no} & \textsc{yes} & \textsc{no} & \textsc{yes} & \textsc{no} \\
sf\_an\_82 & 428 & 8 & 32 & \textsc{no} & \textsc{no} & \textsc{no} & \textsc{yes} & \textsc{no} & \textsc{yes} & \textsc{no} \\
sf\_ao\_82 & 616 & 8 & 32 & \textsc{yes} & \textsc{no} & \textsc{no} & \textsc{yes} & \textsc{no} & \textsc{no} & \textsc{no} \\
sf\_ap\_82 & 548 & 8 & 32 & \textsc{no} & \textsc{no} & \textsc{no} & \textsc{yes} & \textsc{no} & \textsc{no} & \textsc{no} \\
sf\_aq\_82 & 642 & 8 & 32 & \textsc{no} & \textsc{no} & \textsc{no} & \textsc{yes} & \textsc{no} & \textsc{no} & \textsc{no} \\
sf\_ar\_23 & 283 & 8 & 33 & \textsc{no} & \textsc{yes} & \textsc{no} & \textsc{yes} & \textsc{no} & \textsc{no} & \textsc{no} \\
sf\_as\_22 & 248 & 8 & 33 & \textsc{yes} & \textsc{yes} & \textsc{no} & \textsc{yes} & \textsc{no} & \textsc{yes} & \textsc{no} \\
sf\_at\_25 & 360 & 8 & 33 & \textsc{yes} & \textsc{yes} & \textsc{no} & \textsc{yes} & \textsc{no} & \textsc{no} & \textsc{no} \\
sf\_au\_20 & 231 & 8 & 33 & \textsc{yes} & \textsc{yes} & \textsc{no} & \textsc{yes} & \textsc{no} & \textsc{yes} & \textsc{no} \\
sf\_av\_22 & 133 & 8 & 34 & \textsc{yes} & \textsc{yes} & \textsc{yes} & \textsc{yes} & \textsc{no} & \textsc{yes} & \textsc{yes} \\
sf\_aw\_20 & 91 & 8 & 33 & \textsc{yes} & \textsc{yes} & \textsc{yes} & \textsc{yes} & \textsc{no} & \textsc{yes} & \textsc{no} \\
sf\_ax\_25 & 140 & 8 & 33 & \textsc{no} & \textsc{yes} & \textsc{yes} & \textsc{yes} & \textsc{no} & \textsc{yes} & \textsc{no} \\
sf\_ay\_23 & 153 & 8 & 34 & \textsc{yes} & \textsc{yes} & \textsc{no} & \textsc{yes} & \textsc{no} & \textsc{yes} & \textsc{no} \\
sf\_b\_20 & 250 & 6 & 20 & \textsc{yes} & \textsc{yes} & \textsc{yes} & \textsc{yes} & \textsc{yes} & \textsc{yes} & \textsc{yes} \\
sf\_c\_44 & 161 & 7 & 20 & \textsc{yes} & \textsc{yes} & \textsc{yes} & \textsc{yes} & \textsc{yes} & \textsc{yes} & \textsc{yes}* \\
sf\_d\_40 & 404 & 6 & 19 & \textsc{yes} & \textsc{yes} & \textsc{yes} & \textsc{yes} & \textsc{yes} & \textsc{yes} & \textsc{yes} \\
sf\_e\_110 & 1727 & 7 & 31 & \textsc{yes} & \textsc{yes}* & \textsc{no} & \textsc{no} & \textsc{no} & \textsc{no} & \textsc{no} \\
sf\_f\_21 & 126 & 5 & 27 & \textsc{yes} & \textsc{yes} & \textsc{yes} & \textsc{yes} & \textsc{yes} & \textsc{yes} & \textsc{yes} \\
sf\_g\_21 & 152 & 5 & 27 & \textsc{yes} & \textsc{yes} & \textsc{yes} & \textsc{yes} & \textsc{yes} & \textsc{yes} & \textsc{yes} \\
sf\_h\_20 & 129 & 5 & 20 & \textsc{yes} & \textsc{yes} & \textsc{yes} & \textsc{yes} & \textsc{yes} & \textsc{yes} & \textsc{yes} \\
sf\_i\_55 & 362 & 5 & 16 & \textsc{yes} & \textsc{yes} & \textsc{yes} & \textsc{yes} & \textsc{yes} & \textsc{no} & \textsc{no} \\
sf\_j\_40 & 302 & 5 & 21 & \textsc{yes} & \textsc{yes} & \textsc{yes} & \textsc{yes} & \textsc{yes} & \textsc{yes} & \textsc{yes} \\
sf\_k\_30 & 111 & 3 & 15 & \textsc{yes} & \textsc{yes} & \textsc{yes} & \textsc{yes} & \textsc{yes} & \textsc{yes} & \textsc{yes} \\
sf\_l\_23 & 256 & 8 & 34 & \textsc{yes} & \textsc{yes} & \textsc{yes} & \textsc{yes} & \textsc{no} & \textsc{no} & \textsc{no} \\
sf\_m\_20 & 197 & 8 & 34 & \textsc{yes} & \textsc{yes} & \textsc{yes} & \textsc{yes} & \textsc{no} & \textsc{yes} & \textsc{no} \\
sf\_n\_25 & 298 & 8 & 34 & \textsc{yes} & \textsc{yes} & \textsc{yes} & \textsc{yes} & \textsc{no} & \textsc{no} & \textsc{no} \\
sf\_o\_22 & 250 & 8 & 34 & \textsc{yes} & \textsc{yes} & \textsc{yes} & \textsc{yes} & \textsc{no} & \textsc{yes} & \textsc{no} \\
sf\_p\_30 & 85 & 8 & 34 & \textsc{no} & \textsc{yes} & \textsc{yes} & \textsc{yes} & \textsc{no} & \textsc{yes} & \textsc{no} \\
sf\_q\_30 & 151 & 10 & 38 & \textsc{no} & \textsc{yes} & \textsc{no} & \textsc{yes} & \textsc{no} & \textsc{yes} & \textsc{no} \\
sf\_r\_30 & 137 & 10 & 37 & \textsc{no} & \textsc{yes} & \textsc{yes} & \textsc{yes} & \textsc{no} & \textsc{yes} & \textsc{no} \\
sf\_s\_30 & 167 & 10 & 39 & \textsc{no} & \textsc{yes} & \textsc{no} & \textsc{yes} & \textsc{no} & \textsc{yes} & \textsc{no} \\
sf\_t\_30 & 131 & 9 & 36 & \textsc{no} & \textsc{yes} & \textsc{yes} & \textsc{no} & \textsc{no} & \textsc{yes} & \textsc{no} \\
sf\_u\_30 & 114 & 9 & 36 & \textsc{no} & \textsc{yes} & \textsc{yes} & \textsc{yes} & \textsc{yes} & \textsc{yes} & \textsc{yes} \\
sf\_v\_30 & 112 & 6 & 26 & \textsc{yes} & \textsc{yes} & \textsc{yes} & \textsc{yes} & \textsc{yes} & \textsc{yes} & \textsc{yes} \\
sf\_w\_30 & 143 & 7 & 31 & \textsc{no} & \textsc{yes} & \textsc{yes} & \textsc{yes} & \textsc{no} & \textsc{yes} & \textsc{no} \\
sf\_x\_30 & 113 & 10 & 38 & \textsc{no} & \textsc{yes} & \textsc{yes} & \textsc{yes} & \textsc{no} & \textsc{yes} & \textsc{no} \\
sf\_y\_30 & 98 & 9 & 34 & \textsc{no} & \textsc{yes} & \textsc{yes} & \textsc{yes} & \textsc{no} & \textsc{yes} & \textsc{no} \\
sf\_z\_30 & 152 & 10 & 53 & \textsc{no} & \textsc{no} & \textsc{no} & \textsc{yes} & \textsc{no} & \textsc{yes} & \textsc{no} \\
\end{longtable}
\endgroup

\subsection{Details on Timeouts}
\label{apx:timeouts}
The 30 minutes for \textsc{FairMaxEntropy} exclude the generation of marginals using column generation. 
After this, timeouts are handled as follows:
We first run the stage of \textsc{FairMaxEntropy} that coincides with \textsc{MaxEntropy}, i.e., building the dynamic program and drawing uniform samples.
Then, we start the iterative optimization.
If this optimization reaches an $\ell_2$ gradient norm below $0.001$, we treat the optimization as completed and move on to sampling.
If this threshold is not reached after 10 minutes of the iterative optimization, we check if at least two iterations have completed and, if so, move on to sampling, since nontrivial progress towards fairness has been made.
If two iterations have not completed after 10 minutes, or this entire process exceeds the threshold of 30 minutes, we treat the experiment as a timeout and remove the instance from our analysis.

\subsection{\nameref{sec:diversity}}
\label{apx:diveristy}
In this section, we compare the selection algorithms based on the pairwise mutual information between features.
Formally, given a sortition instance $\calI$ and a panel $P \in \calP$, for each pair of features $f, f'$, we define random variables $X_f$ and $X_{f'}$ distributed according to the frequencies of each value within the panel. Then, for each pair of features, we compute the normalized mutual information
\begin{align*}
    \frac{I(X_f; X_{f'})}{H(X_f) + H(X_{f'})}.
\end{align*}
To estimate measuring errors for the sampled selection algorithms, we compute a  $95\%$ confidence interval using the $t$-distribution.
The metrics were computed across the $42$ instances that were solved within the timeout for all selection algorithms (see \cref{tab:instances}).
Across the board, the plots reinforce the conclusions from \cref{sec:diversity}, namely that column-generation-based selection algorithms increase the correlation between features.

\begin{figure}[tbh]
  \centering
  \begin{subfigure}[t]{0.32\textwidth}
    \centering
    \includegraphics[width=\linewidth]{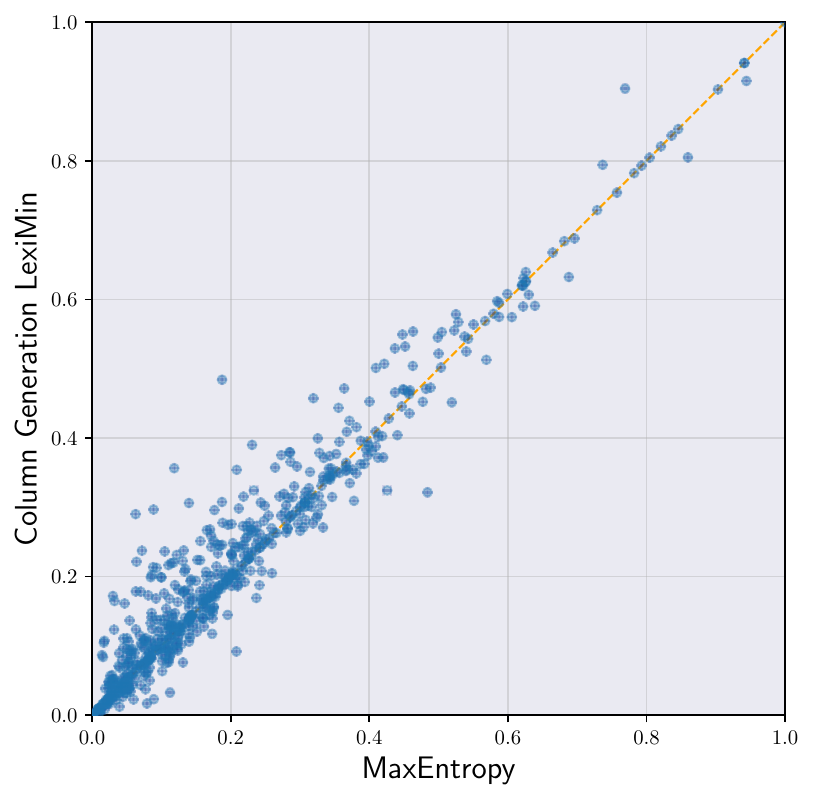}
  \end{subfigure}\hfill
  \begin{subfigure}[t]{0.32\textwidth}
    \centering
    \includegraphics[width=\linewidth]{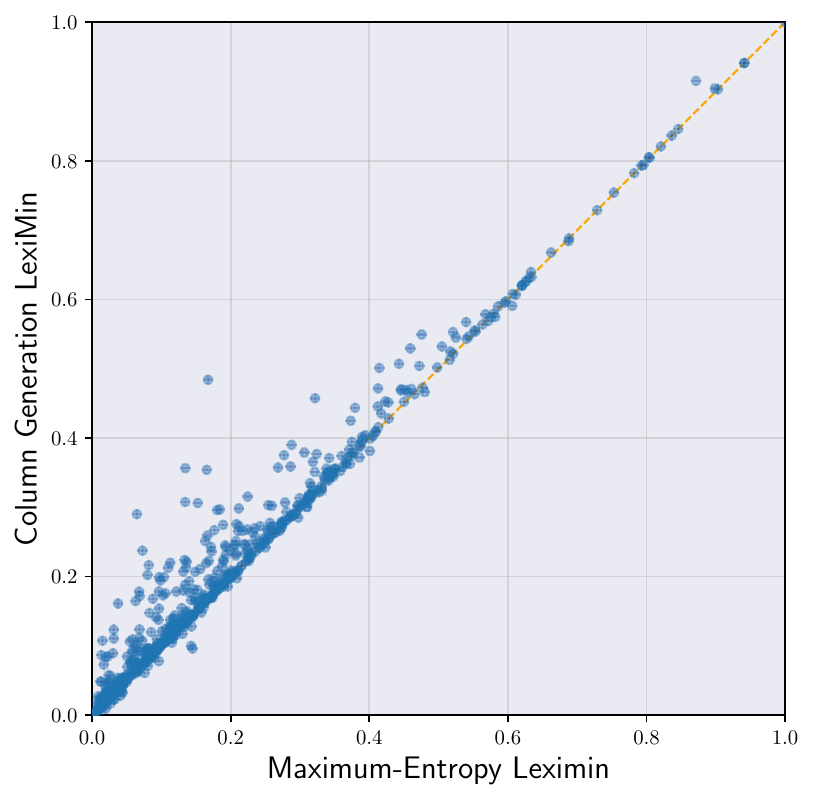}
  \end{subfigure}
  \begin{subfigure}[t]{0.32\textwidth}
    \centering
    \includegraphics[width=\linewidth]{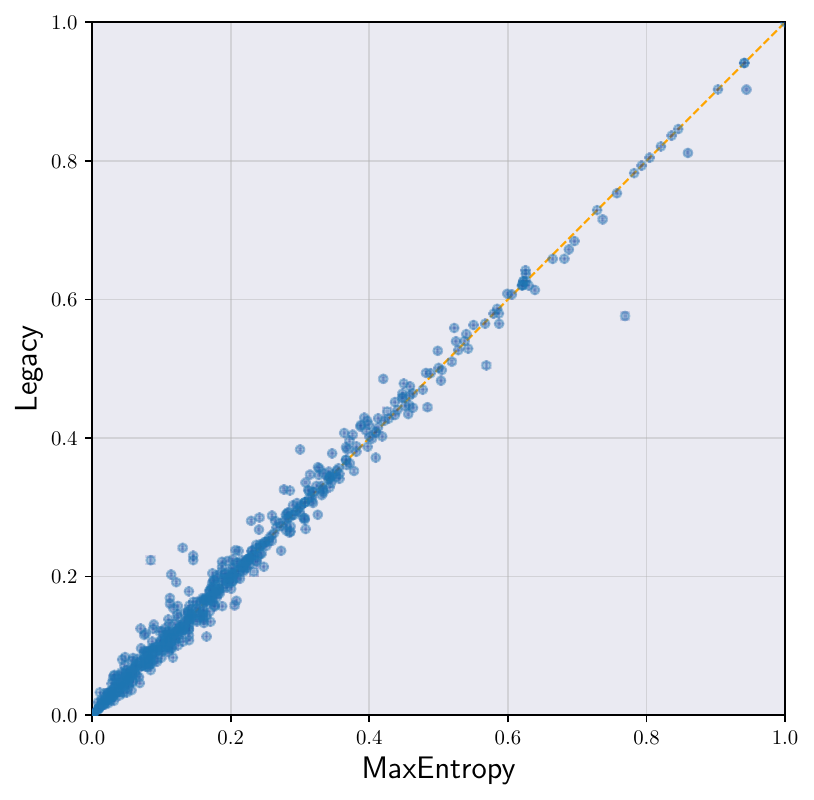}
   \end{subfigure}
\end{figure}

\begin{figure}[h!]
  \centering
  \begin{subfigure}[t]{0.32\textwidth}
    \centering
    \includegraphics[width=\linewidth]{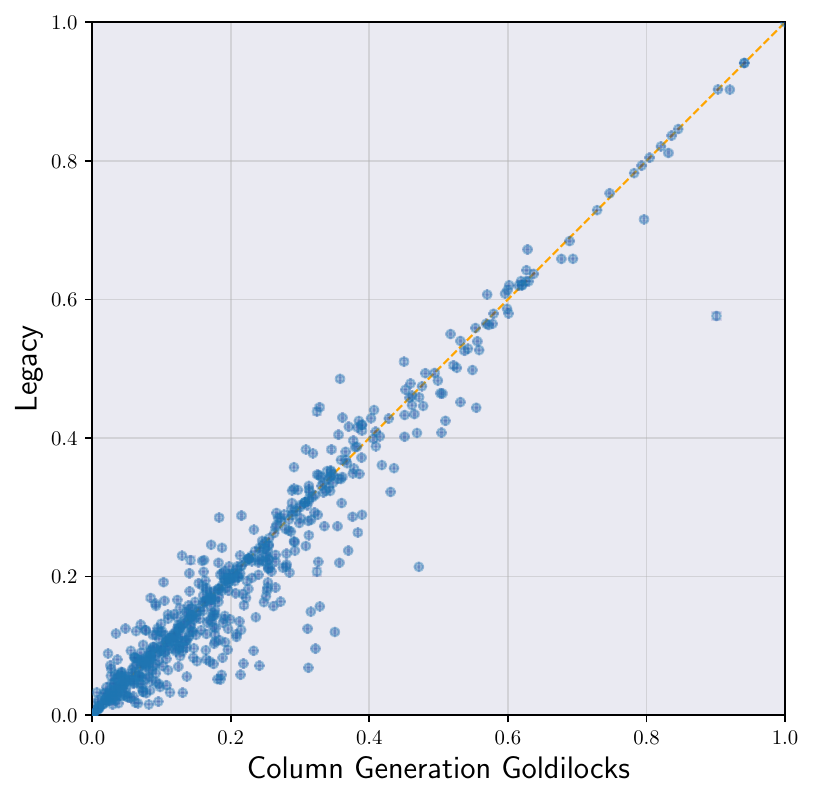}
  \end{subfigure}\hfill
  \begin{subfigure}[t]{0.32\textwidth}
    \centering
    \includegraphics[width=\linewidth]{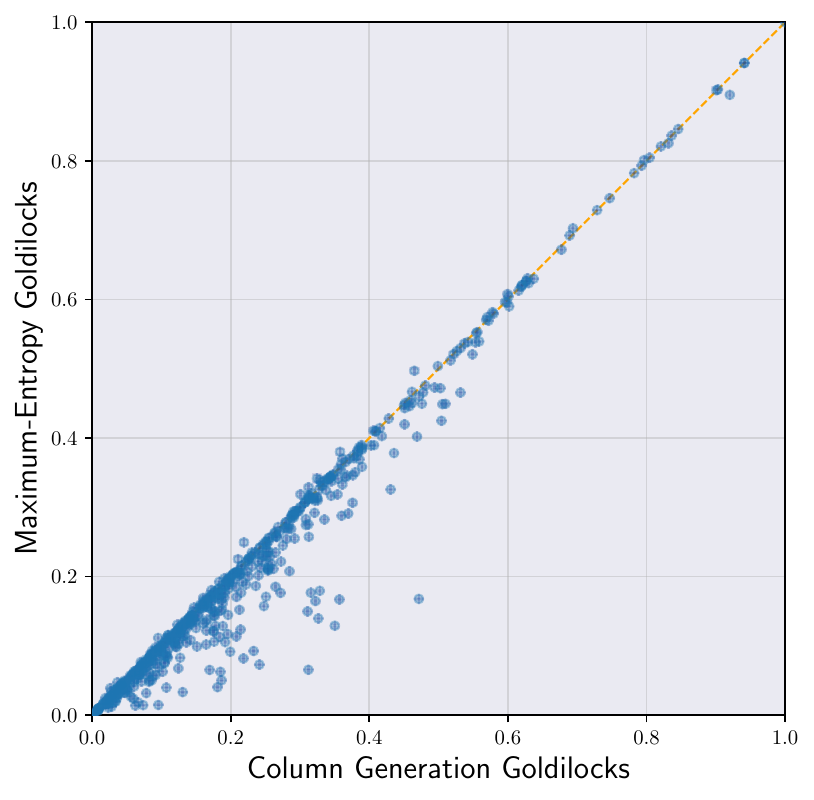}
  \end{subfigure}
  \begin{subfigure}[t]{0.32\textwidth}
    \centering
    \includegraphics[width=\linewidth]{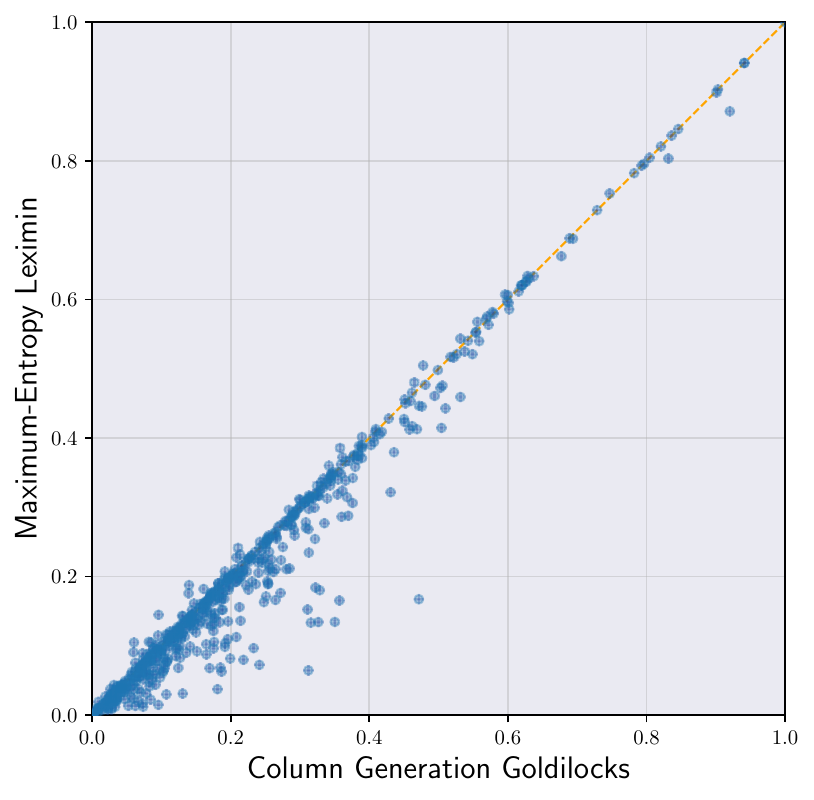}
   \end{subfigure}
\end{figure}

\begin{figure}[h!]
  \centering
  \begin{subfigure}[t]{0.32\textwidth}
    \centering
    \includegraphics[width=\linewidth]{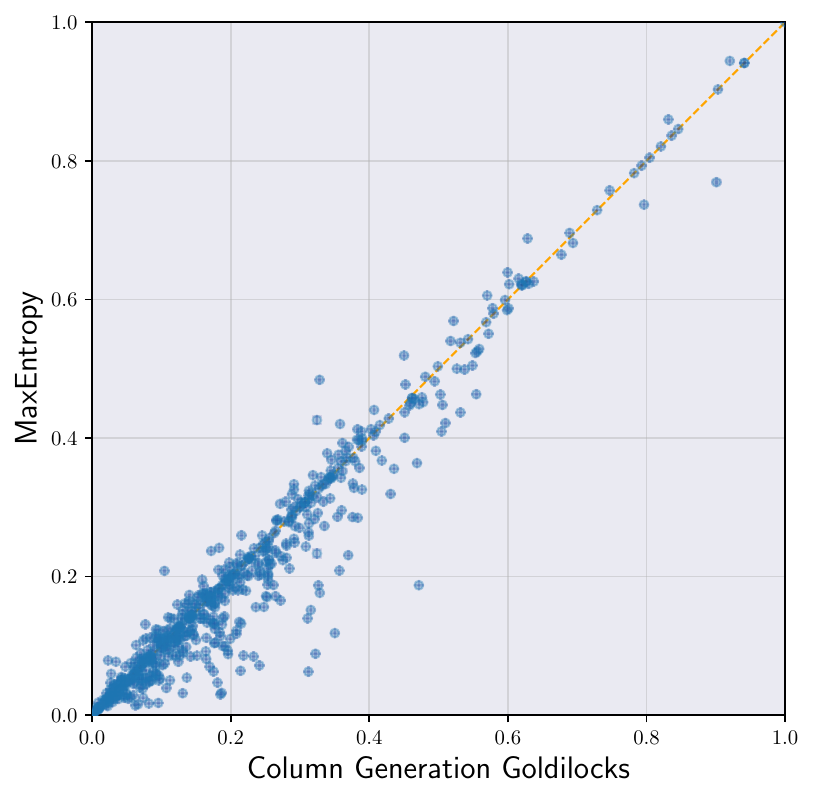}
  \end{subfigure}\hfill
  \begin{subfigure}[t]{0.32\textwidth}
    \centering
    \includegraphics[width=\linewidth]{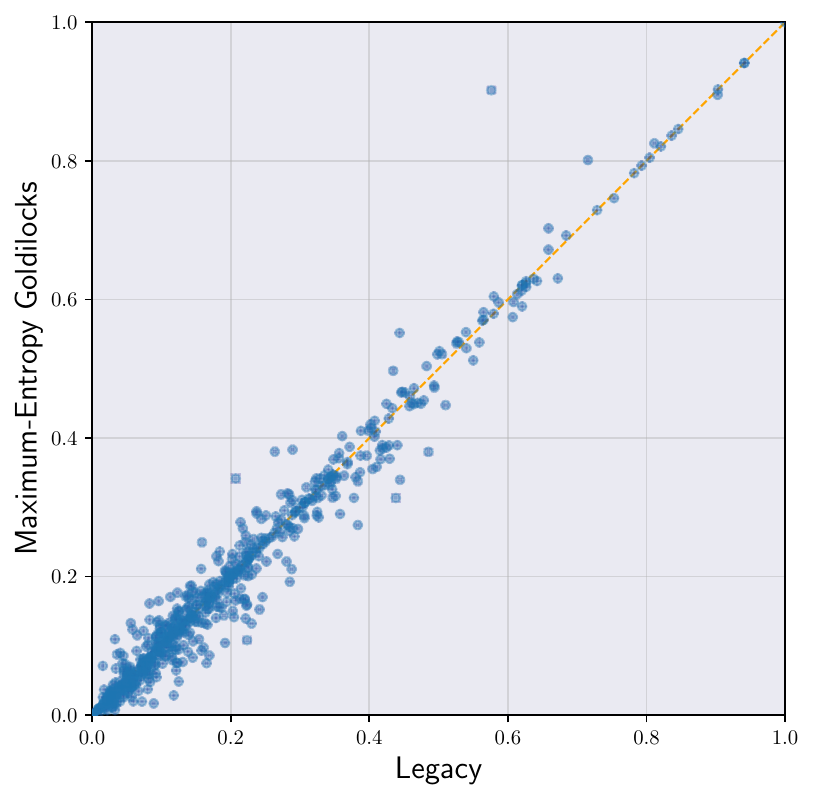}
  \end{subfigure}
  \begin{subfigure}[t]{0.32\textwidth}
    \centering
    \includegraphics[width=\linewidth]{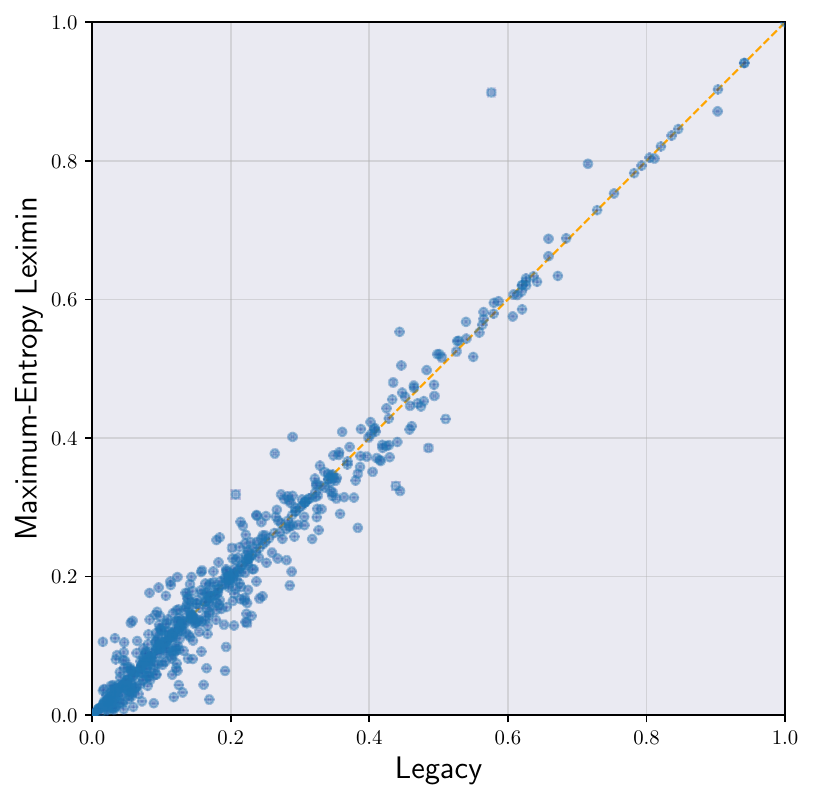}
   \end{subfigure}
\end{figure}

\begin{figure}[h!]
  \centering
  \begin{subfigure}[t]{0.32\textwidth}
    \centering
    \includegraphics[width=\linewidth]{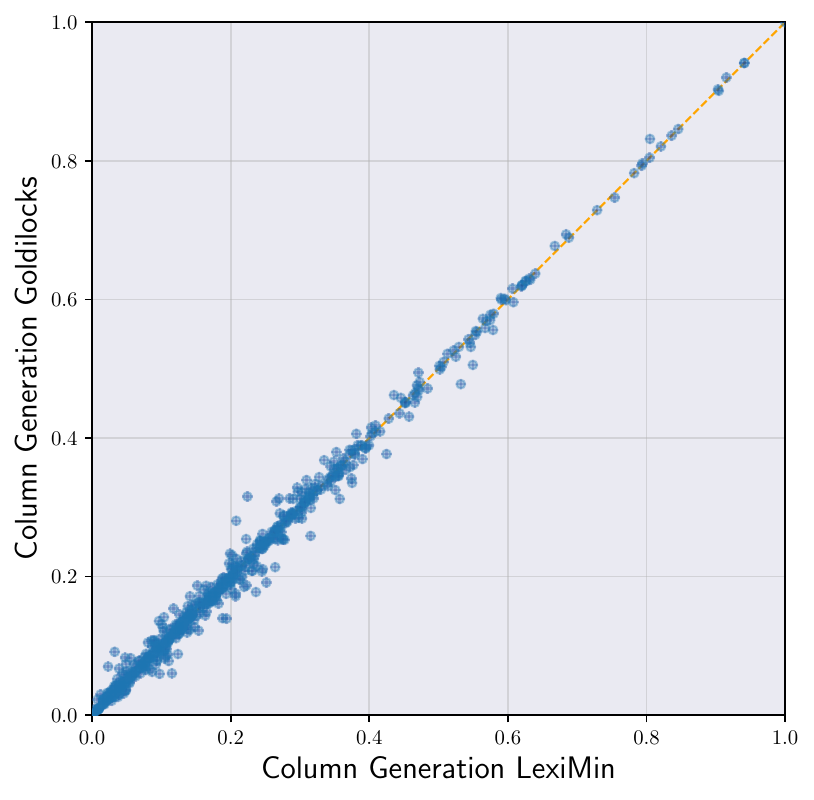}
  \end{subfigure}\hfill
  \begin{subfigure}[t]{0.32\textwidth}
    \centering
    \includegraphics[width=\linewidth]{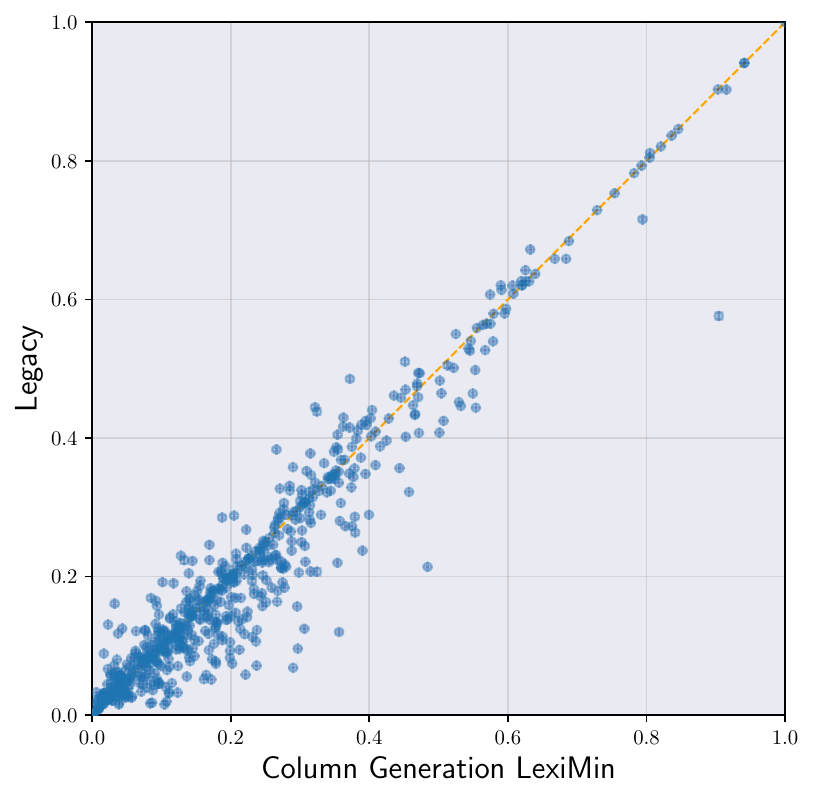}
  \end{subfigure}
  \begin{subfigure}[t]{0.32\textwidth}
    \centering
    \includegraphics[width=\linewidth]{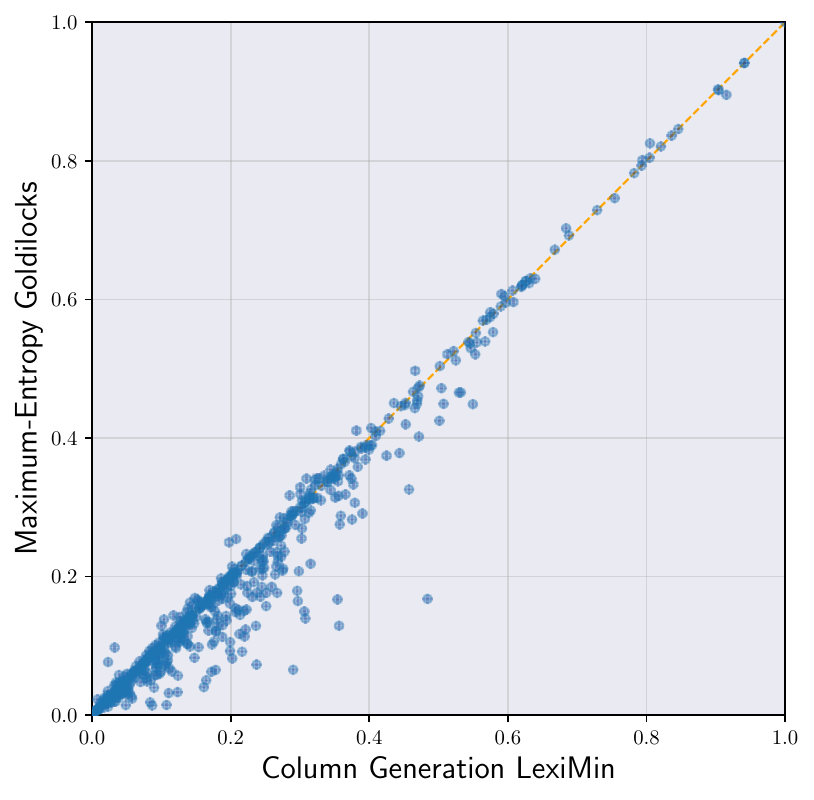}
   \end{subfigure}
\end{figure}

\begin{figure}[h!]
  \centering
  \begin{subfigure}[t]{0.32\textwidth}
    \centering
    \includegraphics[width=\linewidth]{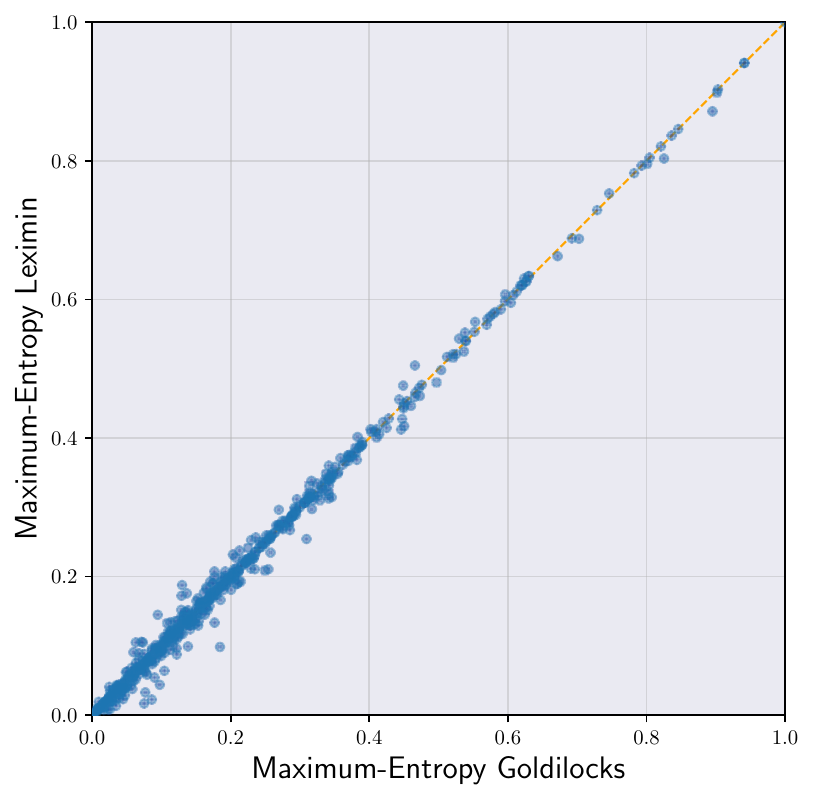}
  \end{subfigure}\hfill
  \begin{subfigure}[t]{0.32\textwidth}
    \centering
    \includegraphics[width=\linewidth]{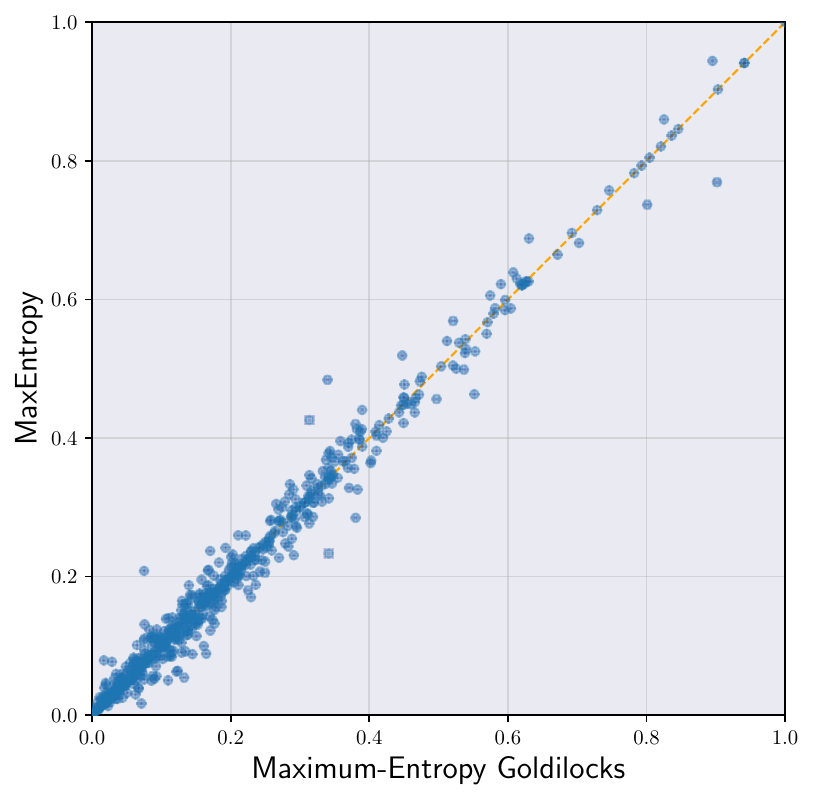}
  \end{subfigure}
  \begin{subfigure}[t]{0.32\textwidth}
    \centering
    \includegraphics[width=\linewidth]{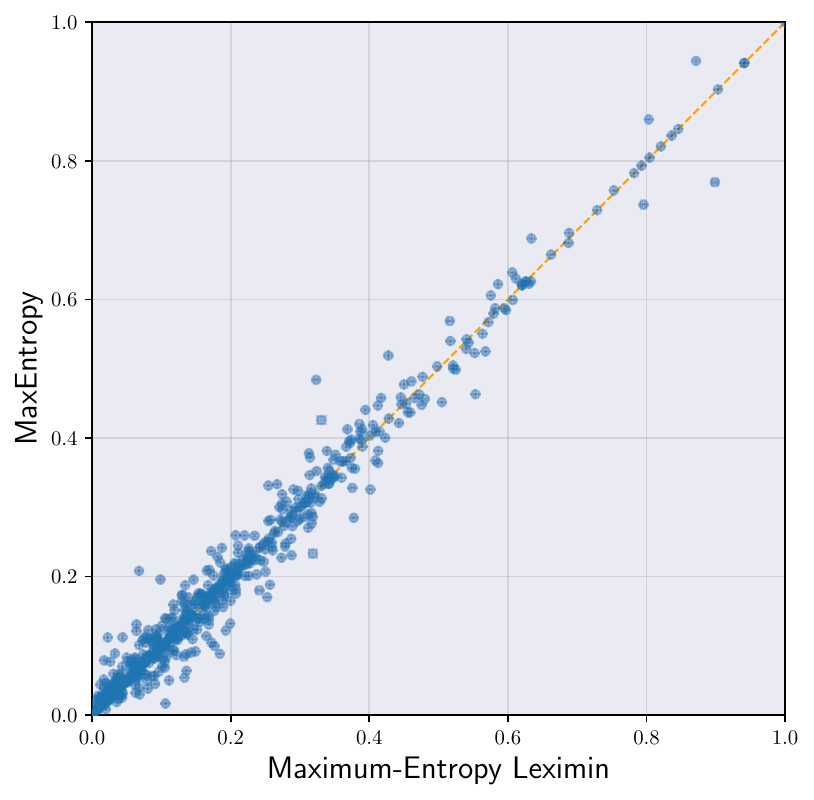}
   \end{subfigure}
\end{figure}

\clearpage

\subsection{\nameref{sec:generalization}}

\label{apx:generalization}
\begin{figure}[h!]
  \centering
  \begin{subfigure}[t]{0.49\textwidth}
    \centering
    \includegraphics[width=\linewidth]{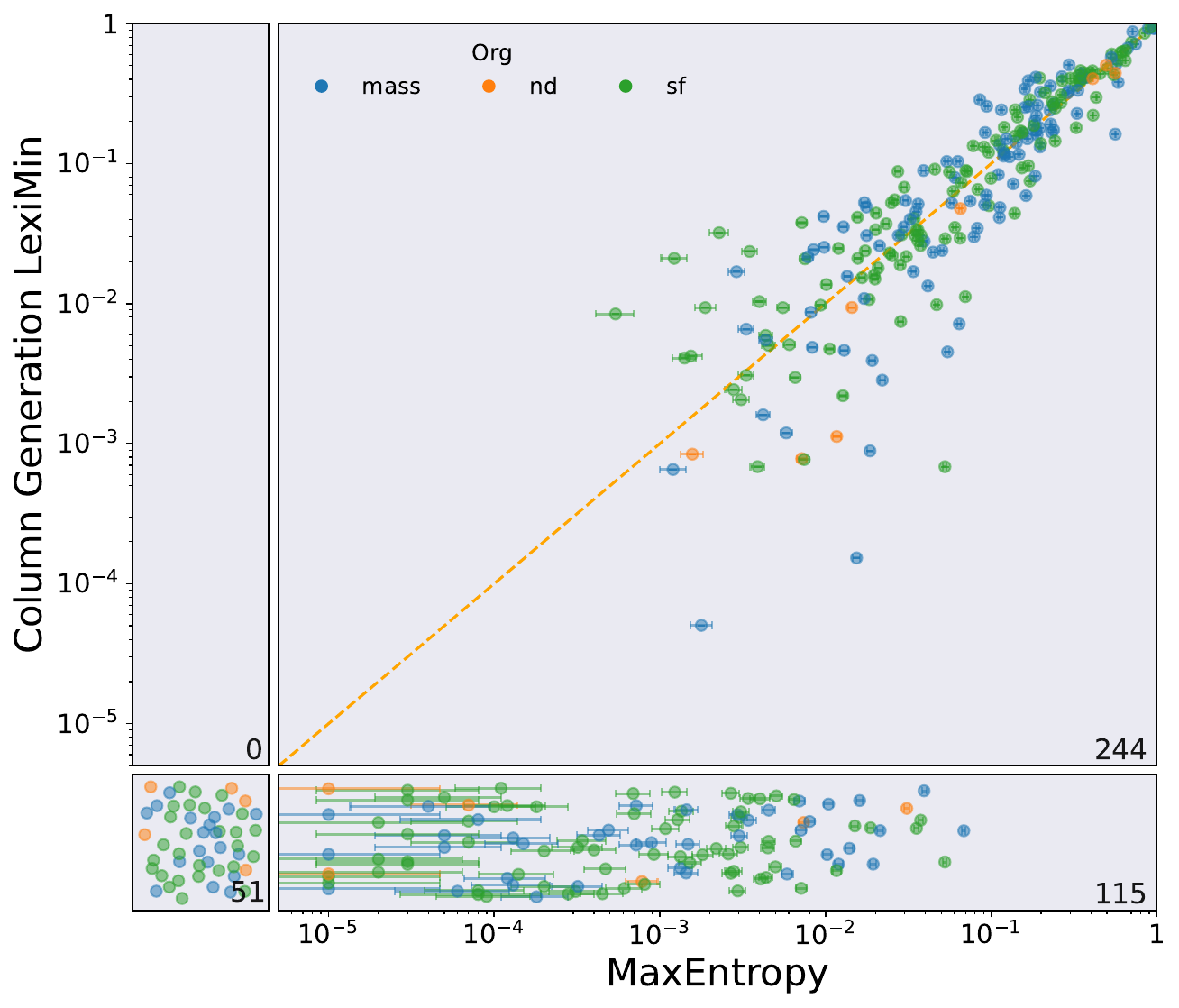}
    \label{fig:left}
  \end{subfigure}\hfill
  \begin{subfigure}[t]{0.49\textwidth}
    \centering
    \includegraphics[width=\linewidth]{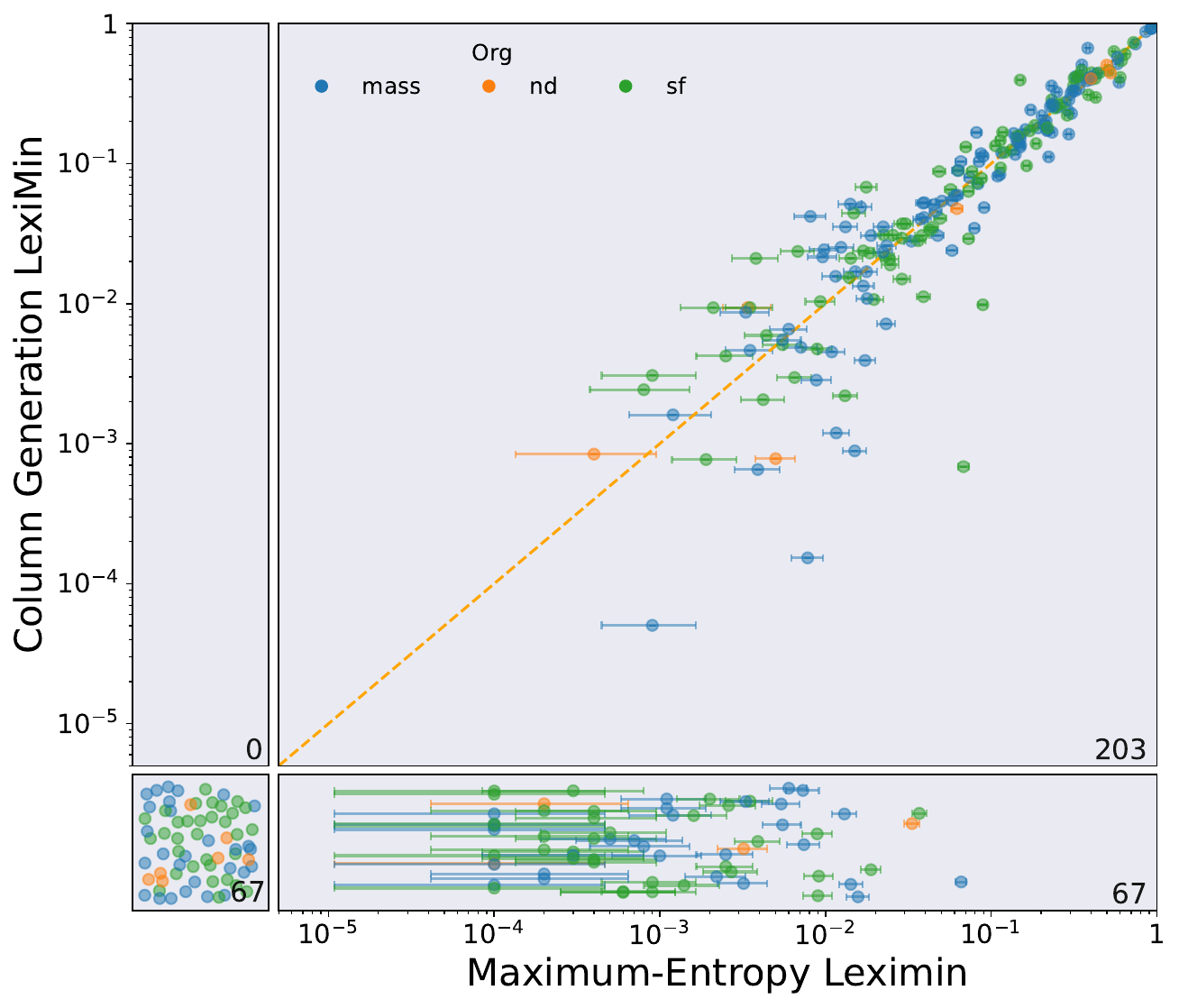}
    \label{fig:right}
  \end{subfigure}
  \caption{Generalization probabilities colored by organization.}
  \label{fig:generalization_org}
\end{figure}

\label{apx:generalization}

\begin{figure}[h!]
  \centering
  \begin{subfigure}[t]{0.49\textwidth}
    \centering
    \includegraphics[width=\linewidth]{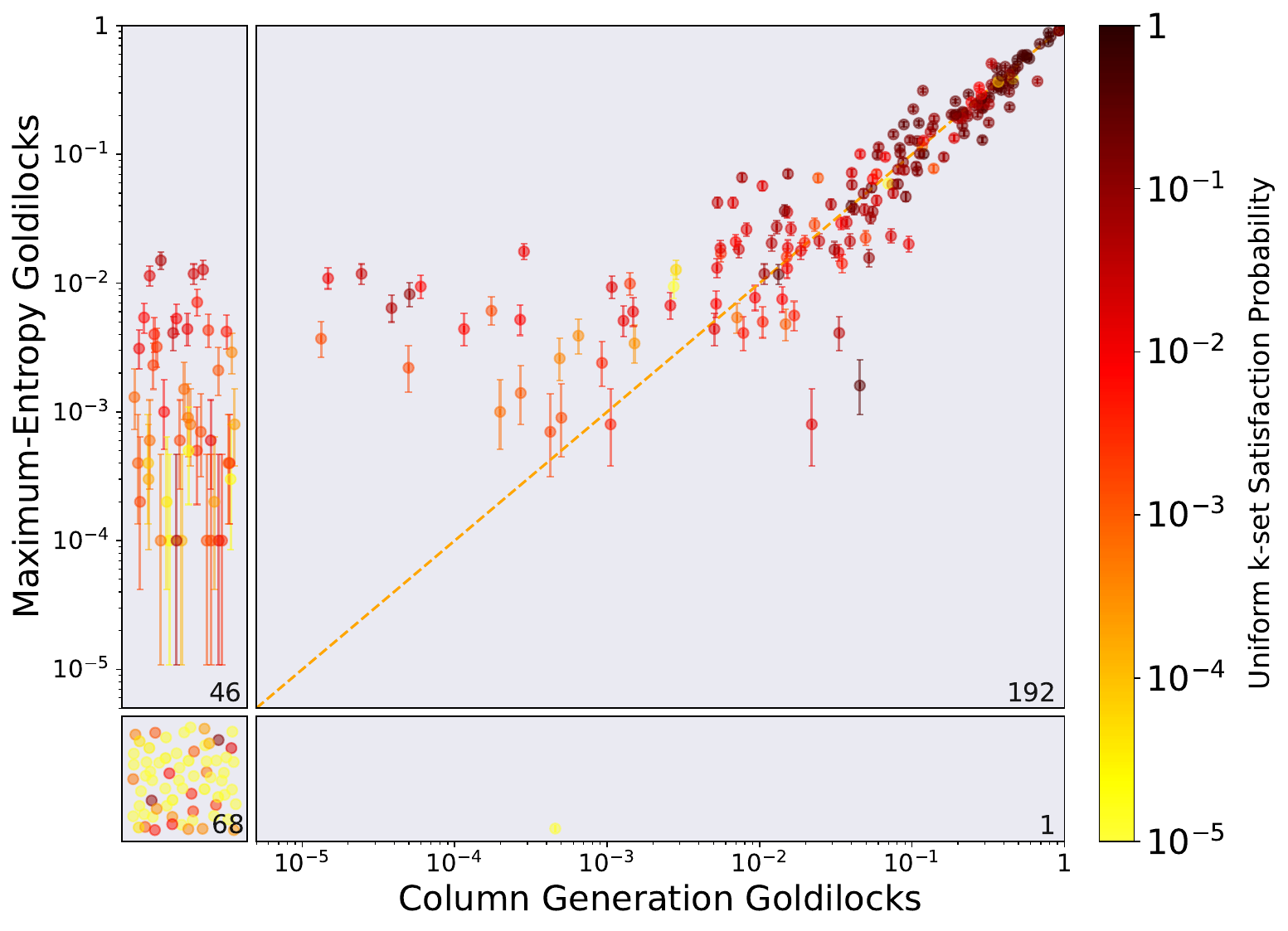}
  \end{subfigure}\hfill
  \begin{subfigure}[t]{0.49\textwidth}
    \centering
    \includegraphics[width=\linewidth]{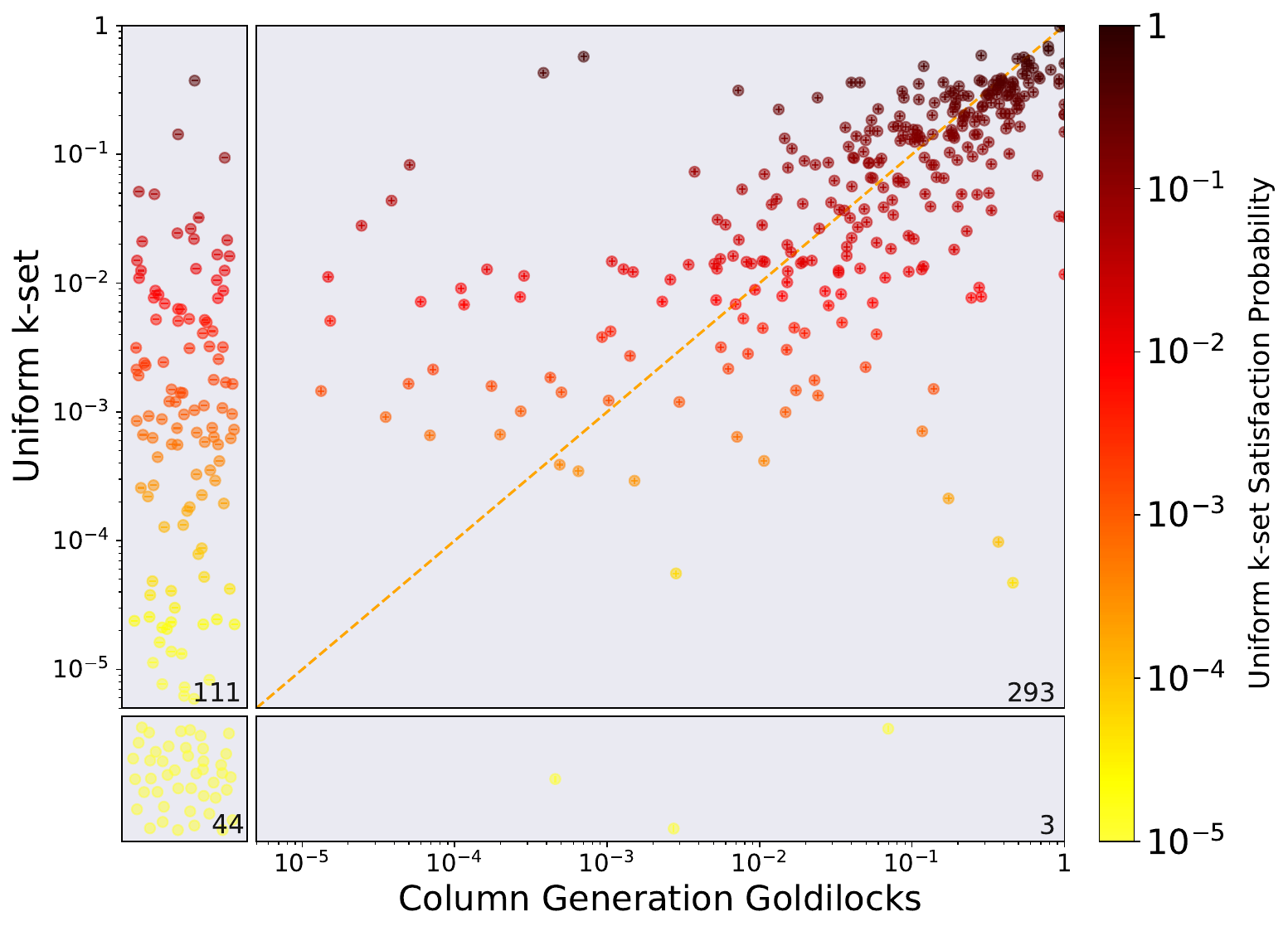}
  \end{subfigure}
\end{figure}

\begin{figure}[h!]
  \centering
  \begin{subfigure}[t]{0.49\textwidth}
    \centering
    \includegraphics[width=\linewidth]{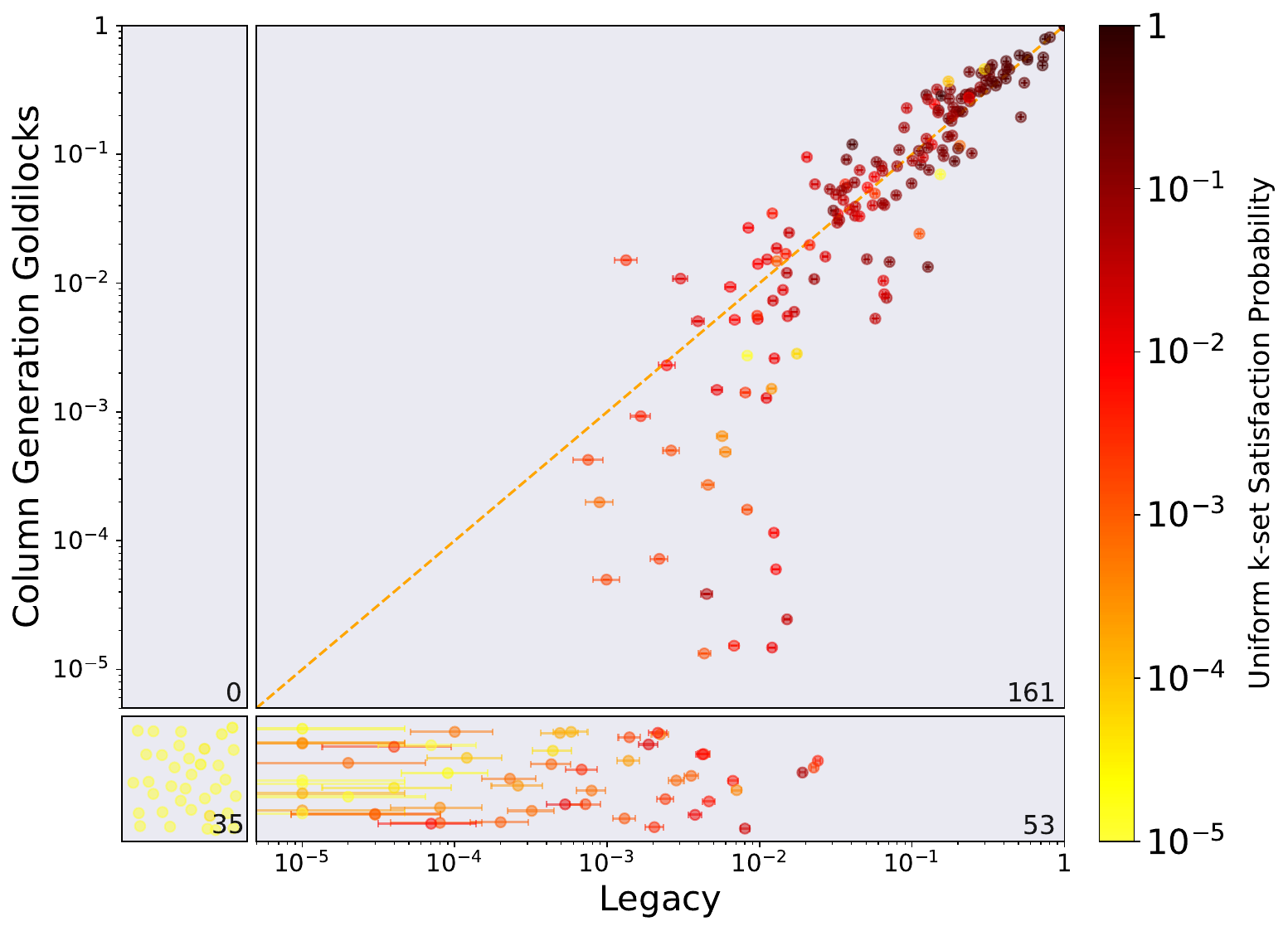}
  \end{subfigure}\hfill
  \begin{subfigure}[t]{0.49\textwidth}
    \centering
    \includegraphics[width=\linewidth]{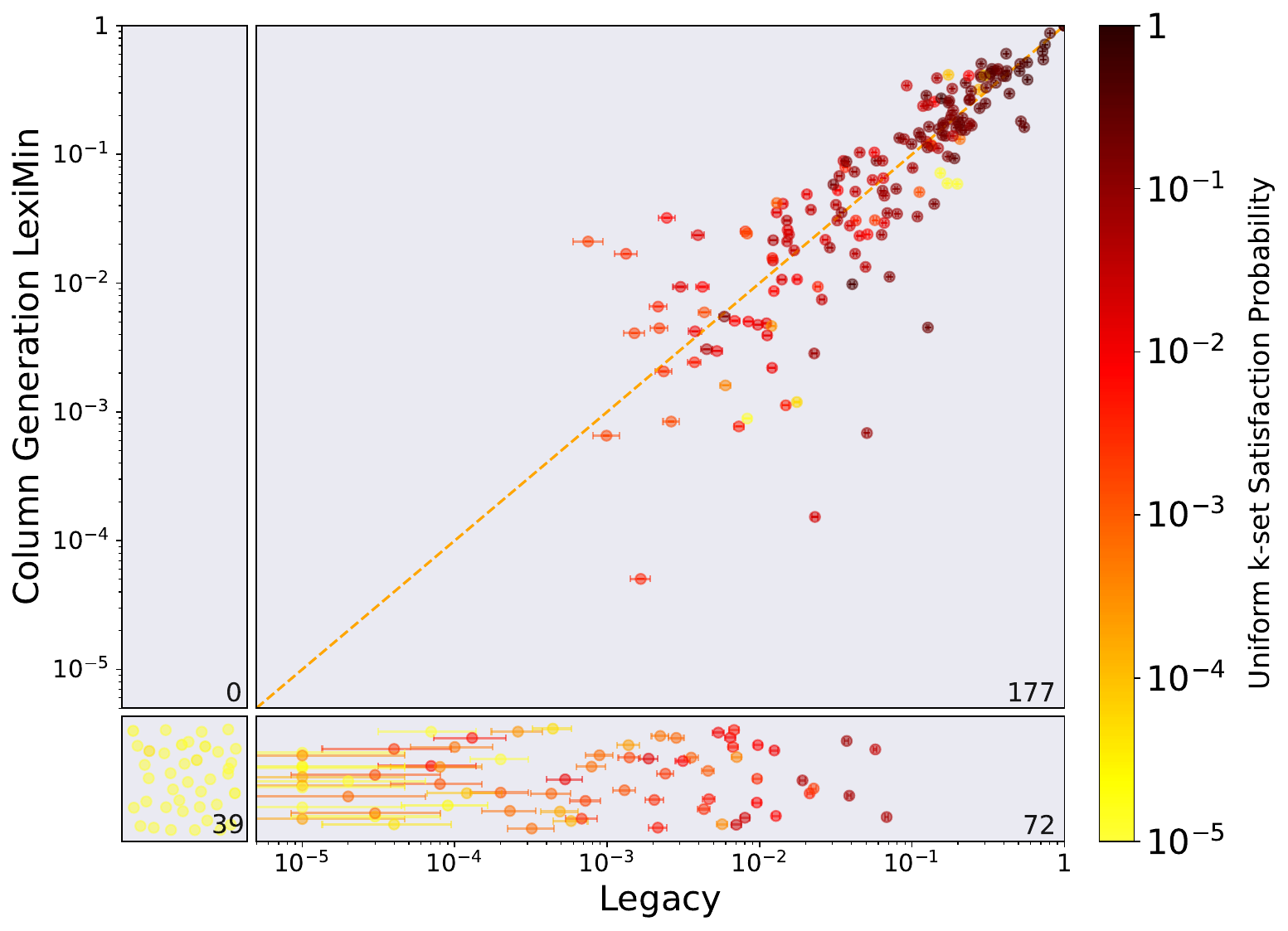}
  \end{subfigure}
\end{figure}

\begin{figure}[h!]
  \centering
  \begin{subfigure}[t]{0.49\textwidth}
    \centering
    \includegraphics[width=\linewidth]{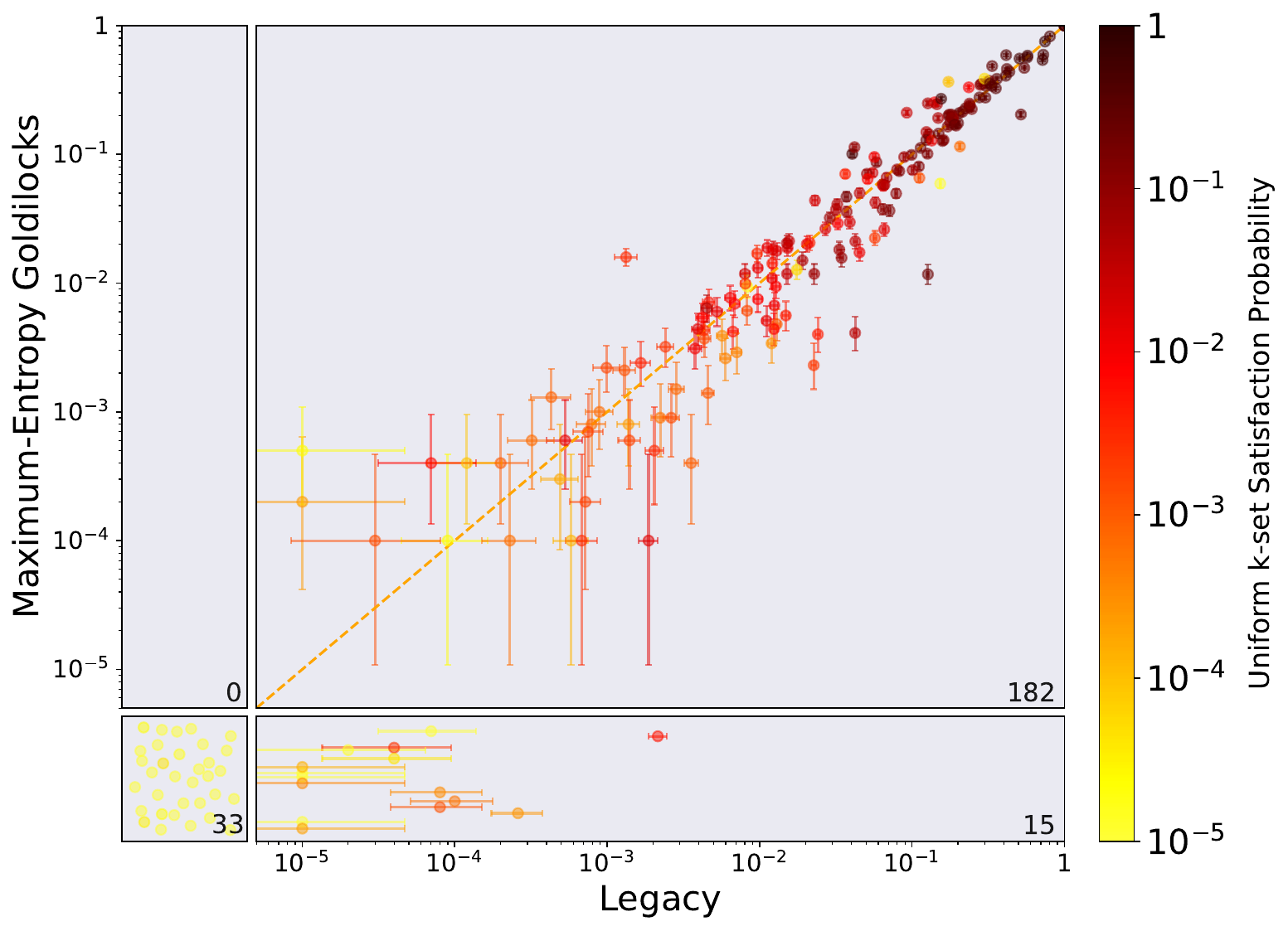}
  \end{subfigure}\hfill
  \begin{subfigure}[t]{0.49\textwidth}
    \centering
    \includegraphics[width=\linewidth]{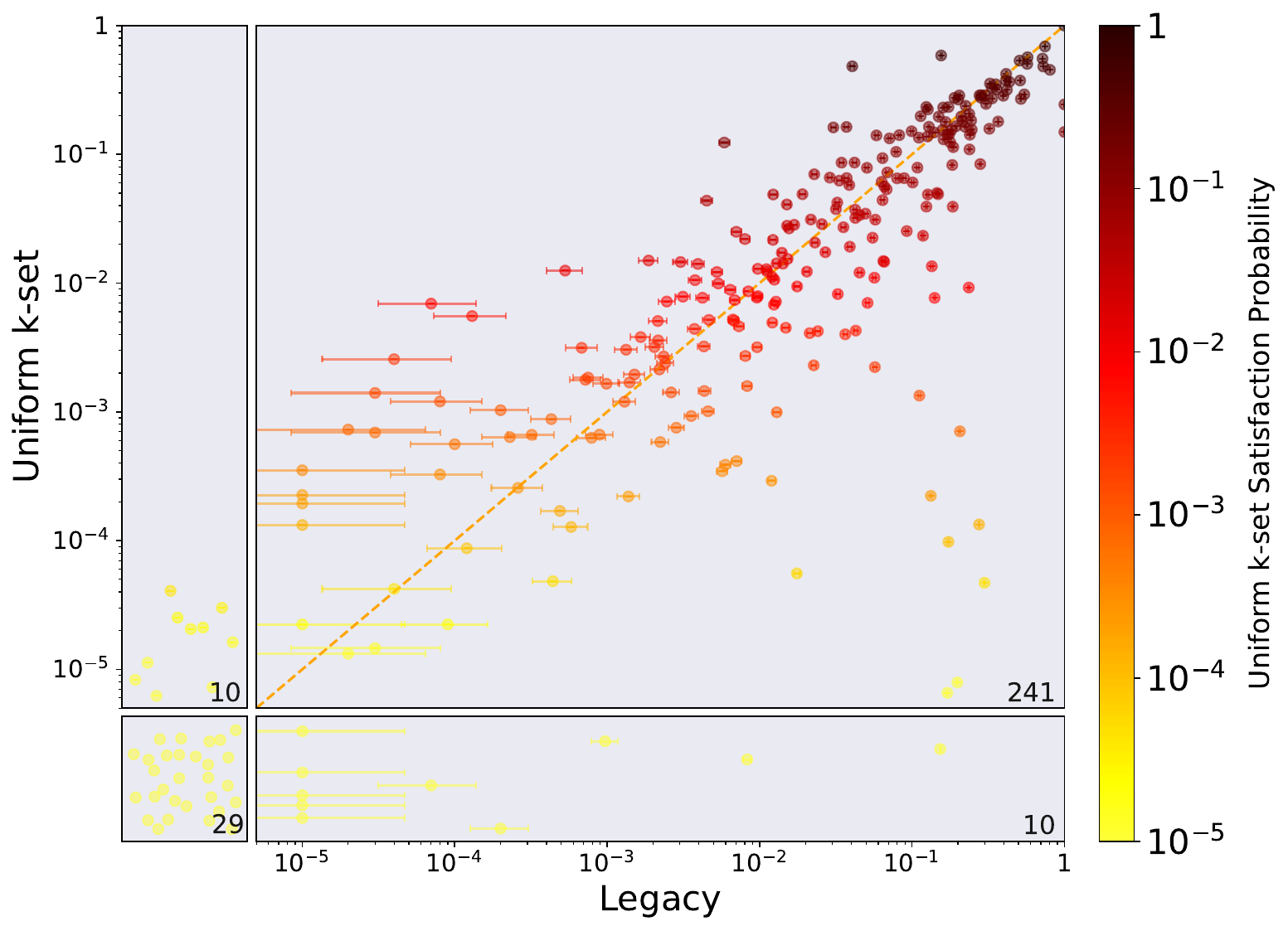}
  \end{subfigure}
\end{figure}

\begin{figure}[h!]
  \centering
  \begin{subfigure}[t]{0.49\textwidth}
    \centering
    \includegraphics[width=\linewidth]{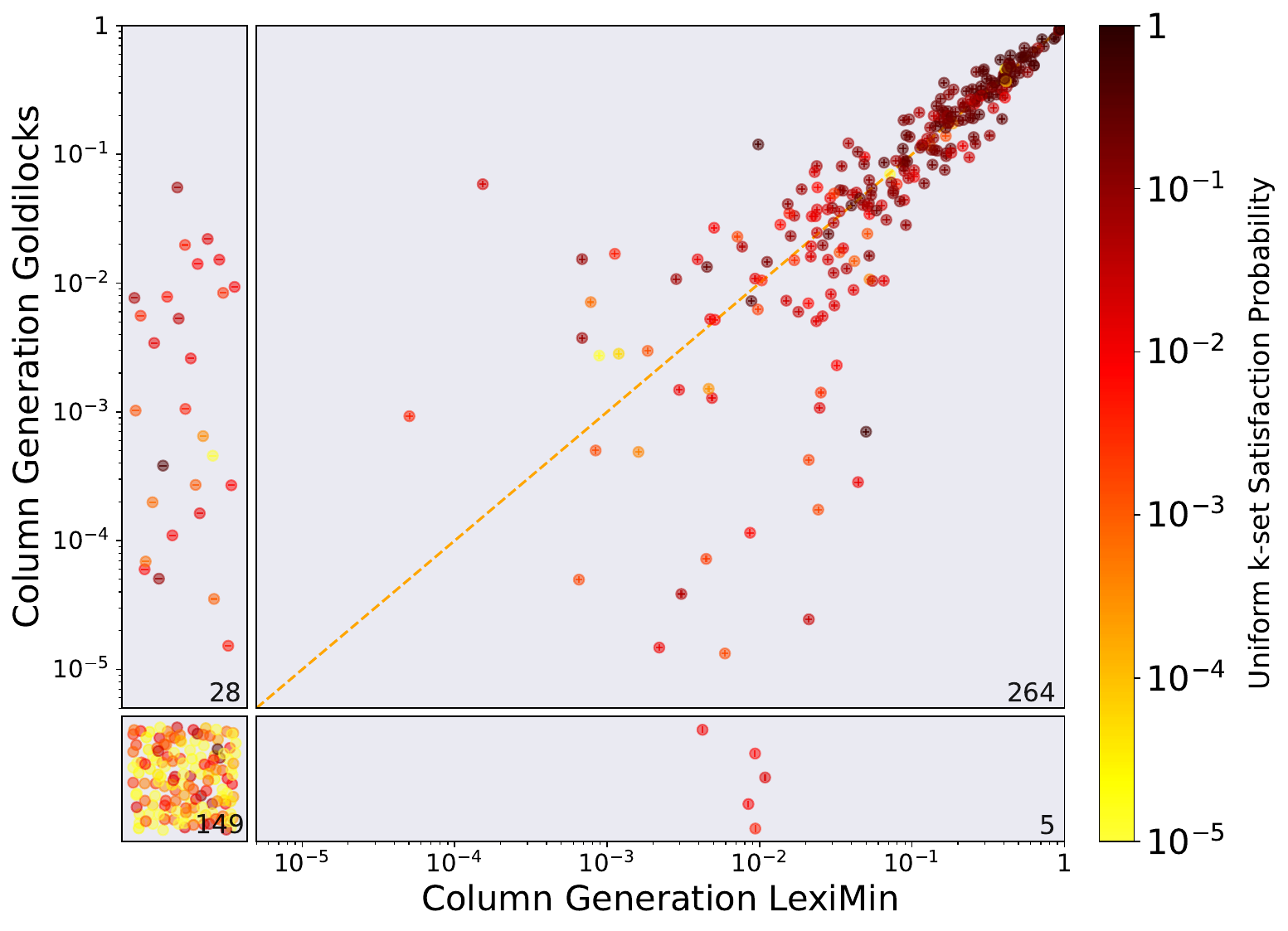}
  \end{subfigure}\hfill
  \begin{subfigure}[t]{0.49\textwidth}
    \centering
    \includegraphics[width=\linewidth]{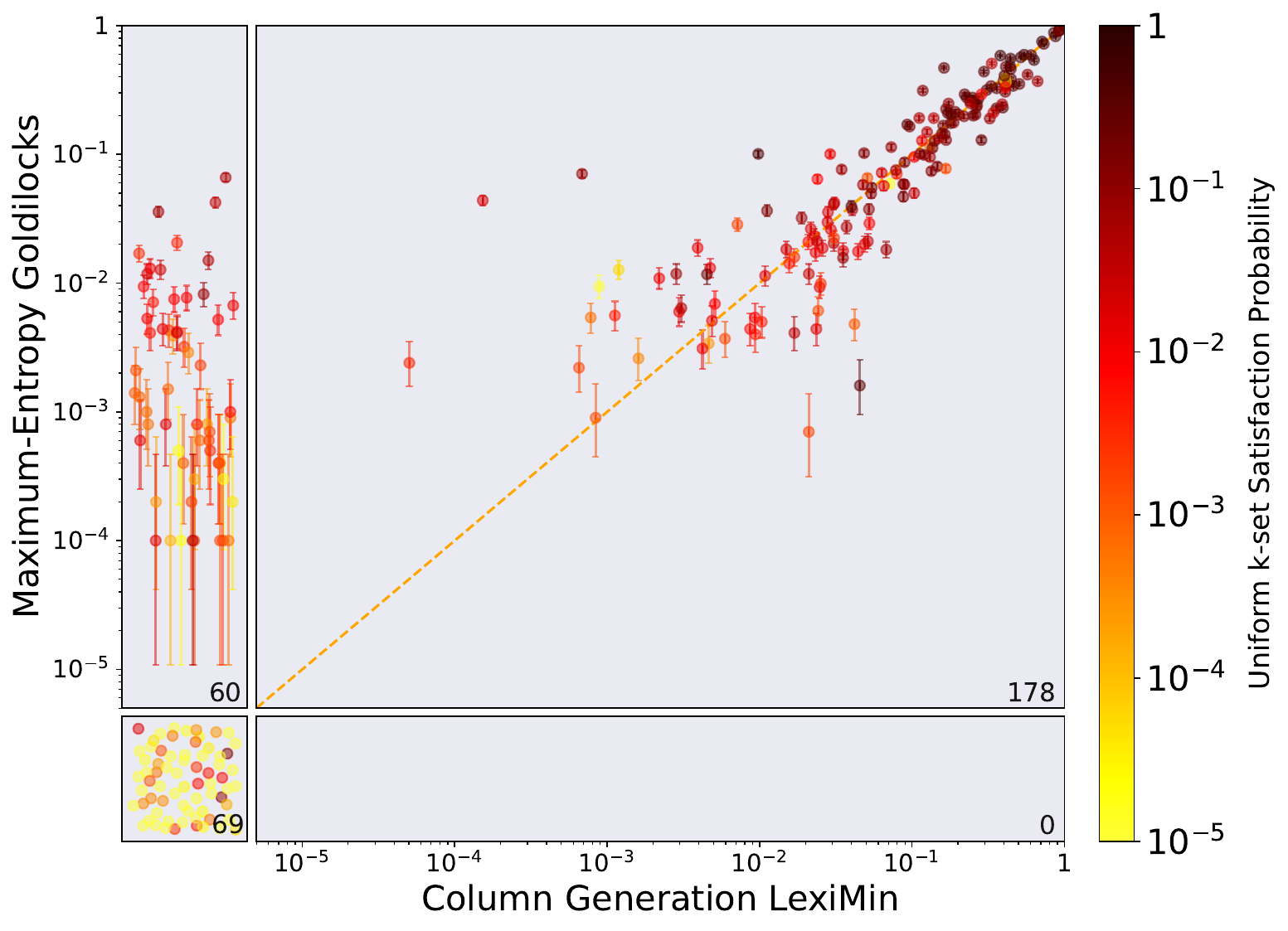}
  \end{subfigure}
\end{figure}

\begin{figure}[h!]
  \centering
  \begin{subfigure}[t]{0.49\textwidth}
    \centering
    \includegraphics[width=\linewidth]{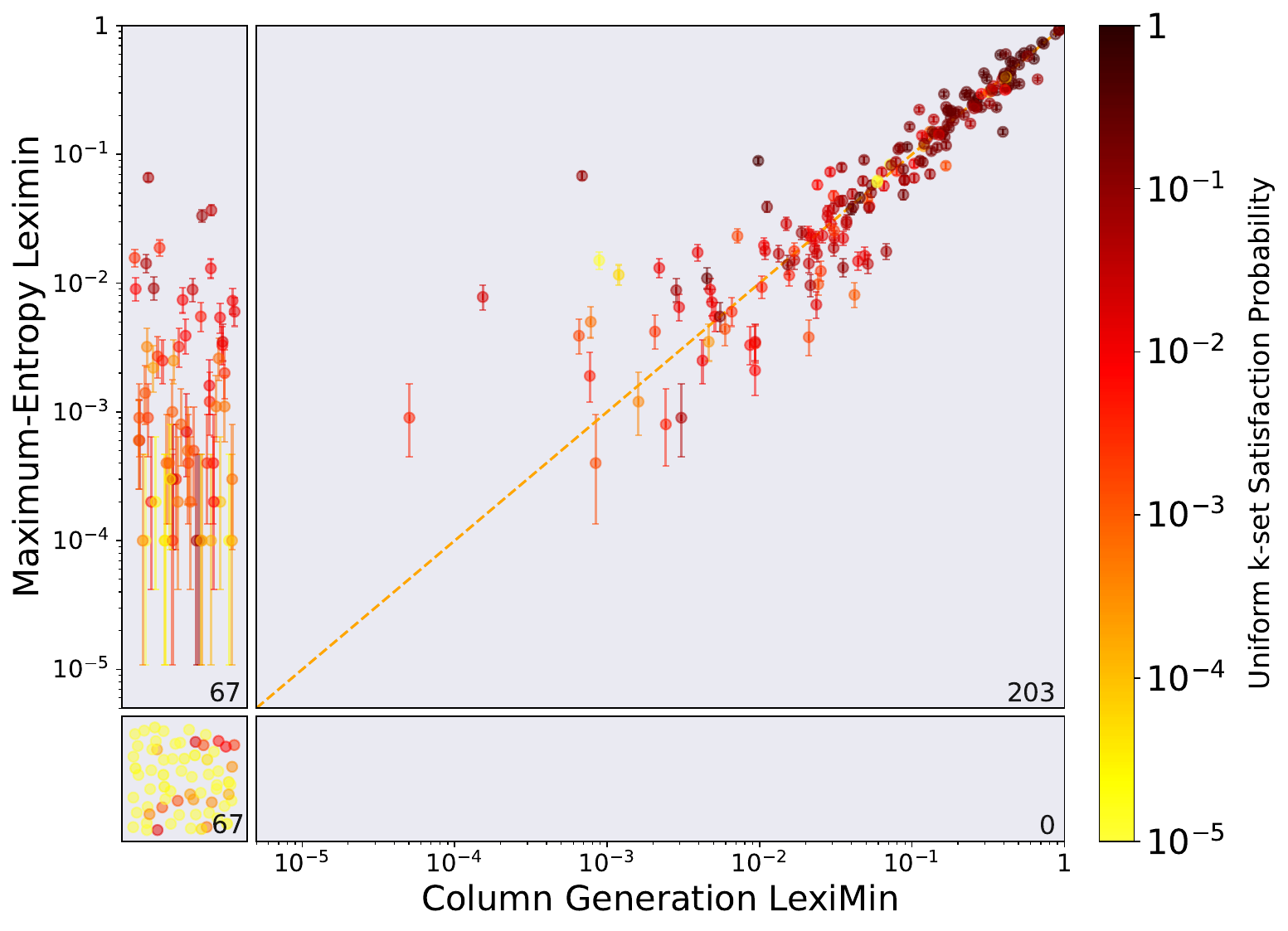}
  \end{subfigure}\hfill
  \begin{subfigure}[t]{0.49\textwidth}
    \centering
    \includegraphics[width=\linewidth]{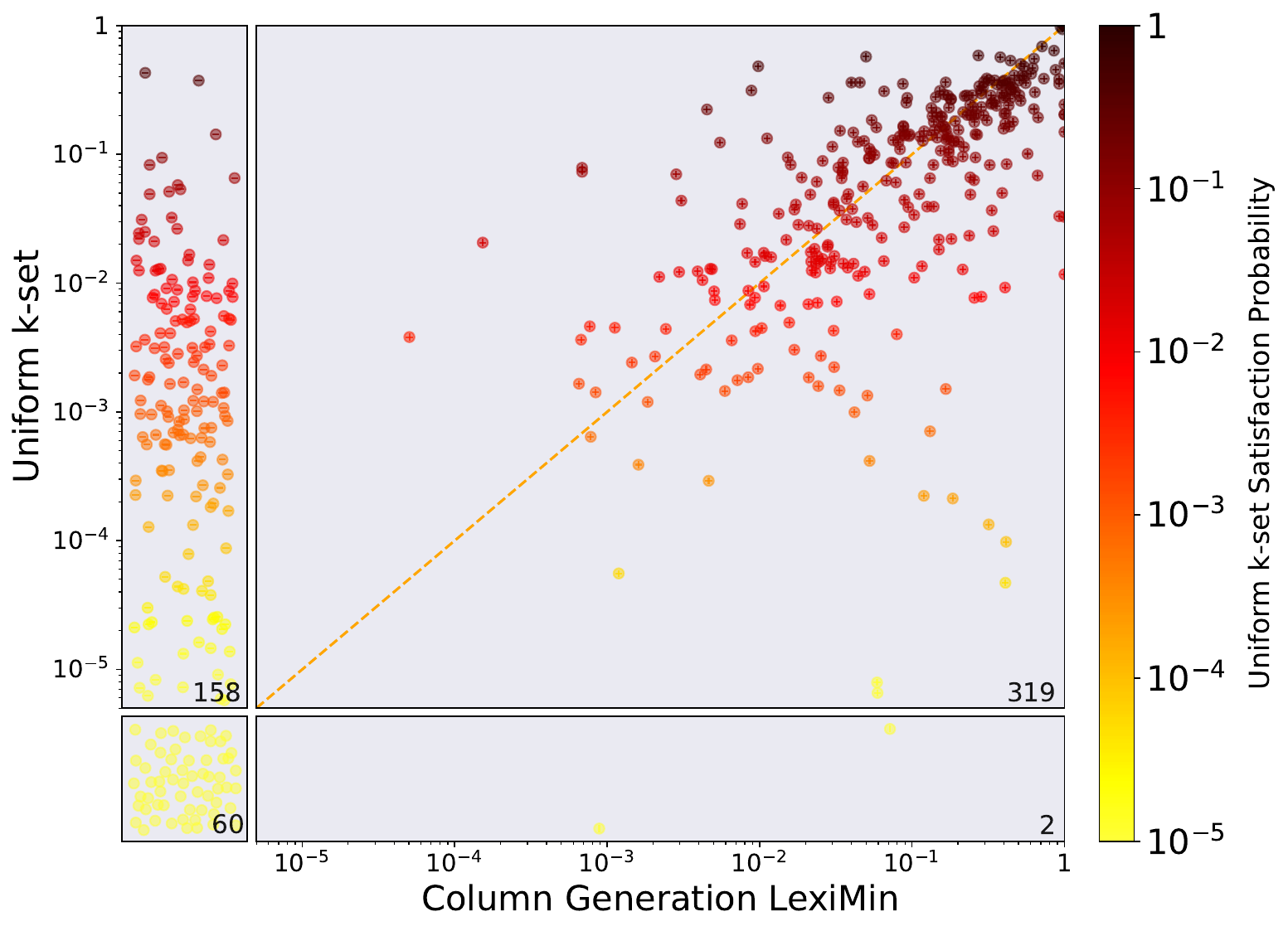}
  \end{subfigure}
\end{figure}

\begin{figure}[h!]
  \centering
  \begin{subfigure}[t]{0.49\textwidth}
    \centering
    \includegraphics[width=\linewidth]{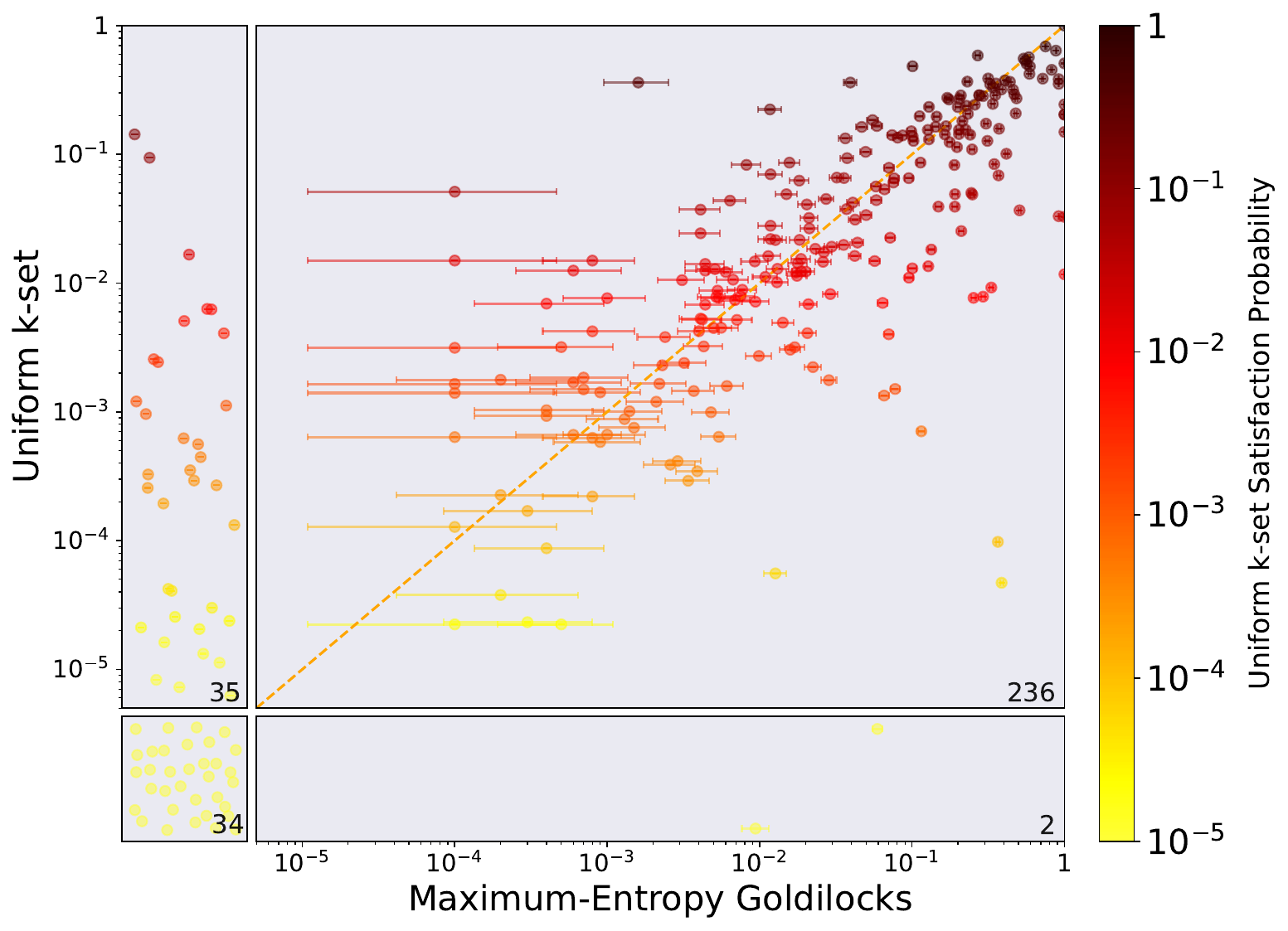}
  \end{subfigure}\hfill
  \begin{subfigure}[t]{0.49\textwidth}
    \centering
    \includegraphics[width=\linewidth]{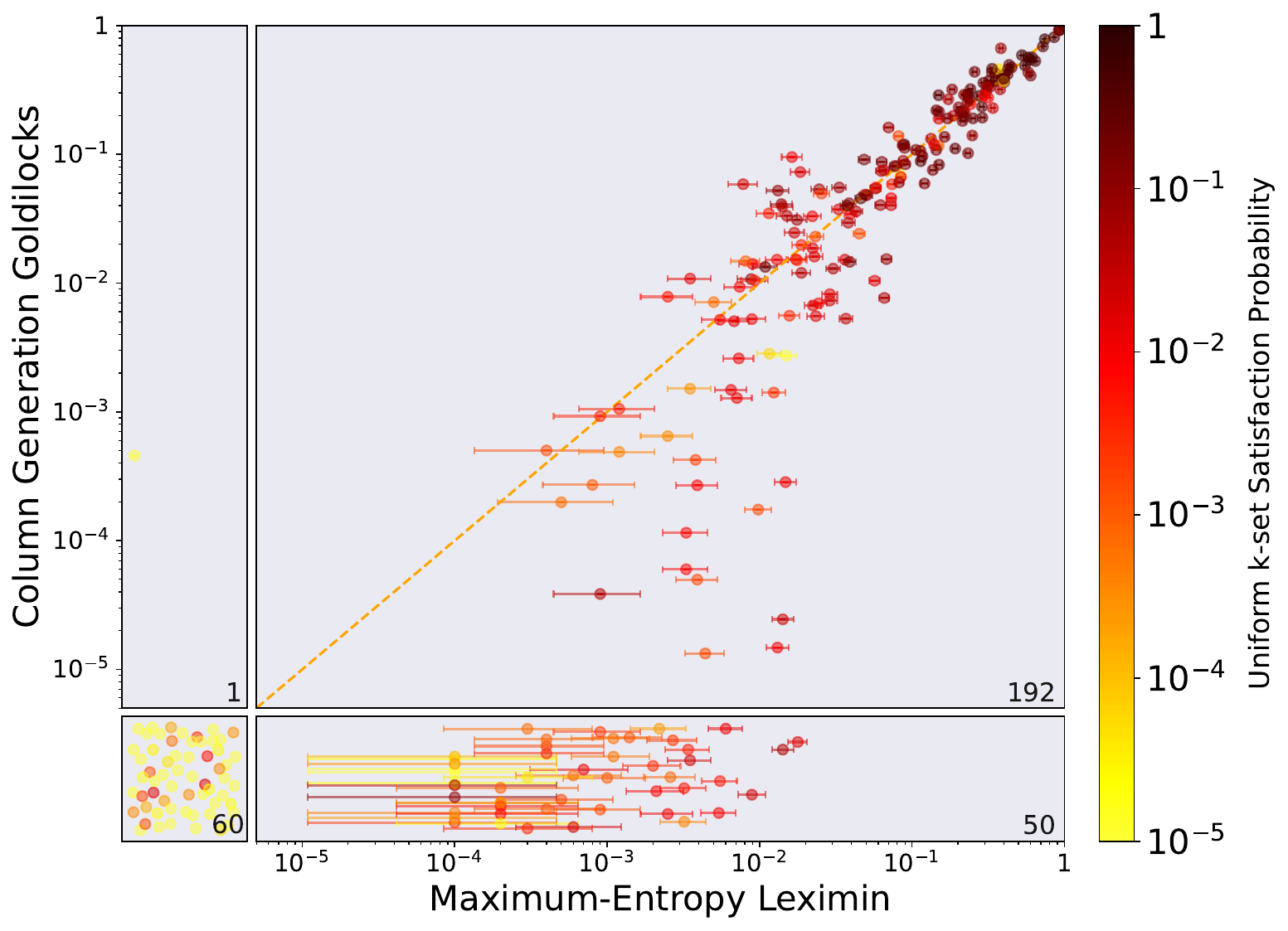}
  \end{subfigure}
\end{figure}

\begin{figure}[h!]
  \centering
  \begin{subfigure}[t]{0.49\textwidth}
    \centering
    \includegraphics[width=\linewidth]{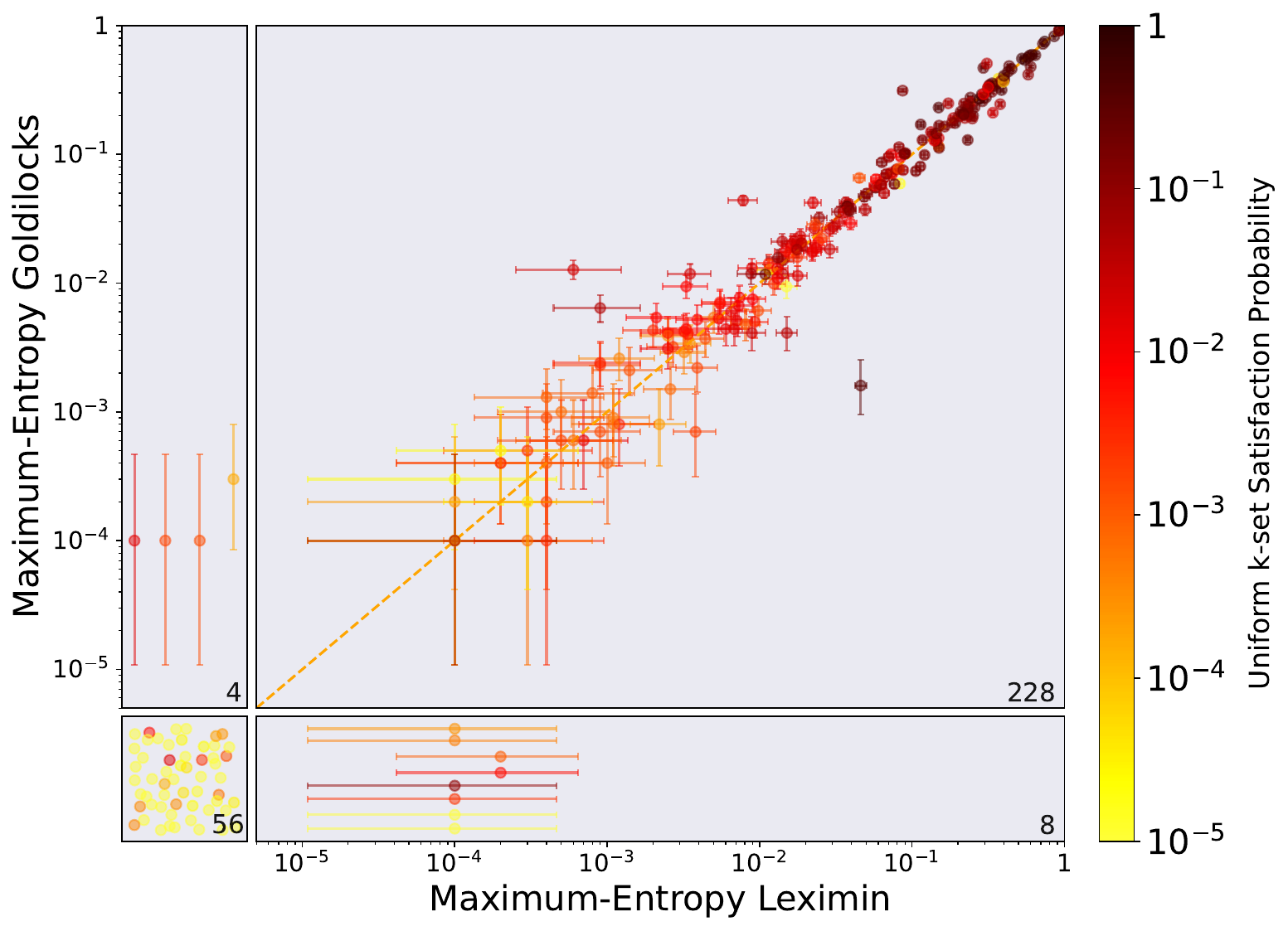}
  \end{subfigure}\hfill
  \begin{subfigure}[t]{0.49\textwidth}
    \centering
    \includegraphics[width=\linewidth]{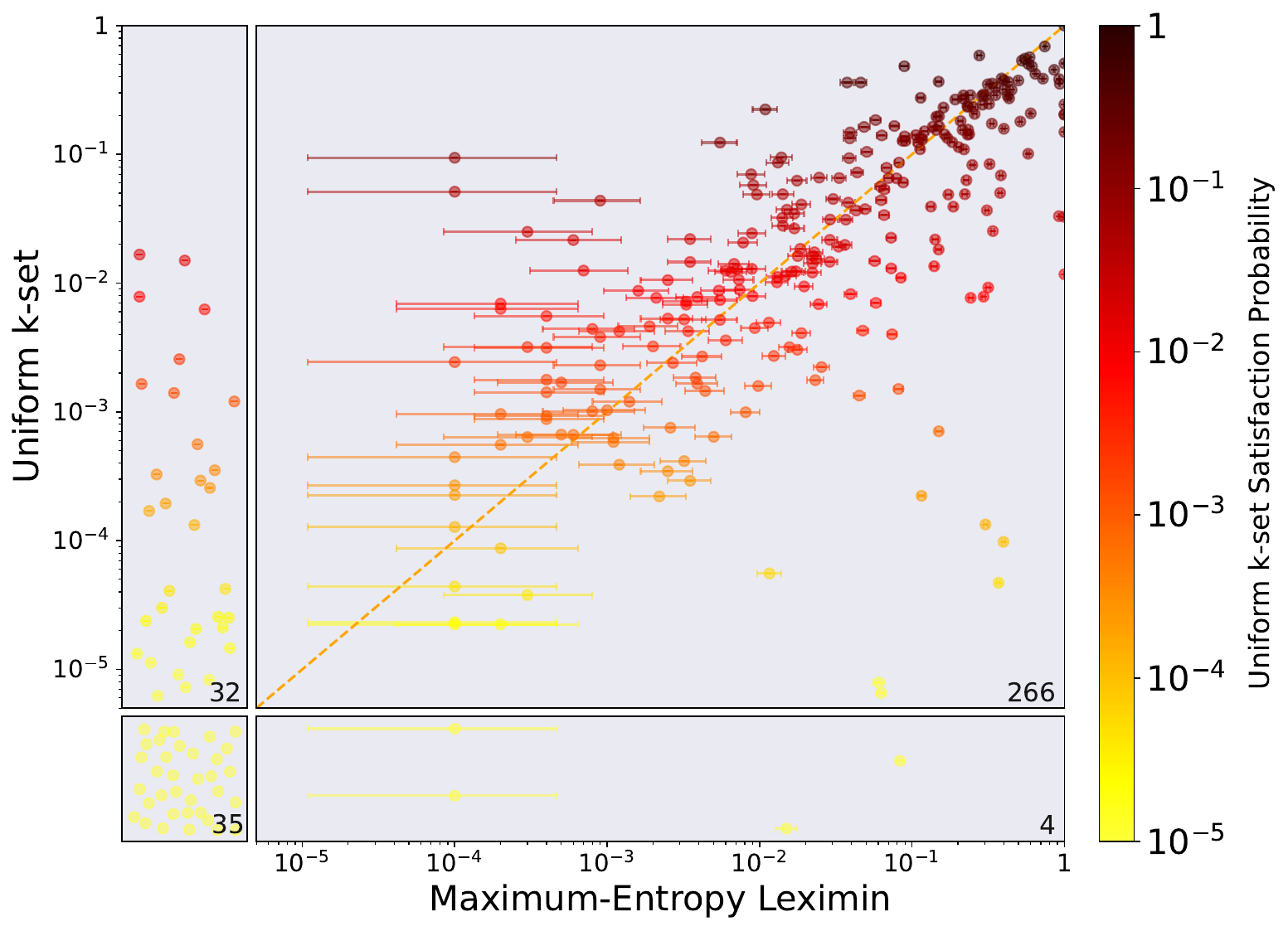}
  \end{subfigure}
\end{figure}

\begin{figure}[h!]
  \centering
  \begin{subfigure}[t]{0.49\textwidth}
    \centering
    \includegraphics[width=\linewidth]{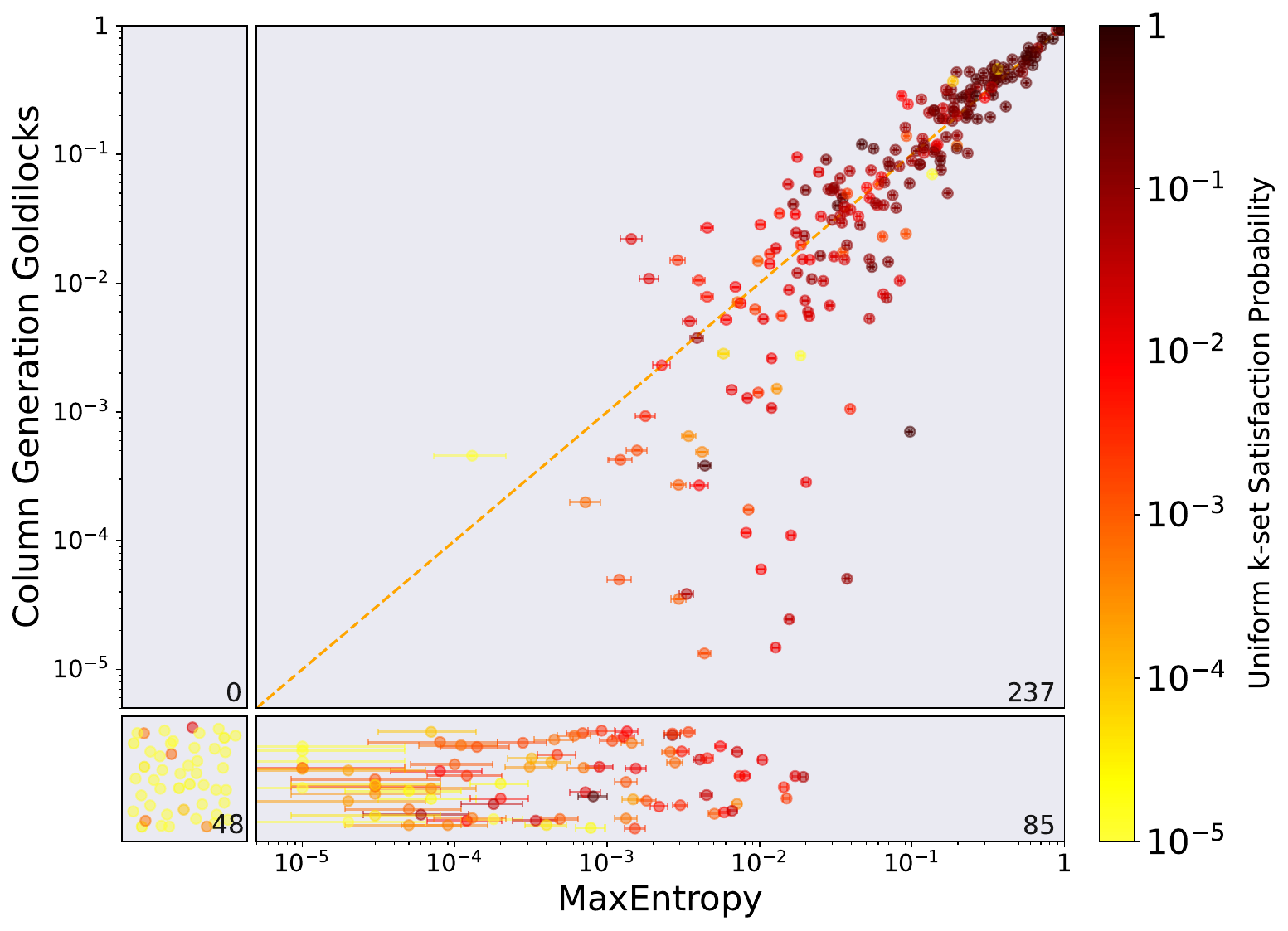}
  \end{subfigure}\hfill
  \begin{subfigure}[t]{0.49\textwidth}
    \centering
    \includegraphics[width=\linewidth]{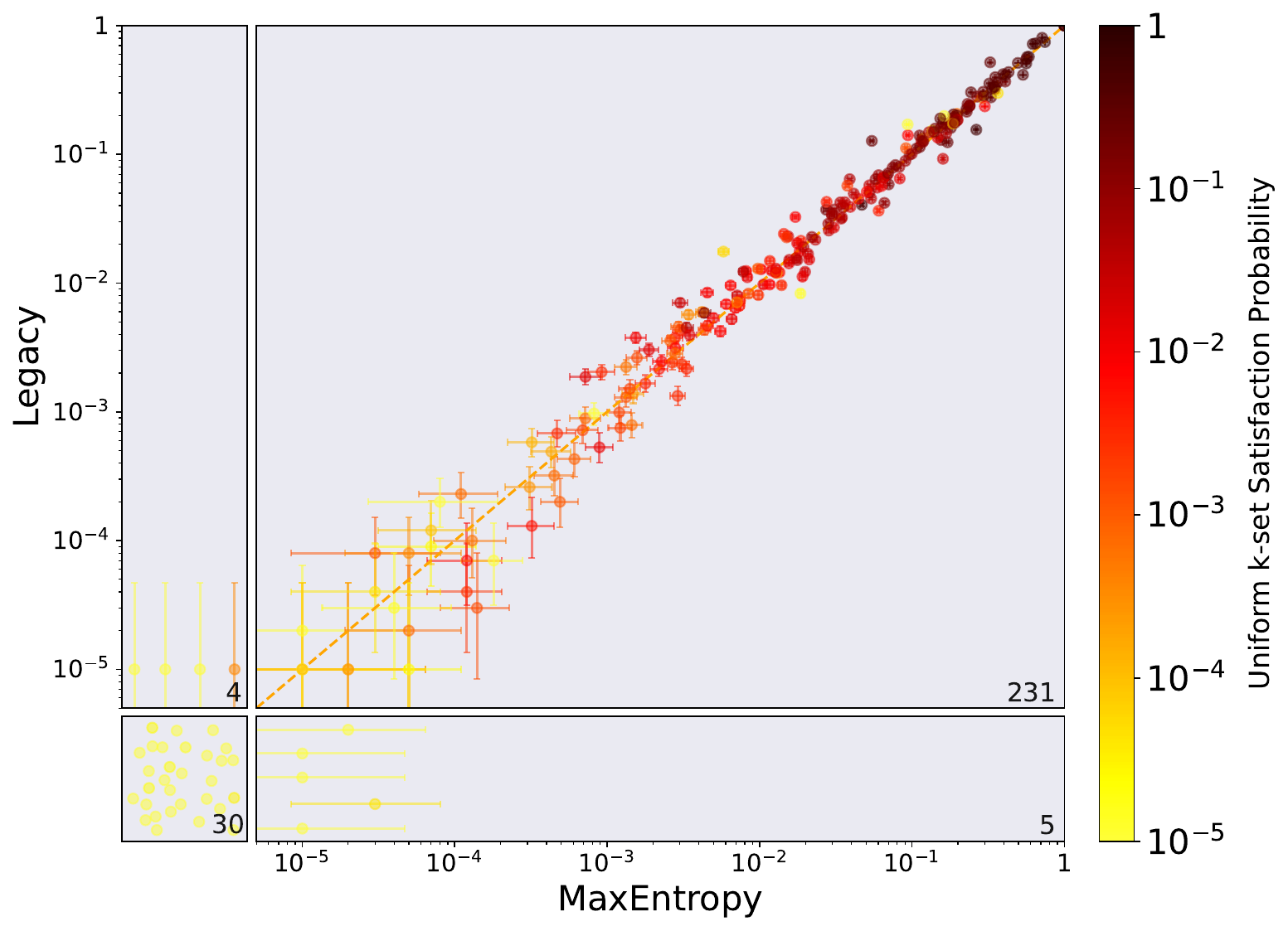}
  \end{subfigure}
\end{figure}

\begin{figure}[h!]
  \centering
  \begin{subfigure}[t]{0.49\textwidth}
    \centering
    \includegraphics[width=\linewidth]{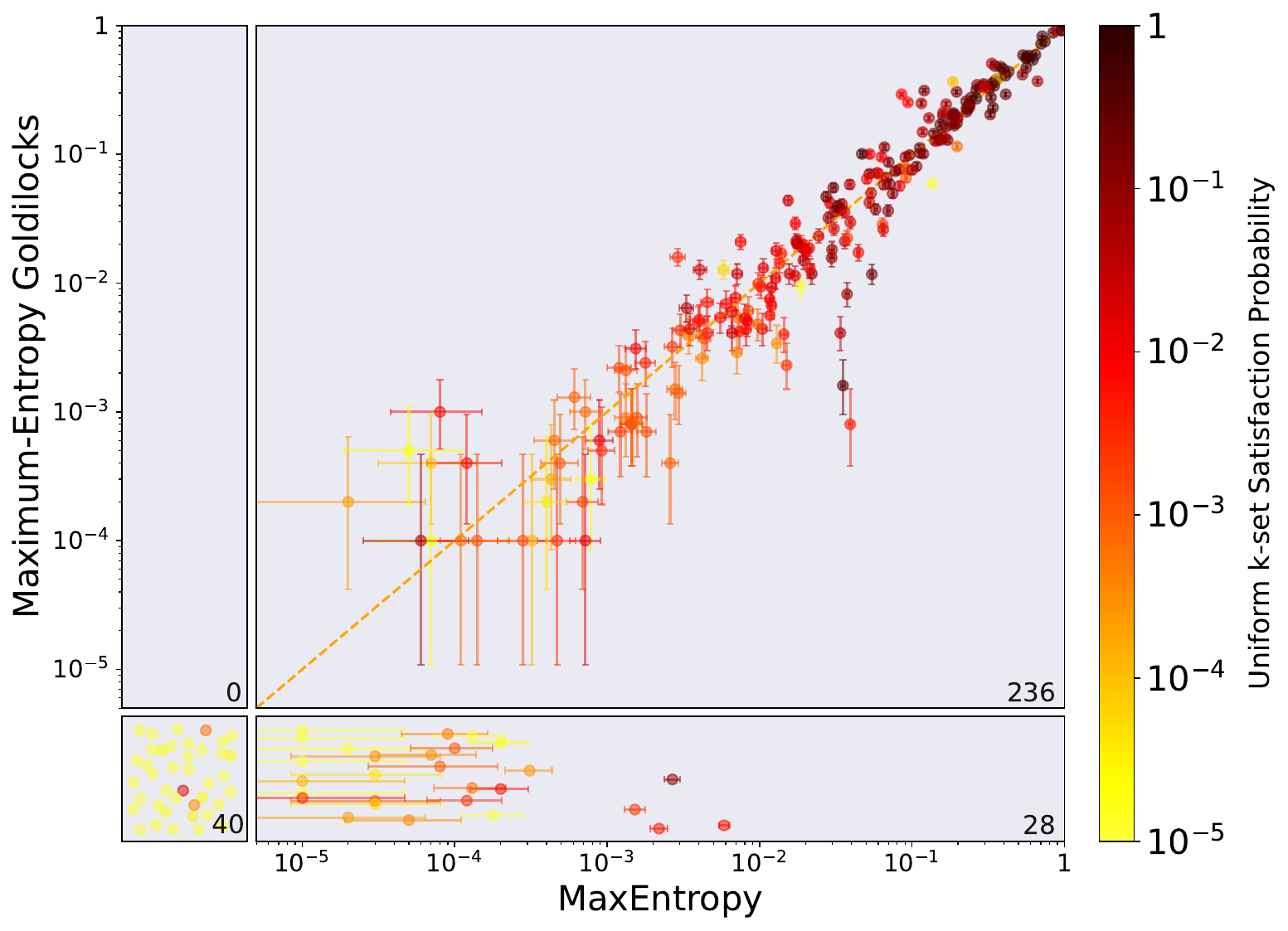}
  \end{subfigure}\hfill
  \begin{subfigure}[t]{0.49\textwidth}
    \centering
    \includegraphics[width=\linewidth]{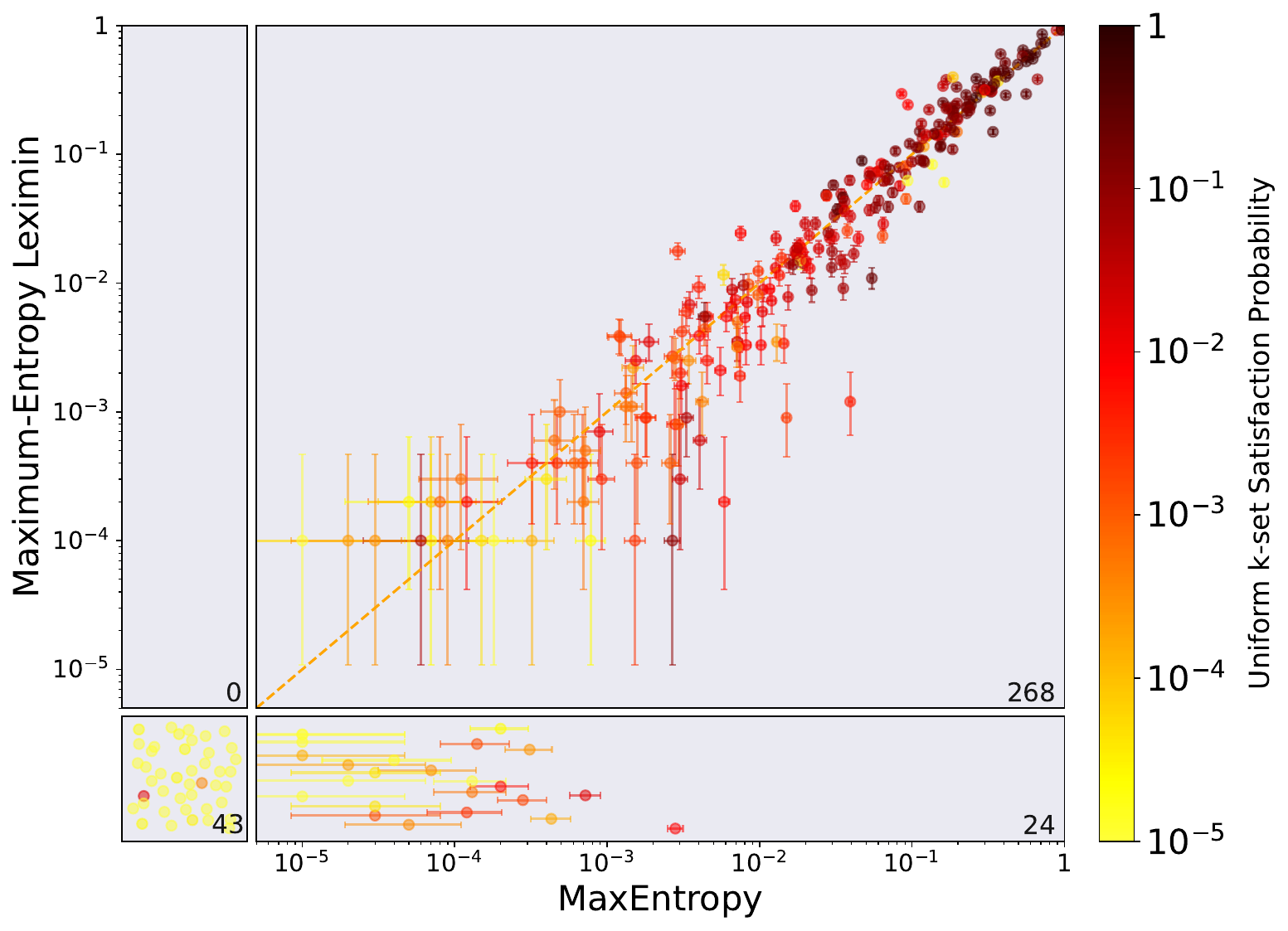}
  \end{subfigure}
\end{figure}

\begin{figure}[h!]
  \centering
  \begin{subfigure}[t]{0.49\textwidth}
    \centering
    \includegraphics[width=\linewidth]{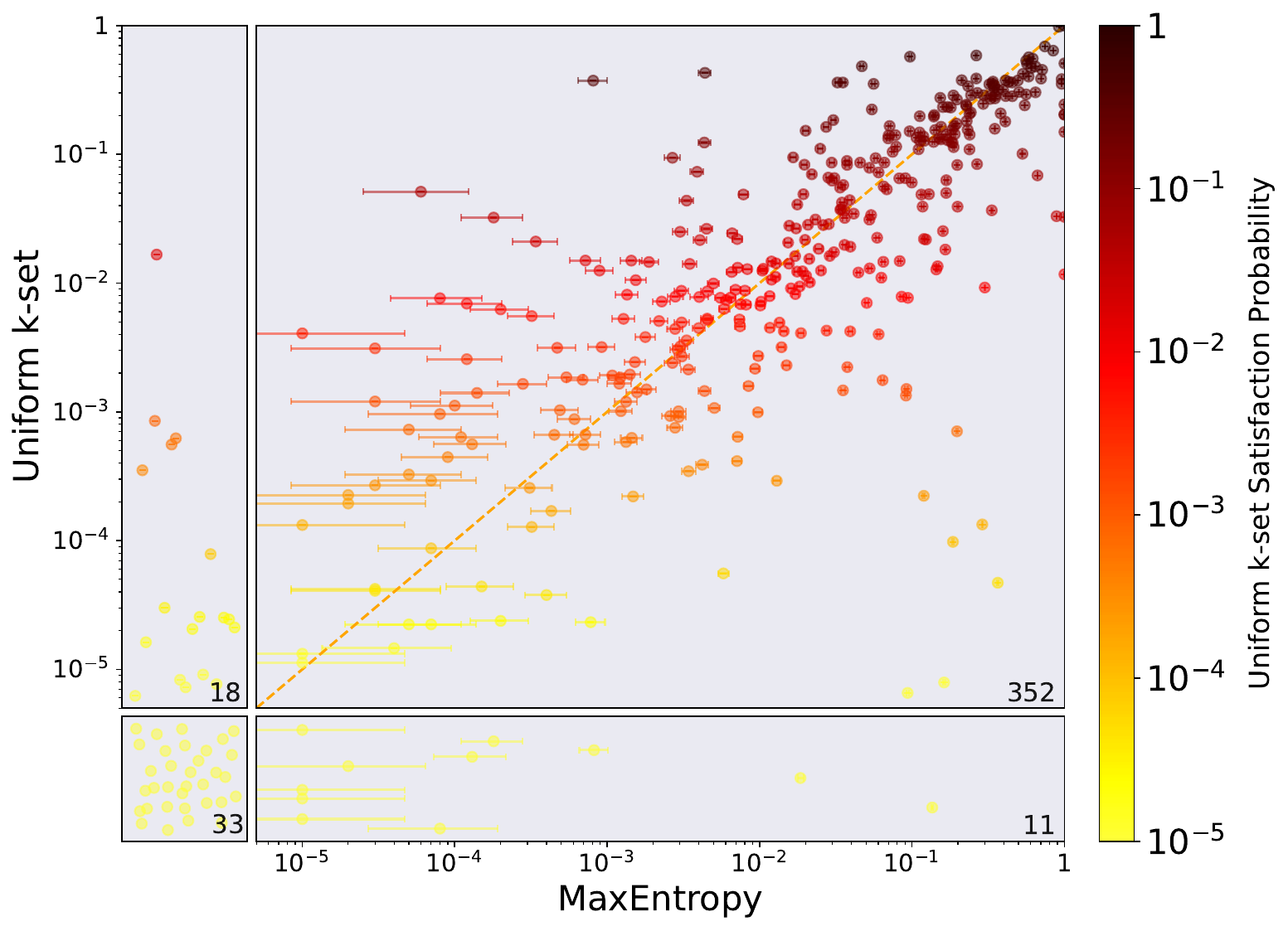}
  \end{subfigure}
\end{figure}

\end{document}